\documentclass[10pt]{article}
\usepackage[utf8]{inputenc}
\usepackage[margin=1in]{geometry}

\usepackage[ruled,vlined]{algorithm2e}

\usepackage[
    backend=biber,
    style=numeric-comp,
    sortcites,
    sorting=none,
    giveninits=true,
    natbib,
    hyperref=false,
    maxbibnames=99,
    doi=false,isbn=false,url=false,eprint=false
]{biblatex}
\makeatletter
\newcommand{\oset}[2]{%
  {\mathop{#2}\limits^{\vbox to 1.75\ex@{\kern-\tw@\ex@
   \hbox{\scriptsize #1}\vss}}}}
\makeatother

\usepackage{amsmath,amsthm,amssymb,bm}
\usepackage{graphicx}
\usepackage{tikz}
\usepackage{float}
\usetikzlibrary{decorations.pathreplacing}
\usepackage{hyperref}
\hypersetup{colorlinks,citecolor=blue,urlcolor=blue,linkcolor=blue}

\usepackage[font=small]{caption}
\usepackage{subcaption}
\usepackage{enumitem}
\usepackage{my_style}
\usepackage{placeins}

\usepackage{booktabs}
\usepackage{longtable}
\usepackage{array}
\usepackage{multirow}
\usepackage{wrapfig}
\usepackage{float}
\usepackage{colortbl}
\usepackage{pdflscape}
\usepackage{tabu}
\usepackage{threeparttable}
\usepackage{threeparttablex}
\usepackage[normalem]{ulem}
\usepackage{makecell}
\usepackage{xcolor}
\newcommand{\sign}{\operatorname{sign}}
\newcommand{\ici}{\operatorname{cst}}
\newcommand{\pici}{\operatorname{pcst}}
\newcommand{\inv}{\operatorname{cst}}
\newcommand{\pinv}{\operatorname{pcst}}

\newcommand{\indep}{\;\rotatebox[origin=c]{90}{$\models$}\;}

\newcommand{\px}{p}
\newcommand{\pz}{p_{\mathrm{z}}}
\newcommand{\pc}{p_{\mathrm{c}}}
\newcommand{\Pa}[1]{\mathrm{Pa}(#1)}
\newcommand{\Pax}[1]{\mathrm{Pa}_{\mathrm{x}}(#1)}
\newcommand{\Pac}[1]{\mathrm{Pa}_{\mathrm{c}}(#1)}
\renewcommand{\P}[1]{\mathbb{P}\left[#1\right]}
\newcommand{\E}[1]{\mathbb{E}\left[#1\right]}

\usepackage{bbm}
\newcommand{\I}[1]{\mathbbm{1}\left[#1\right]}

\usepackage{calc}
\usepackage{accents}

\newcommand{\vardbtilde}[1]{\tilde{\raisebox{0pt}[0.85\height]{$\tilde{#1}$}}}

\usepackage{attrib}
\usepackage{mathtools}

\newcommand{\defeq}{\vcentcolon=}
\newcommand{\defeqi}{=\vcentcolon}

\newtheoremstyle{mystyle}%                % Name
  {}%                                     % Space above
  {}%                                     % Space below
  {\itshape}%                                     % Body font
  {}%                                     % Indent amount
  {\bfseries}%                            % Theorem head font
  {.}%                                    % Punctuation after theorem head
  { }%                                    % Space after theorem head, ' ', or \newline
  {}%                                     % Theorem head spec (can be left empty, meaning `normal')

\theoremstyle{mystyle}

\newtheorem{lemma}{Lemma}
\newtheorem{definition}{Definition}

\usepackage{thm-restate}

\makeatletter
\def\@fnsymbol#1{\ensuremath{\ifcase#1\or *\or 1 \or 2 \or 3 \or 4 \or 5 \else\@ctrerr\fi}}
\makeatother

\providecommand{\keywords}[1]
{
  \small
  \textbf{\textit{Keywords---}} #1
}

\title{Searching for consistent associations with a\\multi-environment knockoff filter}

\author{Shuangning Li\footnote{Equal contribution.} \thanks{Department of Statistics, Stanford University, USA.}, Matteo Sesia\footnotemark[1] \thanks{Department of Data Sciences and Operations, University of Southern California, USA.}, Yaniv Romano\thanks{Departments of Electrical Engineering and of Computer Science, Technion, Israel.}, Emmanuel Cand{\`e}s\thanks{Departments of Statistics and of Mathematics, Stanford University, USA.}, and Chiara Sabatti\thanks{Departments of Statistics and of Biomedical Data Sciences, Stanford University, USA.}}
\date{}

\begin{document}

\maketitle

\begin{abstract}
This paper develops a method based on model-X knockoffs to find conditional associations that are consistent across diverse environments, controlling the false discovery rate.
The motivation for this problem is that large data sets may contain numerous associations that are statistically significant and yet misleading, as they are induced by confounders or sampling imperfections. However, associations consistently replicated under different conditions may be more interesting.
In fact, consistency sometimes provably leads to valid causal inferences even if conditional associations do not.
While the proposed method is flexible and can be deployed in a wide range of applications, this paper highlights its relevance to genome-wide association studies, in which consistency across populations with diverse ancestries mitigates confounding due to unmeasured variants.
The effectiveness of this approach is demonstrated by simulations and applications to the UK Biobank data.
\end{abstract}


\section{Introduction}

A critical goal of statistics is to discover which variables, among the many measured in big-data applications, are meaningfully associated with an outcome of interest. The word ``association'' may have different connotations, ranging from {\em marginal} association, a tendency of two variables to vary together, to {\em causal} association, a relation ensuring interventions on one variable affect another.
For example, a genome-wide association study may ascertain that some genetic variants occur more frequently among individuals with diabetes. An explanation for this marginal association may be that the discovered variants are biologically irrelevant but are shared, because of common ancestry, by a sub-population that happens to follow a less healthy diet~\cite{devlin1999genomic}.
Of course, it would be more actionable to identify variants involved in biological processes which, if modified by a drug, could influence the disease.
Marginal associations are the simplest to recognize but also the least informative, while causal associations are more elusive, especially with high-dimensional observational data, although they better lend themselves to scientific interpretations.

Between marginal and causal associations one finds {\em conditional} association: the tendency of two variables to vary together when other quantities are fixed.
Conditional testing has been traditionally tackled through parametric models; however, these require strong assumptions, which are not always justified.
A new ``model-X'' strategy was proposed by~\cite{candes2018panning}, making no assumptions about the conditional distribution of the outcome and approximating instead the joint distribution of the predictors. This has led to two methods, {\em knockoffs}~\cite{barber2015controlling,candes2018panning} and the {\em conditional randomization test}~\cite{candes2018panning}, which can harness the power of any machine learning algorithm while providing type-I error guarantees in finite-samples.
The model-X assumptions are particularly well-suited to genome-wide associations studies because reliable prior knowledge is available about the joint distribution of the explanatory variables~\cite{sesia2018,sesia2020multi,sesia2020controlling}, but the framework is quite robust to approximations~\cite{candes2018panning,romano2019} and thus broadly applicable.

Despite the above advantages, conditional testing is not fully satisfactory for at least three reasons.
First, it does not account for unmeasured {\em confounders}: missing variables which, if conditioned upon, would explain away the association~\cite{pearl2009causality}.
For instance, one may discover that an irrelevant genetic marker is associated with a disease simply because it is physically close to, and thus inherited alongside, an unobserved causal variant~\cite{pritchard2001linkage}.
Second, some data sets may be collected from a population that does not match exactly the target one, either because of accidental sampling bias~\cite{heckman1979sample}, or for convenience~\cite{harford2014}. Clearly, any conditional associations may be misleading in that case.
Third, conditional testing methods typically assume individual samples are independent of one another, and this can lead to spurious associations in the presence of unexpected dependencies, such as network effects~\cite{lee2020network}.

This paper extends the methodology of knockoffs to mitigate the three aforementioned limitations of conditional testing by analyzing data from many {\em environments}. The word ``environment'' is employed loosely here, referring to specific sub-populations, experimental settings, or data collection strategies depending on the context.
Our work is motivated by the conjecture that the most informative associations are those which can be consistently reproduced under different environments, because these tend to enable more generally reliable predictions and may even reflect scientifically illuminating causal relations. This is an old idea, dating back at least to Hume~\cite{hume1739}.
\begin{quotation}
``There is no phaenomenon in nature, but what is compounded and modified by so many different circumstances, that in order to arrive at the decisive point, we must carefully separate whatever is superfluous, and enquire by new experiments, if every particular circumstance of the first experiment was essential to it. These new experiments are liable to a discussion of the same kind; so that the utmost constancy is required to make us persevere in our enquiry, and the utmost sagacity to choose the right way among so many that present themselves.''
\attrib{\citet[]{hume1739}}
\end{quotation}
We will translate the above logic into a practical method for the analysis of high-dimensional data, provably controlling the false discovery rate~\cite{benjamini1995} for hypotheses of consistent conditional association.
The proposed solution utilizes knockoffs because these are computationally efficient in very high dimensions and allow controlling the false discovery rate even if the predictors have strong dependencies~\cite{candes2018panning}. However, an alternative approach based on the conditional randomization test~\cite{candes2018panning} would be easy to implement; see Appendix~\ref{sec:app-crt} and Figure~\ref{fig:CRT} therein.

The outline is the following. Section~\ref{sec:prob-statement} states the problem.
Section~\ref{sec:ci-to-causal} proves that testing our hypotheses sometimes leads to causal inferences.
Section~\ref{sec:methods} develops our methods to test the aforementioned hypotheses using knockoffs. 
Section~\ref{sec:gwas} dives into genome-wide association studies~\cite{sesia2018}, explaining how consistency across sub-populations mitigates confounding due to unobserved variants.
Section~\ref{section:experiments} validates empirically our method with simulations. Section~\ref{sec:ukb} applies it to analyze the genetic determinants of several phenotypes in the UK Biobank resource~\cite{bycroft2018} using 600k genotypes from 500k individuals.
Section~\ref{sec:discussion} concludes by discussing some opportunities for future research.

\subsection*{Related work}

This paper was inspired by~\cite{peters2016causal,heinze2018invariant}, which advanced invariance across environments as a framework for causal inference. 
Departing from their work, however, we do not assume the existence of an invariant model with homogeneous effects; indeed, our approach is fully non-parametric  and leads to meaningful inference without reference to a specific  causal model.
Further, we seek different guarantees: the method in~\cite{peters2016causal} searches for a conservative confidence set of possible causal predictors, which may be desirable if the number of variables is small but does not scale well to high-dimensions, while ours controls the false discovery rate and can achieve high power even with hundreds of thousands of variables.
Notions of invariance similar to that of~\cite{peters2016causal} have also been utilized to improve predictive accuracy in the face of changes in the distribution of the explanatory variables~\cite{zhang2015multi,rojas2018invariant,arjovsky2020invariant,rothenhausler2021}, a phenomenon known also as {\em covariate shift}~\cite{shimodaira2000improving,sugiyama2007covariate}; although the problems are related, this paper concentrates on testing rather than prediction.

Our problem is related to {\em causal discovery}~\cite{spirtes1999algorithm,chickering2002optimal,koivisto2004exact,zhang2008completeness,pearl2009causality,glymour2019review}, which is the challenge of learning the graph describing the relations between all variables in a system, without pre-specifying an outcome of interest as we do;
this has also been extended to leverage invariance across environments~\cite{mooij2020joint}. Those methods can discover the direction of causal relations, as opposed to ours which solely tests conditional independence, but they require parametric assumptions and provide asymptotic rather than finite-sample guarantees.
There has also been interest in invariance within the feature selection literature~\cite{yu2019multi,ling2019bamb,wang2020towards}, though typically without seeking finite-sample inferences; see~\cite{yu2020causality} for a review.

Knockoffs were introduced by~\cite{barber2015controlling} and later extended by~\cite{candes2018panning} to the high-dimensional model-X framework. Subsequently, algorithms were developed to construct knockoffs for different distributions of predictors~\cite{sesia2018,gimenez2019knockoffs,romano2019,bates2020metropolized}, while others studied robustness to model misspecifications~\cite{barber2020robust} and power~\cite{liu2019power,katsevich2020theoretical,wang2020power,spector2020powerful}. Our work is orthogonal, as we extend the knockoff filter~\cite{barber2015controlling} to analyze data from many environments.
We focus on applications to genome-wide association data~\cite{sesia2018}, for which prior efforts addressed the problems of accounting for dependencies across the genotyped markers~\cite{sesia2020multi}, population structure~\cite{sesia2020controlling}, and even other unmeasured confounders~\cite{Bates2020}, but did not deal with missing variants. Knockoffs have also been deployed in other fields~\cite{shen2019false,chia2020interpretable,fan2020ipad,srinivasan2020compositional} for which our methodology may be helpful.

\section{Conditional associations that hold across environments} \label{sec:prob-statement}

Consider $E$ {\em environments}, or experimental settings, from which one can sample observations $(X,Y)$ consisting of $p$ explanatory variables, $X \in \mathcal{X}^{p}$, and an outcome, $Y \in \mathcal{Y}$.
Here, $\mathcal{X}$ is the set of possible values for each variable and $\mathcal{Y}$ denotes the possible outcomes; both sets may be either discrete or continuous.
Assume the joint distribution of the explanatory variables within any environment $e \in \{1,\ldots,E\}$, $P_X^e$, is known.
For simplicity, we imagine different samples as being independent of one another, although knockoffs can accommodate known patterns of dependency~\cite{sesia2020controlling}.
The model-X framework~\cite{candes2018panning} provides practical methods to test the conditional independence hypothesis
\begin{align} \label{eq:null-ci}
 \mathcal{H}^{\mathrm{ci}, e}_{j} : Y^{e} \indep X^{e}_{j} \mid X^{e}_{-j},
\end{align}
for any $j \in \{1,\ldots,\px \}$. Here, $X_{-j}$ denotes all observable explanatory variables except $X_j$, and the superscript $e$ clarifies we are focusing on the distribution of the data in the $e$-th environment.

Our goal is to powerfully test, for all $j \in \{1,\ldots,\px\}$, the following {\em consistent} conditional independence hypothesis,
\begin{align} \label{eq:null-ici}
 \mathcal{H}^{\ici}_{j} : \exists e \in \{1,\ldots, E\} \text{ such that the null } \mathcal{H}^{\mathrm{ci}, e}_{j} \text{ in~\eqref{eq:null-ci} is true,}
\end{align}
controlling the false discovery rate. Intuitively, we would interpret any findings by noting that, if $\mathcal{H}^{\ici}_{j}$~\eqref{eq:null-ici} is false, the conditional association of $X_j$ with $Y$ {\em must hold across all environments}.

It may be tempting to test $\mathcal{H}^{\mathrm{ci}, e}_{j}$~\eqref{eq:null-ci} separately environment by environment~\cite{candes2018panning} and then report the set of common discoveries. Unfortunately, this {\em intersection} heuristic would not control the false discovery rate for $\mathcal{H}^{\ici}_{j}$~\eqref{eq:null-ici} even if all environment-specific tests control it for $\mathcal{H}^{\mathrm{ci}, e}_{j}$~\eqref{eq:null-ci}~\cite{katsevich2018controlling}.
Alternatively, one may apply the standard knockoffs methodology~\cite{candes2018panning} on the {\em pooled} data from all environments. This would control the false discovery rate for $\mathcal{H}^{\mathrm{ci}}_{j}$~\eqref{eq:null-ci} defined in the broader population obtained by taking the union of all environments, but is not a test of consistency. Indeed, the problem is non-trivial, as illustrated by the simulations in Figure~\ref{fig:inv_sim}. These preview that our proposed method controls the false discovery rate for $\mathcal{H}^{\ici}_{j}$~\eqref{eq:null-ici}, unlike the two aforementioned heuristics, and achieves relatively high power.
These simulations will be explained with more details in Section~\ref{section:simulation_invariant}, after we develop our method.

\begin{figure}[!htb]
 \centering
 \includegraphics[width=\linewidth]{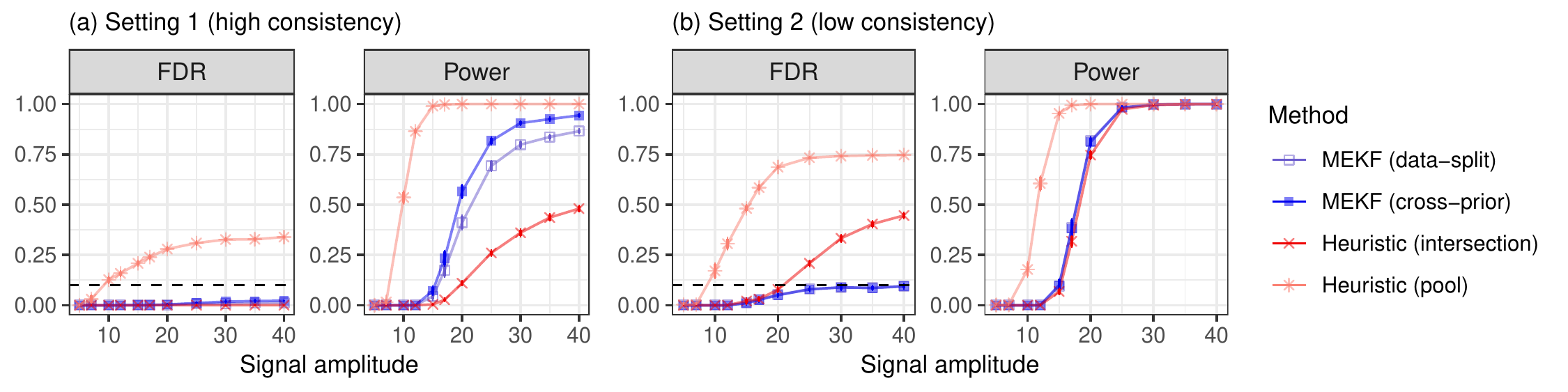}
 \caption{Performance of the multi-environment knockoff filter (MEKF) on simulated data from many environments, implemented with two alternative statistics. The performance is measured in terms of consistent conditional associations and is compared to that of two heuristics. The nominal false discovery rate is 0.1 (dashed line). (a) Data in which most conditional associations are consistent. (b) Data in which most conditional associations are not consistent.}
 \label{fig:inv_sim}
\end{figure}

A possible limitation of $\mathcal{H}^{\ici}_{j}$~\eqref{eq:null-ici} is that it becomes harder to reject if the number of environments grows but the sample size in each remains constant, because every one must provide evidence against the null.
Such intransigence is not always necessary; sometimes, it may be satisfactory to find an association in most environments, especially if some have smaller sample sizes.
This idea motivates {\em partial consistency} testing, also known as {\em partial conjunction}~\cite{friston2005conjunction,benjamini2008screening,heller2014replicability,wang2016detecting}.
For any variable $j \in \{1,\ldots,p\}$ and parameter $1 \leq r \leq E$, we define the partial consistency null hypothesis as:
\begin{align} \label{eq:null-pici}
 \mathcal{H}^{\pici,r}_{j} : \left| \left\{ e \in \{1,\ldots,E\} : \mathcal{H}^{\mathrm{ci}, e}_{j} \text{ is true} \right\} \right| > E-r.
\end{align}
In words, $\mathcal{H}^{\pici,r}_{j}$ states there are strictly less than $r$ environments $e$ in which $\mathcal{H}^{\mathrm{ci}, e}_{j}$~\eqref{eq:null-ci} is false; thus, a rejection suggests $X_j$ is associated with $Y$ in at least $r$ environments.
This generalizes $\mathcal{H}^{\ici}_{j}$~\eqref{eq:null-ici}, which is recovered if $r=E$.
Note that, unlike~\cite{benjamini2008screening}, we will not account for multiplicity over different possible $r$, which we take instead as fixed.

Before  outlining the methodology we propose to  test the consistent conditional independence hypotheses $\mathcal{H}^{\ici}_{j}$~\eqref{eq:null-ici},   we pause to motivate their interest. With this in mind, the next section  will establish a link between $\mathcal{H}^{\ici}_{j}$~\eqref{eq:null-ici} and causal inference, under some specific assumptions. 
We want to underscore, however, that identifying variables which are not independent from an outcome conditionally on all other observed sources of variation can be a powerful tool, even outside of the causal framework we will describe. 

Consider, as an example, the case where $Y$ is a measure of the quantitative competences acquired by students in public schools K-12, and let $X$ collect a number of variables that could possibly be related to the outcome, as student/teacher ratio, class size, availability of tutoring, school meals, type of curriculum, teacher qualifications, family size, attendance, etc. It is entirely possibly that  these variables interact in different ways to influence the final student competences in different environments (e.g.~urban vs.~rural, low vs.~high income neighborhoods, ethnically diverse vs.~homogeneous population, and so on). Without hoping to estimate constant causal effects, it would be of interest to be able to identify those variables which show a conditional association with the outcomes that holds across all environments. It is these ``robust" associations that one looks to in order to design interventions with the goal of improving outcomes, or detecting early signs of academic difficulties. While ``one size fits all'' is certainly wishful thinking, considerations of practicality, transparency and fairness make it preferable to focus on policies that have good chances of being effective across the board. Moreover, constant associations are also more likely to be robust to changes in environments or covariate shifts, such as those that may be expected with the passing of time. 

More generally, an ``environment'' can be understood as a specific population, as that typically observed in a data set. The goal of science is to make conclusions that are not valid only in one data set, but that hold more generally. Yet, it is well known that when looking for patterns among a very large number of variables, collected without much prior selection, it is possible to fit models that relate $X$ to $Y$ with a  precision that is hard to replicate in other data sets. We routinely carry out ``cross-validation'' to mitigate this problem, but this still only guarantees a form of {\em internal} consistency. It might very well be that among some Facebook users taken at a certain point in time, liking curly fries is a good predictor of IQ \cite{Kosinski5802} but this association is probably not robust to different times and settings. Testing for $\mathcal{H}^{\ici}_{j}$~\eqref{eq:null-ici} can be considered as a way to do ``{\em external} cross-validation'' \cite{Waldron2014}, thereby identifying models that are less ephemeral \cite{Efron2020}.

\section{From consistent associations to causal inferences} \label{sec:ci-to-causal}
We now  explicitly investigate the connection between consistent associations and causal inference. We will assume a constant causal model across environments---an assumption that is by no means necessary for the interest or validity of the tests we propose, but that is helpful to illustrate a number of ways in which focusing on conditional associations that are consistent across environments facilitates the discovery of variables with causal effects.

\subsection{A constant causal model} \label{sec:sem}

Assume a structural equation model~\cite{bollen1989} to describe the relation between the outcome ($Y$) and $\pz$ explanatory variables ($Z$), of which $\px$ are observed ($X$) and $\pc$ are unobserved ($C$). Thus, we can write $Z=(X,C) \in \mathcal{X}^{\pz}$, with $\pz=\px+\pc$.
According to this model, which we assume to be constant across environments, the $i$-th individual outcome, $Y^{(i)}$, is determined as a function of $Z$ through:
\begin{align} \label{eq:sem-full}
   Y^{(i)} = \bar{f}\left( Z^{(i)}, V^{(i)} \right),
\end{align}
where $\bar{f}$ is unknown and $V$ is exogenous noise from a standard uniform distribution, for example.
Assume the causal direction in this model is known: $Y$ is determined by some combination of $Z$ and $V$, not the other way around.
Although this simplification does not always make sense, it is appropriate in genetic studies, for example, because the genotypes predate the phenotype.
With this in place, consider the goal of discovering which variables are causal, in the sense that they have a direct effect on the outcome and do not satisfy the following {\em sharp} causal null hypothesis
\begin{align} \label{eq:null-causal}
 \mathcal{H}^{\mathrm{causal}}_{j} : \bar{f}\left( (z_1, \ldots, z_{j-1}, z_{j}, z_{j+1}, \ldots, z_{\pz}), v \right) = \bar{f}\left( (z_1, \ldots, z_{j-1}, z'_{j}, z_{j+1}, \ldots, z_{\pz}), v \right), \qquad \forall z,z_{j}',v.
\end{align}
Intuitively, this tells us  that intervening on $Z_j$ while holding all other variables fixed would have no effect on $Y$ at all; hence the adjective ``sharp''.
Alternatively, one may think of $\mathcal{H}^{\mathrm{causal}}_{j}$ as saying the potential outcomes are identical under all $Z_j$~\cite{rubin2005causal}.
Sharp hypotheses have received some criticism in the causal inference literature~\cite{neyman1935statistical}, specifically for not describing heterogeneous treatment effects~\cite{keele2015statistics}, but these concerns are more relevant from an estimation perspective. From a testing perspective, sharp hypotheses are helpful to discover which variables are more likely to be causal, especially for an exploratory analysis informing the design of follow-up studies~\cite{imbens2015causal}.
Further, the model in~\eqref{eq:sem-full} remains very flexible despite being constant across environments because it is fully non-parametric. In particular, as it allows any causal relations to be complex and involve interactions, this model does not exclude that a variable may appear to have different {\em linear} effects on the outcome across environments with covariate shifts. Indeed, the sole purpose of the model in~\eqref{eq:sem-full} is to allow a manageable definition of causality; in practice, our method simply tests whether a variable has {\em some} influence on the outcome in all environments.

In the following, we will assume we have data samples from $E$ environments, each corresponding to a distribution $P_Z^e$ of explanatory variables, for $e \in \{1,\ldots,E\}$. Conditional on $Z$, the outcome $Y$ is generated by the constant model in~\eqref{eq:sem-full}.
For simplicity, we equivalently rewrite this model as a function only of the causal variables, which we list as $\Pa{Y} \subseteq \{1,\ldots,\pz\}$ (standing for {\em parents of $Y$}):
\begin{align} \label{eq:parents}
  \Pa{Y} = \left\{ j \in \{1,\ldots,\pz\} : \mathcal{H}^{\mathrm{causal}}_{j} \text{ in~\eqref{eq:null-causal} is false} \right\}.
\end{align}
In particular, we write
\begin{align} \label{eq:sem}
   Y^{(i)} = f\left( Z^{(i)}_{\Pa{Y}}, V^{(i)} \right),
\end{align}
where $f$ is the restriction of $\bar{f}$ on $\Pa{Y}$
To lighten the subsequent notation, we partition $\Pa{Y}$ into:
\begin{align*}
  & \Pax{Y} \defeq \left\{j \in \{1,\ldots,\px\} : j \in \Pa{Y} \right\},
  & \Pac{Y} \defeq \left\{j \in \{1,\ldots,\pc\} : (\px + j) \in \Pa{Y} \right\};
\end{align*}
in words, these are the observed and unobserved causal variables, respectively.
Of course, in practice we analyze only the former, seeking to make inferences about $\Pax{Y}$, because we have no data involving the latter.

\subsection{The gap between conditional testing and causal inference} \label{sec:ci-testing}

We begin by focusing on a single environment $e$.
The result below states that testing $\mathcal{H}^{\mathrm{ci}, e}_{j}$ in~\eqref{eq:null-ci} amounts to making causal inferences if there is no confounding, in the sense that either no unobserved variables are causal, or the unobserved causal variables are independent of the observed ones.
\begin{prop} \label{prob:ci-to-causal}
Fix an environment $e$ and an observable variable $j$. Assume either (i) $\Pac{Y} = \emptyset$ or (ii) $X^{e}_j \indep C^{e}_{\Pac{Y}} \mid X^{e}_{-j}$.
Then, under the structural equation model and the data sampling scheme described in Section~\ref{sec:sem}, the causal null hypothesis $\mathcal{H}^{\mathrm{causal}}_{j}$~\eqref{eq:null-causal} implies the conditional independence null hypothesis $\mathcal{H}^{\mathrm{ci}, e}_{j}$~\eqref{eq:null-ci}.
\end{prop}
\begin{proof}
  If $\mathcal{H}^{\mathrm{causal}}_{j}$~\eqref{eq:null-causal} is true, $Y$ is a function of $Z_{\Pa{Y}}$ and $V$~\eqref{eq:sem}, with $j \notin \Pa{Y}$. (We drop the superscript $e$ for simplicity.)
Suppose assumption (i) holds: $Z_{\Pa{Y}}=X_{\Pa{Y}}$. As $Y$ is a function of $X_{\Pa{Y}}$ and $V$, it is independent of $X_j \mid  X_{-j}$ because $j \notin \Pa{Y}$ and $V \indep X$; thus $\mathcal{H}^{\mathrm{ci}, e}_{j}$~\eqref{eq:null-ci} is true.
Suppose instead (ii) holds. We may still write $Y$ as a function of $C_{\Pac{Y}}$, $V$, and $X_{\Pax{Y}}$. Thus, $Y \indep X_j \mid X_{-j}, C_{\Pac{y}}, V$ because $j \notin \{1,\ldots,\px\} \cap \Pa{Y}$. Then, it follows from the contraction property of conditional independence that $(Y,C_{\Pac{y}},V) \indep X_j \mid X_{-j}$, implying $\mathcal{H}^{\mathrm{ci}, e}_{j}$~\eqref{eq:null-ci}.
\end{proof}

The assumptions of Proposition~\ref{prob:ci-to-causal} are strong.
The first one would require one to have measured every variable with a possible effect on the outcome; this is often unrealistic, especially when studying complex phenomena. The second one holds in randomized experiments, and may be justified in certain observational studies such as those involving genetic parents-child trio data~\cite{Bates2020}. However, without suitable domain knowledge as in~\cite{Bates2020}, it is generally unclear why all potentially relevant missing variables should be conditionally independent of the observed ones.
Therefore, even if the structural equation model in Section~\ref{sec:sem} is acceptable, it remains very challenging to draw causal inferences from conditional associations.
Further, true random samples from the population of interest are not always easy to obtain; in fact, the process of gathering data is often far from ideal, as it may involve unknown sampling biases~\cite{heckman1979sample,harford2014,hargittai2015bigger} or network effects~\cite{shalizi2011homophily,lee2020network}.
Searching for consistent associations will not fully resolve these difficulties, but it can mitigate them if the environments are sufficiently different from one another, as discussed next.

\subsection{Consistency improves robustness to missing variables} \label{sec:robust-missing}

The next simple result demonstrates that the assumptions under which conditional associations yield causal inferences can be relaxed if the associations are consistent across many environments.
The intuition is that covariate shifts (changes in $P_Z^e$) may induce different observed variables to pick up spurious associations in different environments, while causal associations tend to be consistent.
 Figure~\ref{fig:graphical-missing} visualizes this concept with a toy example involving two causal variables, one of which is unmeasured. Here, the two environments differ in $P_Z^e$ to a sufficient extent that their spurious associations have no overlap and can thus be perfectly winnowed down through consistency.
Section~\ref{sec:gwas} will explain the relevance of this idea to genome-wide association studies involving individuals with different ancestries.

\begin{prop} \label{prob:ici-to-causal}
Fix any observable variable $j$ and consider $E$ environments $\{1,\ldots,E\}$. Assume either (i) $\Pac{Y} = \emptyset$ or (ii) $\exists e \in \{1,\ldots,E\} : X^{e}_j \indep C^{e}_{\Pac{Y}} \mid X^{e}_{-j}$.
Then, under the same setting as in Proposition~\ref{prob:ci-to-causal}, the causal null hypothesis $\mathcal{H}^{\mathrm{causal}}_{j}$ in~\eqref{eq:null-causal} implies the consistent association null hypothesis $\mathcal{H}^{\ici}_{j}$ in~\eqref{eq:null-ici}.
\end{prop}
\begin{proof}
Suppose either assumption (i) or (ii) is satisfied. Then, Proposition~\ref{prob:ci-to-causal} implies $\exists e \in \{1,\ldots,E\}$ such that the conditional null $\mathcal{H}^{\mathrm{ci}, e}_{j}$~\eqref{eq:null-ci} is true. In turn, this implies $\mathcal{H}^{\ici}_{j}$~\eqref{eq:null-ici}, by definition of the latter.
\end{proof}

\begin{figure}[!htb]
\centering

\begin{subfigure}{0.25\textwidth}
\begin{minipage}{0.12\linewidth}
\caption*{(a)}
\end{minipage}
\begin{minipage}{0.8\linewidth}
\includegraphics[width = \textwidth]{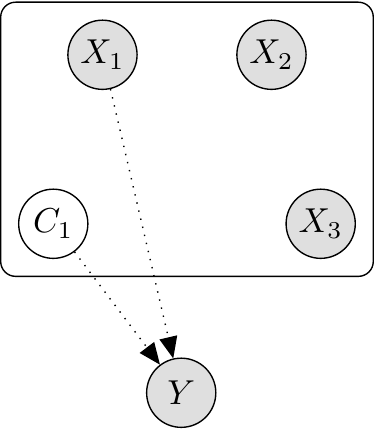}
\end{minipage}
~
\end{subfigure}
\begin{subfigure}{0.25\textwidth}
\begin{minipage}{0.12\linewidth}
\caption*{(b)}
\end{minipage}
\begin{minipage}{0.8\linewidth}
\includegraphics[width = \textwidth]{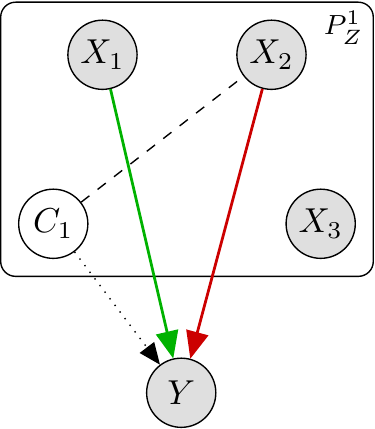}
\end{minipage}
\end{subfigure}
~
\begin{subfigure}{0.25\textwidth}
\begin{minipage}{0.12\linewidth}
\caption*{(c)}
\end{minipage}
\begin{minipage}{0.8\linewidth}
\includegraphics[width = \textwidth]{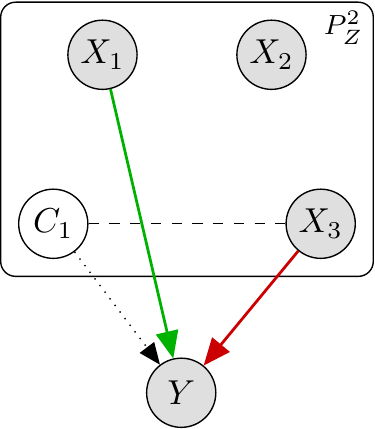}
\end{minipage}
\end{subfigure}

\caption{Graphical representation of consistency across environments improving robustness to missing variables. (a) Structural equation model for $Y \mid X,C$. (b) Conditional associations in environment one. (c) Conditional associations in environment two. The shaded nodes indicate the outcome or the observed variables, while the white node indicates the unobserved variable. The dotted arrows in (a) indicate causal links. The dashed segments represent graphically $P^e_Z$, which differs across environments and is such that $C_1$ is associated with $X_2$ in the first environment, and with $X_3$ in the second one, while all other variables are conditionally independent. The solid arrows indicate conditional associations between the observable variables and the outcome in each environment (green if causal, red if spurious).}
\label{fig:graphical-missing}

\end{figure}

\subsection{Consistency improves robustness to sampling biases} \label{sec:robust-bias}

To understand why consistency can also improve robustness to sampling biases, consider Berkson's paradox~\cite{berkson1946limitations}, which often arises in medicine~\cite{herbert2020spectre} and social science~\cite{dawes1975graduate}, and is also known as the problem of ``conditioning on a collider''~\cite{pearl2009causality}.
We illustrate this idea with a classical example from~\cite{dawes1975graduate}: imagine a school admitting students based on a composite score summarising their performance across different disciplines, and suppose only students whose composite score (standardized test + GPA + sports performance) exceeds a threshold are admitted.
Even if GPA and sports performance were independent in the applicant population, one should expect a negative correlation between these variables among the admitted students: students with low GPA must perform well in sports, or else they would have not been observed.
This bias makes it challenging to discover which variables may have a causal effect on GPA, especially if the admission criteria are unknown.
Fortunately, consistency may help remove spurious associations if data are available from schools relying on different admission criteria.

To make our argument more general without sacrificing concreteness, consider a scenario in which the data are collected through a simple environment-specific rejection-sampling rule. Suppose a random individual from the population of interest is included in the sample for the $e$-th environment with probability proportional to some bias $\phi^e$, which is a function of a subset $S^e \subseteq \{1,\ldots,\pz\}$ of the variables and possibly also of the outcome:
\begin{align} \label{eq:sampling-bias}
  \P{\text{include } (Z,Y) \text{ in the sample for environment } e} \propto \phi^e(Z_{S^e}, Y).
\end{align}
The effective joint distribution of $(Z,Y)$ accessible from the $e$-th environment then becomes
\begin{align*}
  P^e(Z,Y) \propto P_Z^e(Z) \cdot P^*(Y \mid Z) \cdot \phi^e(Z_{S^e}, Y),
\end{align*}
where $P^*(Y \mid Z)$ indicates the ideal conditional distribution of $Y \mid Z$ determined by the causal model.
Conditional testing on these data cannot lead to valid causal inferences even if no variables are missing, because the effective distribution of $Y \mid Z$ is now $P^e(Y \mid Z) \propto P^*(Y \mid Z) \cdot \phi^e(Z_{S^e}, Y)$, which no longer corresponds to the model of interest.
However, consistency can mitigate this issue, as long as the sampling biases (the subsets $S^e$) are not constant.

\begin{prop} \label{prob:ici-to-causal-bias}
Fix any observable variable $j$ and consider $E$ environments.
In the setting of Proposition~\ref{prob:ici-to-causal}, suppose the sampling mechanism in each environment $e$ is biased by a subset $S^e$ of explanatory variables as in~\eqref{eq:sampling-bias}.
Assume that there are no unmeasured confounders ($\pc = 0$), and $\exists e \in \{1,\ldots,E\} : j \notin S^e$.
Then, $\mathcal{H}^{\mathrm{causal}}_{j}$~\eqref{eq:null-causal} implies $\mathcal{H}^{\ici}_{j}$~\eqref{eq:null-ici}, where each $\mathcal{H}^{\mathrm{ci}, e}_{j}$ in~\eqref{eq:null-ici} refers to the effective population induced by the biased sampling mechanism.
\end{prop}
\begin{proof}
Focus on an environment $e$ such that $j \notin S^e$, which is assumed to exist (we will drop the superscript $e$ for simplicity).
The effective conditional distribution of $Y \mid Z$ is $P(Y \mid Z) \propto P^*(Y \mid Z) \cdot \phi(Z_{S}, Y)$, which does not depend on $Z_j$ because $j \notin S$ and $P^*(Y \mid Z)$ only depends on $\Pa{Y} \not\ni j$. Hence $\mathcal{H}^{\mathrm{ci},e}_{j}$~\eqref{eq:null-ci} is true, implying $\mathcal{H}^{\ici}_{j}$~\eqref{eq:null-ici}.
\end{proof}

Figure~\ref{fig:graphical-bias} visualizes this concept with an example in which two environments have biased sampling depending on disjoint sets of variables, ensuring no spurious associations are consistent.
The robustness of our proposed methods to sampling biases will be demonstrated with simulations in Section~\ref{sec:sim-sampling}.
Note that Proposition~\ref{prob:ici-to-causal-bias} presumes for simplicity that no variables are missing; a more general result relaxing this assumption could be obtained easily, although we do not explore that here because it would complicate the notation unnecessarily.

\begin{figure}[!htb]
\centering

\begin{subfigure}{0.25\textwidth}
\begin{minipage}{0.12\linewidth}
\caption*{(a)}
\end{minipage}
\begin{minipage}{0.8\linewidth}
\includegraphics[width = \textwidth]{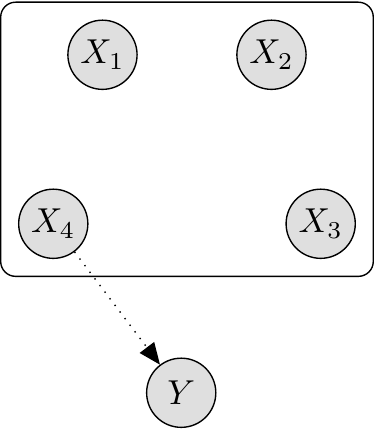}
\end{minipage}
~
\end{subfigure}
\begin{subfigure}{0.25\textwidth}
\begin{minipage}{0.12\linewidth}
\caption*{(b)}
\end{minipage}
\begin{minipage}{0.8\linewidth}
\includegraphics[width = \textwidth]{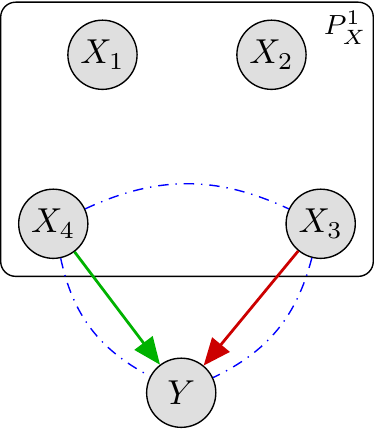}
\end{minipage}
\end{subfigure}
~
\begin{subfigure}{0.25\textwidth}
\begin{minipage}{0.12\linewidth}
\caption*{(c)}
\end{minipage}
\begin{minipage}{0.8\linewidth}
\includegraphics[width = \textwidth]{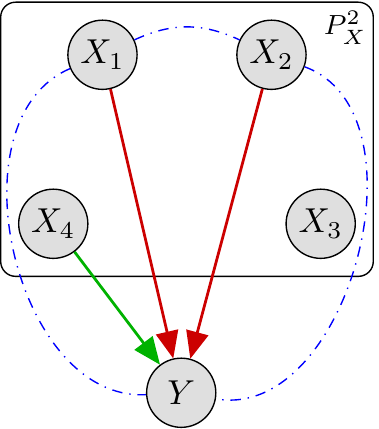}
\end{minipage}
\end{subfigure}

\caption{Graphical representation of consistency across environments improving robustness to sampling biases.
The dash-dotted curves represent conditional dependence relations induced by biased sampling in each environment; the first one depends on $X_3, X_4, Y$, while the second one depends on $X_1,X_2,Y$.
Other details are as in Figure~\ref{fig:graphical-missing}.
}
\label{fig:graphical-bias}

\end{figure}

\subsection{Consistency improves robustness to homophily and contagion}
\label{sec:robstness-homophily}

Consistency can also improve robustness to homophily and contagion~\cite{shalizi2011homophily} or, more generally, to unaccounted dependencies among the observations within the same environment, which can lead to spurious associations~\cite{lee2020network}.
Broadly speaking, we define homophily as the tendency of individuals sharing certain features to co-occur within one environment, and contagion as the mutual dependence of their outcomes.
This is a well-known phenomenon, especially in the social sciences and in network studies~\cite{shalizi2011homophily, mcpherson2001birds, fowler2008dynamic, aral2009distinguishing, lee2020network}, and it can be explained as follows. Consider a school in which students who enjoy reading are likely to join the book club (homophily). One student in the book club (patient zero) happens to be also interested in chess and convinces other members to learn how to play it (contagion).
Later, a school-wide survey finds an association between reading and playing chess, even though there is no relation between the two at the population level.
In particular, we do not expect to replicate this association in different schools (environments), as those may have different clubs and patient zeroes with other interests.
Therefore, the risk of incorrectly reporting spurious findings can be mitigated by searching for consistent associations. This argument could be formalized as in Section~\ref{sec:robust-bias}, but we prefer to avoid introducing additional notation here, partially for lack of space, and partially because the main idea should at this point be already clear.
The robustness of our proposed methods to homophily and contagion will be demonstrated with simulations in Section~\ref{sec:sim-homophily}.

\section{Methods} \label{sec:methods}
We start with a selective review  of the knockoff methodology upon which our solution will be built. We will then develop a method for testing $\mathcal{H}^{\ici}_{j}$~\eqref{eq:null-ici} and finally extend it to $\mathcal{H}^{\pici,r}_{j}$~\eqref{eq:null-pici}.

\subsection{Review: the methodology of knockoffs} \label{sec:knockoffs-review}

For an environment $e$, knockoffs enable testing $\mathcal{H}^{\mathrm{ci}, e}_{j}$~\eqref{eq:null-ci} for all $j \in \{1,\ldots,p\}$.
The idea is to augment the data for each of the $n$ observed individuals with $p$ synthetic features, the {\em knockoffs}, which serve as negative control variables~\cite{barber2015controlling,candes2018panning}.
The knockoffs are created by the statistician as a function of $X$, without looking at $Y$, and therefore they are conditionally independent of $Y$ given $X$. Knockoffs are however pairwise exchangeable with the original variables in their joint distribution.
That is, if $[ \mathbf{X}^e, \tilde{\mathbf{X}}^e ] \in \mathbb{R}^{n \times 2p}$ is the matrix obtained by concatenating the variables $\mathbf{X}^e \in \mathbb{R}^{n \times p}$ with the corresponding knockoffs $\tilde{\mathbf{X}}^e \in \mathbb{R}^{n \times p}$ and, for any $j \in \{1,\ldots,p\}$, the matrix $[ \mathbf{X}^e, \tilde{\mathbf{X}}^e]_{\mathrm{swap}(j)}$ is obtained by swapping the $j$-th column of $\mathbf{X}^e$ with the $j$-th column of $\tilde{\mathbf{X}}^e$, then
\begin{align} \label{eq:knock_cond_1}
  \left[ \mathbf{X}^e, \tilde{\mathbf{X}}^e \right]_{\mathrm{swap}(j)}
  \overset{d}{=} \left[ \mathbf{X}^e, \tilde{\mathbf{X}}^e \right].
\end{align}
The equation above says that swaps of a variable with its knockoff cannot be detected without looking at $Y$; in fact, 
the only possible significant difference between $X_j$ and $\tilde{X}_j$ is the lack of conditional association of the latter with~$Y$.
The construction of knockoffs depends on the joint distribution $P_X^e$ of all explanatory variables, and we refer to prior works for specific algorithms; see~\cite{candes2018panning} for the multivariate Gaussian case,~\cite{sesia2018} for hidden Markov models,~\cite{bates2020metropolized} for general graphical models, or~\cite{romano2019} for fully nonparametric approximate methods.
%The assumption that $P_X^e$ is known is justifiable in applications with relevant prior knowledge, such as genome-wide association studies~\cite{sesia2018}, or when large unsupervised data sets can be utilized to learn an accurate approximation~\cite{candes2018panning,barber2018robust}.
As our contribution does not concern this aspect of the analysis, we assume $P_X^e$ is known and a knockoff generation algorithm is available.

The second step is to fit a model predicting $Y$ given $X$ and $\tilde{X}$, computing importance measures $T_j^e$ and $\tilde{T}_j^e$ for each $X_j$ and $\tilde{X}_j$, respectively.
Any model can be employed, as long as it does not unfairly discriminate between variables and knockoffs---swapping any $X_j$ with $\tilde{X}_j$ should result in $T_j$ being swapped with $\tilde{T}_j$~\cite{candes2018panning}.
A typical choice is to fit a sparse generalized linear model (e.g., the lasso~\cite{tibshirani2011regression}), tuning its regularization parameter via cross-validation; then, the absolute values of the (scaled) regression coefficients are powerful importance measures~\cite{candes2018panning}.
For each $j \in \{1,\ldots,p\}$, $T_j^e$ and $\tilde{T}_j^e$ are combined into an anti-symmetric statistic, e.g., $W_j^e = T_j^e - \tilde{T}_j^e$.
This yields statistics $W_j^e$ that are equally likely to be positive or negative if $\mathcal{H}^{\mathrm{ci}, e}_{j}$~\eqref{eq:null-ci} is true~\cite{candes2018panning}. By contrast, a large positive $W_j^e$ is evidence against the null. Further, the signs of these statistics are mutually independent for all null indices $j$ conditional on the absolute values, $|W^e| = (|W_1^e|, \ldots, |W_p^e|)$.
Formally, if $\epsilon^e \in \{-1,+1\}^p$ is a random vector such that $\epsilon_j^e = +1$ if $\mathcal{H}^{\mathrm{ci}, e}_{j}$ is false and $\P{\epsilon_j^e = +1 } = 1/2$, independently of everything else, otherwise, then $W^e$ satisfies the following {\em coin-flip} property~\cite{candes2018panning}:
\begin{align} \label{eq:flip-sign}
  W^e \overset{d}{=} W^e \odot \epsilon^e,
\end{align}
where $\odot$ indicates element-wise multiplication. Therefore, the signs of $W^e$ can be seen as one-bit conservative p-values~\cite{barber2015controlling} for $\mathcal{H}^{\mathrm{ci}, e}_{j}$~\eqref{eq:null-ci}, if we transform them as $p_j^e = 1/2$ if $W_j^e>0$ and $p_j^e = 1$ otherwise. The ordering provided by the absolute values of $W^e$ allows one to powerfully test the above hypotheses with a sequential procedure~\cite{barber2015controlling}.

Concretely, the knockoff filter~\cite{barber2015controlling} computes a significance threshold for the test statistics such that the rejection of $\mathcal{H}^{\mathrm{ci}, e}_{j}$~\eqref{eq:null-ci} for all $j$ with larger $W_j^e$ controls the false discovery rate below the desired level $\alpha$. That is, one can prove
\begin{align*}
  \text{FDR} \defeq \E{\frac{| \{j :\mathcal{H}^{\mathrm{ci}, e}_{j} \text{ is rejected}\} \cap \{j : \mathcal{H}^{\mathrm{ci}, e}_{j} \text{ is true}\} |}{|\{j :\mathcal{H}^{\mathrm{ci}, e}_{j} \text{ is rejected}\} | \lor 1 }} \leq \alpha,
\end{align*}
where $a \lor b \defeq \max\{a,b\}$.
Equivalently, the knockoff filter can be seen as an instance of the selective SeqStep+ test~\cite{barber2015controlling} applied to the above ordered one-bit p-values.
In the next section, we will extend this methodology to simultaneously analyze data from many environments, controlling the false discovery rate for the hypotheses $\mathcal{H}^{\ici}_{j}$ defined in~\eqref{eq:null-ici}.

\subsection{Multi-environment knockoff statistics} \label{sec:ici-knockoffs}

Consider $E$ environments $e \in \{1,\ldots,E\}$, each corresponding to observations $\mathbf{Y}^e \in \mathbb{R}^n, \mathbf{X}^e \in \mathbb{R}^{n \times p}$, and knockoffs $\tilde{\mathbf{X}}^e \in \mathbb{R}^{n \times p}$.
Assume the sample size is $n$ in all environments for ease of notation, although this is unnecessary.
The first ingredient for testing $\mathcal{H}^{\ici}_{j}$~\eqref{eq:null-ici} are the multi-environment statistics defined below, which generalize those from~\cite{candes2018panning}.

\begin{definition}[Multi-environment knockoff statistics] \label{def:kenv-stats}
$\mathbf{W} \in \mathbb{R}^{E \times p}$ are valid multi-environment knockoff statistics if they satisfy $\mathbf{W} \,\oset{d}{=}\, \mathbf{W} \odot \bm{\epsilon}$, and $\bm{\epsilon} \in \{\pm 1\}^{E \times p}$ is a random matrix with independent entries and rows $\epsilon^e$ such that $\epsilon^e_j = \pm 1 $ with probability $1/2$ if $ \mathcal{H}_{j}^{\mathrm{ci}, e}$ in~\eqref{eq:null-ci} is true and $\epsilon^e_j = +1$ otherwise, for $j \in \{1,\ldots,p\}$ and $e \in \{1,\ldots,E\}$.
\end{definition}

A simple solution to obtain such statistics is to analyze different environments separately with the existing method from the previous section, and then stack their  output $W^e$.
This approach, which we call {\em data-splitting}, is computationally efficient and private, as it allows researchers in different environments to collaborate without disclosing their data, but it is not the most powerful. For example, if all environments are identical, data splitting effectively divides the total sample size by $E$ compared to a regular analysis of the pooled data, although the latter would test the same hypotheses in this scenario; this motivates a more general approach.

Let $\mathbf{Y} \in \mathbb{R}^{En}, \mathbf{X} \in \mathbb{R}^{En \times p}, \tilde{\mathbf{X}} \in \mathbb{R}^{En \times p}$ indicate the data matrices obtained by stacking the observations or knockoffs from all environments.
We define $[\mathbf{T},\tilde{\mathbf{T}}] \in \mathbb{R}^{E \times 2p}$ as a matrix of multi-environment importance measures for all variables and knockoffs, computed by applying a randomized function $\boldsymbol{\tau}$ to the full data set:
\begin{align} \label{eq:def-tau}
  & [\mathbf{T},\tilde{\mathbf{T}}] = \boldsymbol{\tau} \left( \mathbf{Y}, [\mathbf{X}, \tilde{\mathbf{X}}] \right)
    \defeq \left[ \bm{t} \left( \mathbf{Y}, [\mathbf{X}, \tilde{\mathbf{X}}] \right), \tilde{\bm{t}} \left( \mathbf{Y}, [\mathbf{X}, \tilde{\mathbf{X}}] \right) \right].
\end{align}
Above, $\bm{t}$ (resp.~$\tilde{\bm{t}}$) are defined in terms of the matrix $\boldsymbol{\tau}$, as its first (resp.~last) $p$ columns.
For any subset $\mathcal{S} \subseteq \{1,\ldots,E\} \times \{1,\ldots,p\}$, let $[\mathbf{X}, \tilde{\mathbf{X}}]_{\text{swap}(\mathcal{S})}$ be the matrix obtained from $[\mathbf{X}, \tilde{\mathbf{X}}]$ after swapping the column $\mathbf{X}^{e}_j$ for environment $e$ with the corresponding $\tilde{\mathbf{X}}^{e}_j$, for all $(e,j) \in \mathcal{S}$.
Note the slight change of notation compared to~\eqref{eq:knock_cond_1}: there, swapping was defined only for one environment.
The function $\boldsymbol{\tau}$ may be almost anything, possibly involving sophisticated machine learning algorithms, as long as it satisfies
\begin{align} \label{eq:swap-tau}
  \boldsymbol{\tau}\left( \mathbf{Y}, [\mathbf{X}, \tilde{\mathbf{X}}]_{\text{swap}(\mathcal{S})} \right)
  \overset{d}{=}
  \left[
  \bm{t}\left( \mathbf{Y}, [\mathbf{X}, \tilde{\mathbf{X}}] \right),
  \tilde{\bm{t}}\left( \mathbf{Y}, [\mathbf{X}, \tilde{\mathbf{X}}] \right)
  \right]_{\text{swap}(\mathcal{S})} \quad \mid \mathbf{Y}, [\mathbf{X}, \tilde{\mathbf{X}}].
\end{align}
In words, this says that swapping variables and knockoffs in any one environment should have the only effect of swapping the corresponding importance measures in that environment, leaving all other elements of $\boldsymbol{\tau}$ unchanged. Importantly, the equality only needs to hold in distribution conditional on the data and on the knockoffs, as $\bm{\tau}$ may be randomized.
For our method to be powerful, we also need $\bm{\tau}$ to be such that a larger value in row $e$ and column $j$ indicates evidence from environment $e$ that $X_j$ is conditionally associated with $Y$, while a larger value in row $e$ and column $j+p$ points to a corresponding association of $\tilde{X}_j$ with $Y$, which we know must be spurious because all knockoffs are null. Concrete examples of powerful multi-environment importance measures will be presented later.

Given any matrix $[\mathbf{T},\tilde{\mathbf{T}}]$~\eqref{eq:def-tau} of importance measures satisfying~\eqref{eq:swap-tau}, we construct multi-environment statistics $\mathbf{W} \in \mathbb{R}^{E \times p}$ by computing its rows $W^{e} \in \mathbb{R}^{p}$ with the standard approach reviewed in Section~\ref{sec:knockoffs-review}, separately environment by environment. Precisely, we contrast each $T^e_j$ with $\tilde{T}^e_j$, pairwise for all variables and knockoffs; e.g.,
\begin{align} \label{eq:def-w}
  W^e_j = T^e_j - \tilde{T}^e_j.
\end{align}
This ensures the signs of $\mathbf{W}$ corresponding to null hypotheses are independent coin flips conditional on the absolute values of all entries, generalizing the result in~\eqref{eq:flip-sign} from~\cite{candes2018panning}.

\begin{prop} \label{prop:coin-flip}
If the statistics $\mathbf{W} \in \mathbb{R}^{E \times p}$ are computed based on~\eqref{eq:swap-tau}--\eqref{eq:def-w}, then they satisfy Definition~\ref{def:kenv-stats}.
\end{prop}

\begin{proof}
Let $\bm{w}$ denote the matrix-valued (randomized) function computing the knockoff statistics in all environments, consistently with~\eqref{eq:swap-tau}--\eqref{eq:def-w}, i.e., $\mathbf{W} = \bm{w}(\mathbf{Y}, [\mathbf{X}, \tilde{\mathbf{X}}])$.
With $\bm{\epsilon}$ as in Definition~\ref{def:kenv-stats}, define the set $\mathcal{S} = \{(e,j) : \epsilon^e_j = -1\}$.
It follows from~\eqref{eq:swap-tau}--\eqref{eq:def-w} that $\bm{w}(\mathbf{Y}, [\mathbf{X}, \tilde{\mathbf{X}}]_{\text{swap}(\mathcal{S})}) \,\oset{d}{=}\, \bm{\epsilon} \odot \bm{w}(\mathbf{Y}, [\mathbf{X}, \tilde{\mathbf{X}}])$ conditional on $\mathbf{Y}$ and $[\mathbf{X}, \tilde{\mathbf{X}}]$, which implies
$\bm{w}(\mathbf{Y}, [\mathbf{X}, \tilde{\mathbf{X}}]_{\text{swap}(\mathcal{S})}) \,\oset{d}{=}\, \bm{\epsilon} \odot \bm{w}(\mathbf{Y}, [\mathbf{X}, \tilde{\mathbf{X}}])$ marginally.
We also know that $[\mathbf{X}, \tilde{\mathbf{X}}]_{\text{swap}(\mathcal{S})} \,\oset{d}{=}\, [\mathbf{X}, \tilde{\mathbf{X}}] \mid \mathbf{Y}$ (Lemma 3.2 in~\cite{candes2018panning}), which implies $\bm{w}(\mathbf{Y}, [\mathbf{X}, \tilde{\mathbf{X}}]_{\text{swap}(\mathcal{S})}) \,\oset{d}{=}\, \bm{w}(\mathbf{Y}, [\mathbf{X}, \tilde{\mathbf{X}}])$. Therefore, $\bm{w}(\mathbf{Y}, [\mathbf{X}, \tilde{\mathbf{X}}]) = \bm{\epsilon} \odot \bm{w}(\mathbf{Y}, [\mathbf{X}, \tilde{\mathbf{X}}])$.
\end{proof}

An intuitive algorithm to compute powerful multi-environment statistics is the following, which we call the {\em empirical cross-prior} approach.
The idea is that the requirements of Proposition~\ref{prop:coin-flip} are satisfied even if each row $W^e$ leverages the observations from other environments $\{1,\ldots,E\} \setminus \{e\}$, as long as the latter are perturbed through suitable random column swaps. This perturbation only acts on a copy of $[\mathbf{X}, \tilde{\mathbf{X}}]$ and all subsequent steps of the analysis will begin with the original data, so the order in which the environments are processed is irrelevant.
Concretely, we perturb the data by randomly swapping each variable with its own knockoff based on a coin flip, independently observation by observation. Then, we estimate importance weights that are symmetric with respect to variables and knockoffs; these will serve as prior information for the subsequent analysis step (hence the name for this approach). Figure~\ref{fig:schematic} provides a schematic of the full procedure, while the details are explained below.

\begin{figure}[!htb]
\centering
\includegraphics[width = 0.85\textwidth]{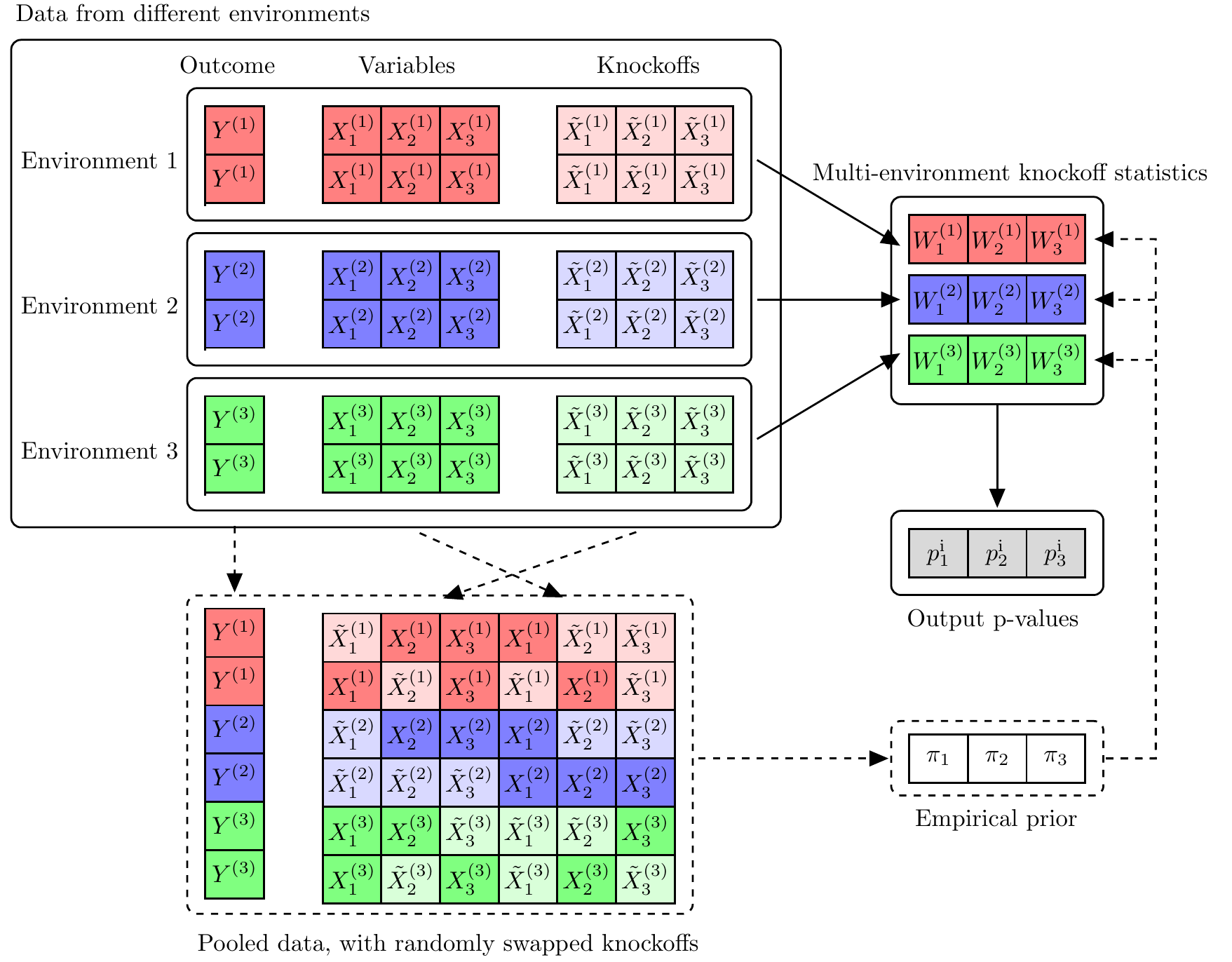}
\caption{Schematics for a multi-environment knockoff analysis. In this example, there are 3 variables, 3 environments, and 2 observations per environment. The solid arrows represent the analysis based on data-splitting statistics, in which the environments are analyzed separately before combining the resulting knockoff statistics. The dashed arrows represent the additional steps corresponding to our empirical cross-prior statistics, which analyze jointly the data from all environments. The darker blocks indicate the real data, while the lighter ones indicate the knockoffs.}
\label{fig:schematic}
\end{figure}

Let $\mathbf{V} \in \{0,1\}^{En \times p}$ be a random matrix of i.i.d.~coin flips, and $[\mathbf{X},\tilde{\mathbf{X}}]_{\text{swap}(\mathbf{V})}$ be the perturbed matrix obtained by swapping the $i$-th observation of $X_j$ with the corresponding knockoffs if and only if $V_{ij} = 1$.
We compute prior importance measures $T^{\text{prior}}$ (resp.~$\tilde{T}^{\text{prior}}$) for all variables (resp.~knockoffs) as the absolute values of the regression coefficients estimated by fitting a sparse generalized linear regression model (e.g., the lasso) to predict $\mathbf{Y}$ given $[\mathbf{X},\tilde{\mathbf{X}}]_{\text{swap}(\mathbf{V})}$, tuning the regularization parameter by cross-validation. The results are combined symmetrically into a prior weight $\pi_j$ for each $j$, e.g., $\pi_j = \zeta(T^{\text{prior}}_j+\tilde{T}^{\text{prior}}_j)$ for some positive and decreasing function~$\zeta$, such as $\zeta(t) = 1/(0.05+t)$.
Finally, we compute the importance measures $T^e$ and $\tilde{T}^e$ based on the unperturbed data in the $e$-th environment by combining the above prior with the same lasso-based approach as in the data-splitting case.
In particular, the regularization penalty is now feature-specific and depends on two parameters, $\lambda^e>0$ and $\gamma^e \in [0,1]$, both tuned by cross-validation, as well as on the prior weights. Specifically, the penalty for the $j$-th variable is $\lambda^e_j = \lambda^e (1-\gamma^e) + \gamma^e \pi_j$.
This reduces to data splitting if we fix $\gamma^e =0$. In general, larger values of $\gamma^e$ may improve power by making it less likely for the environment-specific models to select spurious variables.
However, the prior may not always be very informative and, even if it is, it is unclear how much weight it should be given; thus we also tune $\gamma^e$ by cross-validation. To avoid a two-dimensional grid search, in practice we first tune $\lambda^e$ and then $\gamma^e$.

\begin{prop} \label{prop:cross-prior}
The statistics $\mathbf{W}$ obtained by applying~\eqref{eq:def-w} to the empirical cross-prior importance measures $\mathbf{T}$ and $\tilde{\mathbf{T}}$ described above satisfy Definition~\ref{def:kenv-stats}.
\end{prop}
\begin{proof}
By Proposition~\ref{prop:coin-flip}, it suffices to show $[\mathbf{T}, \tilde{\mathbf{T}}]$ satisfy~\eqref{eq:swap-tau} because~\eqref{eq:def-tau} is trivial.
For any $\mathcal{S} \subseteq \{1,\ldots,E\} \times \{1,\ldots,p\}$ consider how $[\mathbf{T}, \tilde{\mathbf{T}}]$ would change if $[\mathbf{X}, \tilde{\mathbf{X}}]$  were replaced by $[\mathbf{X}, \tilde{\mathbf{X}}]_{\text{swap}(\mathcal{S})}$. The joint distribution of the prior weights is invariant because $\pi$ depends on $[\mathbf{X}, \tilde{\mathbf{X}}]$ through $[\mathbf{X},\tilde{\mathbf{X}}]_{\text{swap}(\mathbf{V})}$, and the composition of $\text{swap}(\mathcal{S})$ with $\text{swap}(\mathbf{V})$ is statistically indistinguishable from $\text{swap}(\mathbf{V})$. Similarly, the cross-validation errors within each environment are invariant to the ordering of variables and knockoffs, and so are the optimal values of $\lambda^{e}$ and $\gamma^{e}$. Therefore, the only change is that the final importance measure $T_j^e$ is swapped with the corresponding $\tilde{T}_j^e$ if and only if $(e,j) \in \mathcal{S}$.
\end{proof}

%The key idea of the empirical cross-prior statistics is to exploit perturbed data from all environments, while the other details are a convenient choice we found to work well.
Perturbing the data by randomly swapping variables and knockoffs is essential (see the proof of Proposition~\ref{prop:cross-prior}) and, although it weakens the signals, it does not destroy them entirely.
In fact, the empirical prior can learn which {\em pairs} of variables and knockoffs are likely to be important, although the effective sample for this task is cut in half.
In any case, the prior cannot do much harm because its influence on the output is controlled by $\gamma^e$, which is tuned by cross-validation. Therefore, we would expect at worst to select $\gamma^e\approx 0$ and thus approximately recover the data-splitting solution.

Statistics computed by naively looking at the full unperturbed data from all environments are generally invalid. For example, one could imagine using all data to compute the magnitudes of $W^e_j$, similarly to how we estimate the above empirical prior but without the initial perturbation, and then leveraging only the observations in environment $e$ to determine the signs.
This would yield statistics that are separately valid in each environment but are not mutually independent for different $e$, compromising our analysis downstream; see Appendix~\ref{sec:app-counter} and Figure~\ref{fig:counter_example} therein.

\subsection{The multi-environment knockoff filter} \label{sec:consistency}

Starting from a matrix $\mathbf{W}$ of multi-environment knockoff statistics, we can combine its rows to obtain one-bit conservative p-values~\cite{barber2015controlling} for testing $\mathcal{H}^{\ici}_{j}$ in~\eqref{eq:null-ici}.
Precisely, for each $j \in \{1,\ldots,p\}$, we compute
\begin{align} \label{eq:ici-pvalues}
    & p_j^{\inv} = \begin{cases}
        1/2, & \text{ if } \min \left\{ \sign \left(W_{j}^{e}\right) \right\}_{e=1}^{E} = +1, \\
        1, & \text{ otherwise.}
    \end{cases}
\end{align}
The order in which these hypotheses will be tested depends on the absolute values of $\mathbf{W}$, which we combine column-wise with some symmetric function $\bar{w}$ to obtain invariant statistics $\left|W_{j}^{\inv}\right|$ for all $j \in \{1,\ldots,p\}$:
\begin{align} \label{eq:invariant-abs-stat}
    |W_{j}^{\inv}| = \bar{w} \p{|W_{j}^{1}|, \dots,|W_{j}^{E}| }.
\end{align}
We will adopt here $\bar{w} \p{|W_{j}^{1}|, \dots,|W_{j}^{E}| } = \prod_{e=1}^{E} |W^e_j|$, although other options are possible.
The next result states that, conditional on the ordering determined by~\eqref{eq:invariant-abs-stat}, the one-bit p-values in~\eqref{eq:ici-pvalues} are valid for $\mathcal{H}^{\ici}_{j}$~\eqref{eq:null-ici}.

\begin{prop}[Multi-environment p-values] \label{prop:pvals-invariant}
If the statistics $\mathbf{W}$ satisfy Definition~\ref{def:kenv-stats}, the p-values $p_j^{\inv}$ in~\eqref{eq:ici-pvalues} corresponding to true $\mathcal{H}^{\ici}_{j}$~\eqref{eq:null-ici} stochastically dominate the uniform distribution conditional on $|W^{\inv}| = (|W_{1}^{\inv}|, \ldots, |W_{p}^{\inv}|)$. Further, these p-values are ``almost independent", i.e., $\P{p_j^{\inv} \leq \alpha \mid |W^{\inv}|, p_{-j}^{\inv}} \leq \alpha$,  $\forall \alpha \in [0,1]$ if $\mathcal{H}^{\ici}_{j}$ is true.
\end{prop}
\begin{proof}
It suffices to prove the second part of this statement because that implies the first one.
Take any $j \in \{1,\ldots,p\}$ corresponding to a true $\mathcal{H}^{\ici}_{j}$~\eqref{eq:null-ici}, and assume for simplicity $|W^{\inv}_j|>0$ (if not, $p_j^{\inv} = 1$ by definition).
By definition of $\mathcal{H}^{\ici}_{j}$, we know $\exists e \in \{1,\ldots,E\}$ such that the environment-specific null $\mathcal{H}^{\mathrm{ci}, e}_{j}$ in~\eqref{eq:null-ci} is true. Therefore,
\begin{align*}
    \P{ p_j^{\inv} = 1/2 \mid |W^{\inv}|, p_{-j}^{\inv}}
    & = \P{W_{j}^{l} \geq 0 \; \forall l \in \{1,\ldots,E\} \mid |W^{\inv}|, p_{-j}^{\inv}} \\
    & \leq \P{W_{j}^{e} > 0 \mid |W^{\inv}|, p_{-j}^{\inv} }
      = 1/2.
\end{align*}
The above inequality follows from $|W^{\inv}_j|>0$ and hence $W^e_j \neq 0$ for all $e$, while the last equality follows from Definition~\ref{def:kenv-stats} because $|W^{\inv}|$ is a function of $|\mathbf{W}|$ and $p_{-j}^{\inv}$ of $\mathbf{W}_{-j}$.
The proof is completed by $p_j^{\inv} \in \{1/2,1\}$.
\end{proof}

Proposition~\ref{prop:pvals-invariant} suggests applying a sequential testing approach, such as selective SeqStep+ (the knockoff filter)~\cite{barber2015controlling}, to the p-values in~\eqref{eq:ici-pvalues}.
However, this is not obviously valid because the null p-values $p_j^{\inv}$ are not independent~\cite{barber2015controlling}. In fact, each of them is also affected by the signs of entries of $\mathbf{W}$ corresponding to non-null environments, which need not be independent; $\mathcal{H}^{\ici}_{j}$ only says there exists a null entry in the $j$-th column of $\mathbf{W}$, but $p_j^{\inv}$ also counts the signs of the others.
Fortunately, the next result, proved in Appendix~\ref{app:proof-selseqstep}, shows the ``almost-independence'' established by Proposition~\eqref{prop:pvals-invariant} is sufficient. The intuition is that this dependence at worse makes our p-values conservative.

\begin{restatable}[Multi-environment knockoff filter]{theorem}{theoremseqstep} \label{th:seqstep-fdr}
The selective SeqStep+ procedure of~\cite{barber2015controlling} applied to p-values $p_j^{\inv}$ ordered by $|W^{\inv}_j|$ and satisfying the ``almost-independence'' property of Proposition~\eqref{prop:pvals-invariant}, i.e., $\P{p_j^{\inv} \leq \alpha \mid |W^{\inv}|, p_{-j}^{\inv}} \leq \alpha$ for any $\alpha \in [0,1]$, controls the false discovery rate below the nominal level.
\end{restatable}

\subsection{Testing partially consistent conditional associations} \label{sec:partial-consistency}

Starting from multi-environment knockoff statistics $\mathbf{W}$, we summarise them column-wise as follows. For each $j \in \{1,\ldots,p\}$, let $n^-_j$ count the negative signs in the $j$-th column, and let $n^0_j$ be the number of zeros. Then, compute
\begin{align}  \label{eq:pici-pvalues}
\begin{split}
  p_j^{\pinv,r}
  & \defeq \Psi\left( n^-_j -1, (E-r+1-n^0_j) \lor 0, \frac{1}{2} \right) + U_j \cdot \psi\left( n^-_j, (E-r+1-n^0_j) \lor 0, \frac{1}{2}\right),
\end{split}
\end{align}
where $\Psi(x, m,\pi)$ is the binomial cumulative distribution function evaluated at $x$,
$\psi(x,m,\pi)$ is the corresponding probability mass, and
$U_j \sim \operatorname{Uniform}[0,1]$ independently of all else.
Below, we show these are conservative p-values for $\mathcal{H}^{\pici,r}_{j}$~\eqref{eq:null-pici}, for any fixed $r$.
As in the previous section, we will filter them sequentially in decreasing order of
\begin{align} \label{eq:pici-order}
  \left|W_{j}^{\pinv,r}\right| = \bar{w}\p{|W_{j}^{1}|, \dots,|W_{j}^{E}| },
\end{align}
for some symmetric function $\bar{w}$. Concretely, we will focus on the $\bar{w}$ that multiplies the top $r$ largest entries in that column by absolute value, thus extending the solution from the previous section:
\begin{align*}
w\p{|W_{j}^{1}|, \dots,|W_{j}^{E}| } = \prod_{e = 1}^r  \abs{W_{j}}^{(E-e+1)}.
\end{align*}
Above, $\abs{W_{j}}^{(e)}$ are the order statistics for the absolute values in the $j$-th column of $\mathbf{W}$.
The next result proves these ordered p-values are conservative for $\mathcal{H}^{\pici,r}_{j}$~\eqref{eq:null-pici}, and ``almost independent'' of each other in the sense of Proposition~\ref{prop:pvals-invariant}.
Combined with Theorem~\ref{th:seqstep-fdr}, this will guarantee false discovery rate control using selective SeqStep+.

\begin{prop}[Partial conjunction multi-environment p-values] \label{prop:PC_FDR}
If the multi-environment statistics $\mathbf{W}$ satisfy Definition~\ref{def:kenv-stats}, for any fixed $r \in \{1,\ldots,E\}$, the p-values $p_j^{\inv}$~\eqref{eq:pici-pvalues} corresponding to true $\mathcal{H}^{\pici,r}_{j}$ in~\eqref{eq:null-pici} stochastically dominate the uniform distribution conditional on $|W^{\pinv,r}| = (|W_{1}^{\pinv,r}|, \ldots, |W_{p}^{\pinv,}|)$~\eqref{eq:pici-order}. Further, these p-values are ``almost independent", in the sense that $\mathbb{P} [ p_j^{\pinv,r} \leq \alpha \mid |W^{\pinv,r}|, p_{-j}^{\pinv,r} ] \leq \alpha$,  $\forall \alpha \in [0,1]$ if $\mathcal{H}^{\pici,r}_{j}$ is true.
\end{prop}

\begin{proof}
As in the proof of Proposition~\ref{prop:pvals-invariant}, it suffices to establish the second part of this statement.
Take any $j \in \{1,\ldots,p\}$ such that $\mathcal{H}^{\pici,r}_{j}$ in~\eqref{eq:null-pici} holds, so that there are at least $(E-r+1-n^0_j) \lor 0$ environments in which $\mathcal{H}^{\mathrm{ci}, e}_{j}$ is true and $W^{e}_j \neq 0$.
Without loss of generality, assume these environments are those indexed by $C_j \defeq \{1,\ldots,(E-r+1-n^0_j) \lor 0\}$, with $C_j = \emptyset$ if $E-r+1-n^0_j \leq 0$.
Then, define the random variable $n^{*-}_j$ as the number of negative signs in $\mathbf{W}^{C_j}_j$.
As $\mathbf{W}$ satisfies the flip-sign property from Definition~\ref{def:kenv-stats}, $\mathbf{W} \,\oset{d}{=}\, \mathbf{W} \odot \bm{\epsilon}$, it follows the $n^{*-}_j$'s are independent $\text{Binomial}(|C_j|, 1/2)$, with $|C_j|=(E-r+1-n^0_j) \lor 0$, conditional on $|\mathbf{W}|$ and hence also conditional on $|W^{\pinv,r}|$. Now, define $p_j^{*\pinv,r}$ as the imaginary p-value obtained by replacing $n^{-}_j$ with $n^{*-}_j$ in~\eqref{eq:null-pici}.
Because $n^{*-}_j \leq n^{-}_j$, $\forall \alpha \in [0,1]$,
\begin{align*}
  \P{ p_j^{\pinv,r} \leq \alpha \mid |W^{\pinv,r}|, p_{-j}^{\pinv,r}}
  & \leq \P{ p_j^{*\pinv,r} \leq \alpha \mid |W^{\pinv,r}|, p_{-j}^{\pinv,r} }
   = \P{ p_j^{*\pinv,r} \leq \alpha}
   \leq \alpha.
\end{align*}
Above, the last inequality follows from the fact that $p_j^{*\pinv,r}$ is the standard randomized binomial p-value~\cite{neyman1933ix}.
\end{proof}

If $r=E$, the p-value $p_j^{\pinv,E}$ in~\eqref{eq:pici-pvalues} does not become identical to the $p_j^{\inv}$ defined earlier in~\eqref{eq:ici-pvalues}. The difference between $p_j^{\inv}$ and $p_j^{\pinv,E}$ is that the former is always equal to~1 when $n^0_j>0$ and it is not randomized.
This discrepancy is due to expository convenience, as we prefer to keep the notation simple in the previous section, and it is practically irrelevant because it would not make sense to reject $\mathcal{H}^{\ici}_{j}$~\eqref{eq:null-ici} if $n^0_j>0$. While randomization is useful here because it allows us to deal as powerfully as possible with the case in which $n^0_j>0$ (it might make sense to reject $\mathcal{H}^{\pici,r}_{j}$ even if $n^0_j>0$, as long as $r<E$), it would have not helped in the previous section.
In fact, the p-values $p_j^{\inv}$ in~\eqref{eq:ici-pvalues} only have one bit of information, and selective SeqStep+ only cares about whether they are greater than $1/2$.

If $r=1$, the hypothesis $\mathcal{H}^{\pici,1}_{j}$~\eqref{eq:null-pici} states $X_j$ is conditionally associated with $Y$ in at least one environment. In that case, a traditional knockoff analysis of the pooled data would test the same hypothesis, and it might be more powerful because it allows more flexibility in the test statistics. Therefore, we only recommend applying the method proposed here with $r>1$.

A limitation of selective SeqStep+ is that it is unclear how to choose its parameter $c$ (the baseline rejection threshold)~\cite{barber2015controlling} to maximize power because the p-values in~\eqref{eq:pici-pvalues} can take several possible values; by contrast, $c=1/2$ is the only option for one-bit p-values. Within our partial consistency problem, different $c$ may result in higher power depending on the data.
Therefore, we also consider filtering $p_j^{\pinv,r}$~\eqref{eq:pici-pvalues} with the accumulation test~\cite{li2017accumulation}, which requires less tuning.
If the p-values are independent, this test controls a modified version of the false discovery rate,
\begin{align} \label{eq:mfdr}
  \text{mFDR}_{q} \defeq \E{\frac{| \{j :\mathcal{H}^{\mathrm{ci}, e}_{j} \text{ is rejected}\} \cap \{j : \mathcal{H}^{\mathrm{ci}, e}_{j} \text{ is true}\} |}{|\{j :\mathcal{H}^{\mathrm{ci}, e}_{j} \text{ is rejected}\} | + q }},
\end{align}
for some constant $q$ whose value will be specified later.
If $q$ is not too large and the discoveries are sufficiently numerous, the above definition is not very different from the false discovery rate.
Although our p-values $p_j^{\pinv,r}$ in~\eqref{eq:pici-pvalues} are dependent, the next result proves the accumulation test is still valid if we add a little extra randomness.

Starting from a matrix $\mathbf{W}$ of multi-environment knockoff statistics (Definition~\ref{def:kenv-stats}), randomly assign imaginary positive or negative signs to any zero entry by flipping independent fair coins. Then, define $n^-_j$ as the number of negative entries in the $j$-th column of the resulting {\em tie-breaking} $\mathbf{W}$, and compute, with the same notation as in~\eqref{eq:pici-pvalues},
\begin{align}  \label{eq:pici-pvalues-coin}
  p_j^{\pinv,r}
  & \defeq \Psi\left( E-r+1, \frac{1}{2}, n^-_j -1 \right) + U_j \cdot \psi\left( E-r+1, \frac{1}{2}, n^-_j \right).
\end{align}

\begin{restatable}[Multi-environment knockoff filter with accumulation test]{theorem}{theoremaccumulation} \label{th:acc-fdr}
The accumulation test of~\cite{li2017accumulation} with an increasing accumulation function (e.g., HingeExp with parameter $C=2$) applied to the p-values defined in~\eqref{eq:pici-pvalues-coin} controls the modified false discovery rate~\eqref{eq:mfdr} (e.g., with $q=C/\alpha$), as in~\cite{li2017accumulation}. That is, Theorem~1 of~\cite{li2017accumulation} still holds for the p-values $p_j^{\pinv,r}$ in~\eqref{eq:pici-pvalues-coin} even though they are not independent.
\end{restatable}

The proof of Theorem~\ref{th:acc-fdr} is in Appendix~\ref{app:proof-acc}.
It is worth emphasizing Theorem~\ref{th:seqstep-fdr} can accommodate replacing the p-values $p_j^{\pinv,r}$ from~\eqref{eq:pici-pvalues} with those in~\eqref{eq:pici-pvalues-coin}, although the converse does not hold for Theorem~\ref{th:acc-fdr}. The issue is that $n_j^0$ may vary arbitrarily across $j$, breaking the symmetry needed by our martingale proof. This may be a limitation of our proof, as the simulations in Section~\ref{section:experiments} suggest the accumulation test works well with the p-values in~\eqref{eq:pici-pvalues}.

\section{Consistent genome-wide associations across diverse ancestries} \label{sec:gwas}

Before investigating the effectiveness of the proposed methods in practice, we present in some detail a genetic problem that has motivated our work; this will inform the design of our simulations and will be the subject of the subsequent data analysis.
Genome-wide association studies aim to identify genetic variants with biological effects on some phenotype. 
Because the DNA of  individuals is determined prior to any of their phenotypes, it is relatively easy to attribute a causal interpretation to observed links between genetic variation and traits. Furthermore, it is reasonable to assume the same biological pathways are involved in influencing the traits across different ethnicities or human populations: modulo the expected variations in ``environment'', it is therefore meaningful to think about one common causal mechanism linking genetic variation to traits and  to attempt to uncover it by testing $\mathcal{H}^{\ici}_{j}$~\eqref{eq:null-ici}. We now describe the variables observed in genome-wide association studies and the changes expected across environments corresponding to different human populations.

\subsection{Missing variants and knockoffs in genome-wide association studies}

Genome-wide association studies are carried out  in practice by measuring ({\em genotyping}) a subset of a few hundred thousands variants across the genome, as full sequencing is expensive.
Such relatively few markers are sufficient to capture most genetic variation because nearby alleles on the same chromosome have strong dependencies and can thus be quite accurately inferred from one another. This property, known as {\em linkage disequilibrium}~\cite{slatkin1994linkage,pritchard2001linkage}, facilitates the localization of broad regions, or {\em loci}, containing associations with the phenotype, but at the same time it complicates the attribution of distinct signals to specific variants. Indeed, many genotypes can be {\em marginally associated} with the phenotype simply because they are in linkage disequilibrium with the same causal variant; this issue is alleviated by a conditional testing approach, but even conditional associations do not account for possible confounding due to unobserved variants.

The traditional analysis of genome-wide association data imputes the missing variants using models of linkage disequilibrium estimated on smaller, fully-sequenced, reference samples~\cite{marchini2010}. Then, the imputed variants are analyzed alongside the typed ones to localize significant associations, through either genome-wide marginal tests or multivariate linear models within narrow genomic regions~\cite{schaid2018genome}.
However, this is not fully satisfactory because imputation is not as informative as a direct measurement; in fact, imputed variants carry no information in addition to that contained in the typed ones, as they are a function of the latter, conditionally independent of the phenotype.
This may pass unnoticed if one fully trusts a multivariate linear model for the outcome---imputed variants may be significant within such models because their dependence with the measured variants is nonlinear---but it makes it impossible to find explicit evidence that a missing variant is causal within our nonparametric model-X perspective~\cite{sesia2020multi,sesia2020controlling}.
Without repeating the arguments supporting a model-X analysis of genome-wide association data, which were explained in~\cite{sesia2018} and expanded in~\cite{sesia2020multi,sesia2020controlling}, we will focus here on leveraging consistency to gather indirect evidence of causal associations within this framework.
First though, we must briefly recall some relevant technical details of the current methodology.

The existing knockoff-based analysis partitions the genome into contiguous segments and then tests whether any of these include variants conditionally associated with the phenotype~\cite{sesia2020multi,sesia2020controlling}.
Let $G \subset \{1,\ldots,p\}$ denote the indices of the typed variants in a particular segment.
Then, knockoffs can be utilized to test a slightly more general version of the conditional hypothesis $\mathcal{H}^{\mathrm{ci}, e}_{j}$~\eqref{eq:null-ci} for each environment (sub-population) $e \in \{1,\ldots,E\}$; namely,
\begin{align} \label{eq:null-ci-grouped}
 \mathcal{H}^{\mathrm{ci}, e}_{G} : Y^{e} \indep X^{e}_{G} \mid X^{e}_{-G},
\end{align}
where $X_G = \{X_j : j \in G\}$, $X_{-G} = \{X_j : j \not\in G\}$.
That is, $\mathcal{H}^{\mathrm{ci}, e}_{G}$~\eqref{eq:null-ci-grouped} asserts that the variants in group $G$ are conditionally independent of the phenotype given all other measured variants.
This analysis can be performed at different levels of resolution, separately controlling the false discovery rate for increasingly refined genomic partitions to balance between power and the value of each discovery~\cite{sesia2020multi}.
In short, hypotheses $\mathcal{H}^{\mathrm{ci}, e}_{G}$~\eqref{eq:null-ci-grouped} involving smaller groups of genotypes (higher resolution) are harder to reject because the variables have strong local dependencies, making it difficult to distinguish the signal of one variant from those of its neighbors. At the same time, high-resolution hypotheses are more specific and thus their rejections more informative.
Although $\mathcal{H}^{\mathrm{ci}, e}_{G}$~\eqref{eq:null-ci-grouped} is not asking whether a physical genetic segment contains causal variants, it is a reasonable statistical approximation of that scientific question.
Intuitively, we expect these hypotheses to be more robust to missing variants at low resolution; this will be verified empirically later.
By contrast, the robustness of $\mathcal{H}^{\mathrm{ci}, e}_{G}$~\eqref{eq:null-ci-grouped} is less clear at high resolution because there each tested segment contains few measured genotypes; this is where consistency will be most useful.

\subsection{Linkage disequilibrium in populations with different ancestries}

Linkage disequilibrium is explained by the inheritance of long and randomly cut genetic segments from parents to offspring, with occasional mutations. Generation after generation, the genotype distribution thus comes to resemble an imperfect mosaic of motifs inherited from the common ancestors, which can be encoded as a hidden Markov model~\cite{li2003}; this is at the heart of the imputation methods mentioned in the preceding section~\cite{marchini2010}, as well as of the algorithms for generating knockoffs~\cite{sesia2018}.
The block-like patterns of linkage disequilibrium vary across human sub-populations because these share different recent ancestors, and so their mosaics involve different patterns, and possibly also different transition ({\em recombination}) rates~\cite{slatkin1994linkage,laan1997demographic,daly2001high}.
In other words, different sub-populations differ by covariate shift.
This heterogeneity has already been factored into the generation of valid knockoffs to test {\em pooled} conditional associations~\cite{sesia2020controlling}, and it will be leveraged here to help highlight causal variants through consistency.

Figure~\ref{fig:diagram-gwas} visualizes why covariate shift helps localize causal variants within an example with three sub-populations, four observed and five missing variables, one of which is causal. This toy genome is partitioned into segments containing one or two typed variants each; three segments are highlighted. Linkage disequilibrium is described by hidden Markov models yielding blocks of variables separated by high-recombination spots~\cite{wall2003haplotype,slatkin2008linkage}. Treating alleles across these spots as approximately independent, we can see that the only consistent association is that of the segment containing the causal variant.
Of course, reality is more complicated. First, linkage disequilibrium is not perfectly organized into independent blocks, although this is a common simplification~\cite{berisa2016approximately}. Second, we can only study a finite number of human sub-populations, so it may be unrealistic to expect all spurious associations to be removed. Nonetheless, consistency enables a step forward in an otherwise challenging problem, and we will verify empirically that our approach is indeed useful.

\begin{figure}[!htb]
\centering
    \includegraphics[width=0.8\linewidth]{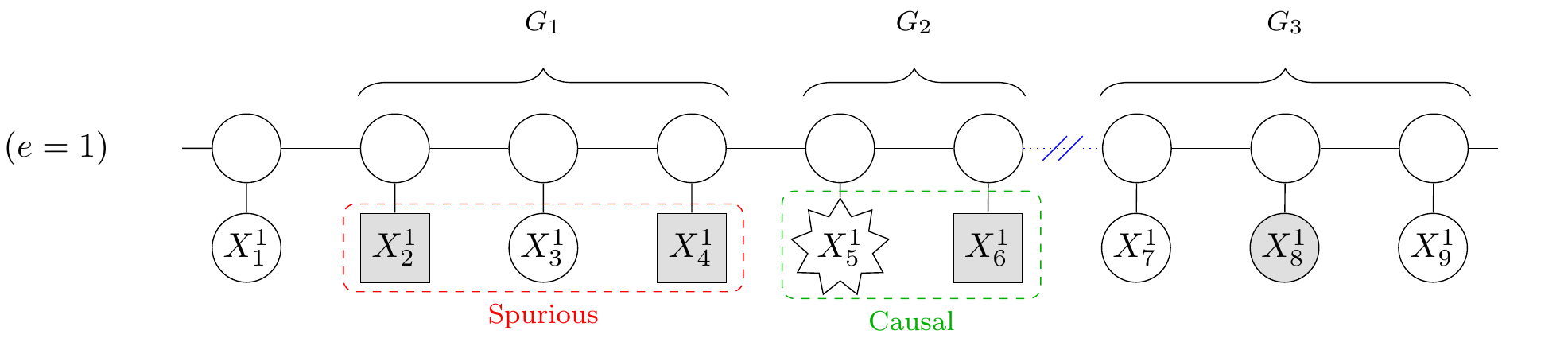}
    \includegraphics[width=0.8\linewidth]{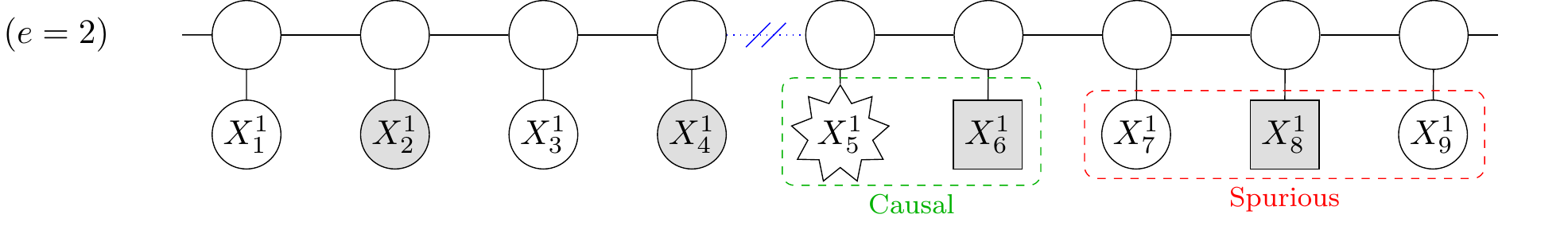}
    \includegraphics[width=0.8\linewidth]{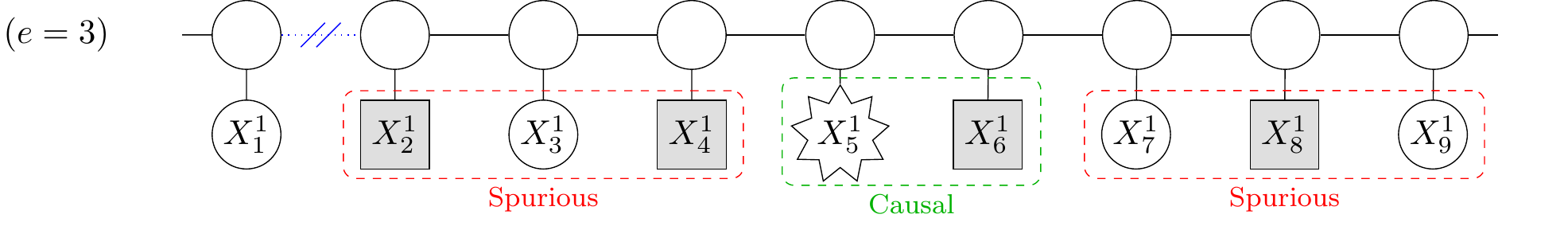}
     \caption{Schematic visualization of consistency in a genome-wide association study involving three sub-populations. The unobserved causal variant (star-shaped node) induces different spurious associations depending on the patterns of linkage disequilibrium, described by population-specific hidden Markov models. Shaded nodes: measured variables. Squares: variables that are not independent of the causal one conditional on the other measured variables. The broken segments in the Markov chain symbolize the boundaries of linkage disequilibrium blocks within a population.}
     \label{fig:diagram-gwas}
\end{figure}

\section{Numerical experiments} \label{section:experiments}

We apply the multi-environment knockoff filter with the data-splitting and empirical cross-prior alternative statistics.
The benchmarks are the two heuristics from Section~\ref{sec:prob-statement}: intersection and pooling.
All test statistics are computed with the R package \texttt{glmnet}~\cite{friedman2010regularization}, or \texttt{bigstatsr}~\cite{prive2019efficient} for genetic data.
Software for our method is available from \url{https://github.com/lsn235711/MEKF_code}, along with code to reproduce the analyses.

\subsection{Testing for full consistency with synthetic data} \label{section:simulation_invariant}

In each environment, $p$ variables $X$ are generated from an autoregressive model of order one with correlation $\rho =0.2$. We will leverage the known $P_X$ to construct exact knockoffs with the Gaussian semi-definite optimization algorithm from~\cite{candes2018panning}. Knockoffs are typically quite robust to model misspecification, and the estimation of $P_X$ is orthogonal to our problem.
The distribution of $Y^e \mid X^e$ in the $e$-th environment is given by a logistic model with log-odds equal to $\text{logit } \P{Y^{e} = 1 \mid X^{e}} = X^{e} \beta^{e}$,
where $\beta^e \in \mathbb{R}^p$ is an environment-specific effect parameter vector.
We consider two settings corresponding to different numbers of environments $E$, explanatory variables $p$, and effect vectors $\beta^e$.

In the first setting, $E=4$, $p=500$, and the number of observations per environment is $n=1000$. First, 100 entries of $\beta^e$ are randomly chosen to be non-zero in all environments and these are the consistent associations we seek; then, for each environment $e$, 10 of the remaining ones are set to be non-zero in all but the $e$-th environment, and these 40 associations are thus not consistent. See Figure~\ref{fig:inv_beta}~(a) in Appendix~\ref{sec:app-synthetic} for a visualization of this setup. The absolute values of the 100 consistent non-null elements of $\beta^e$ are equal to $a / \sqrt{n}$, where $a$ is a signal amplitude parameter which we will vary, while the remaining non-zero values are equal to $0.5 a / \sqrt{n}$. The signs of non-null elements of $\beta^e$ are determined by independent coin flips.
In the second setting, $E=3$, $p=200$, and $n=2000$. The coefficients $\beta^e$ are determined as follows. First, 50 entries of $\beta^e$ are randomly chosen to be non-zero in all environments; then, for each $e$, 50 of the remaining ones are set to be non-zero in all but the $e$-th environment; see Figure~\ref{fig:inv_beta}~(b) in Appendix~\ref{sec:app-synthetic}. Again, the 100 consistent non-null elements of $\beta^e$ are equal to $a / \sqrt{n}$ in absolute value, and the remaining non-zero entries are $0.5 a / \sqrt{n}$. The signs of $\beta^e$ are determined by independent coin flips.

Our goal is to discover the subset of consistent non-nulls, controlling the false discovery rate below $10\%$.
Figure~\ref{fig:inv_sim}, previewed earlier, compares the performance of our method to those of the benchmarks, averaging over 100 experiments.
The results confirm our method controls the false discovery rate, as anticipated by the theory. The cross-prior statistics are more powerful than the data-splitting ones in the first setting, in which most associations are consistent, and equivalent to the latter in the second setting, where most associations are not consistent. Pooling yields too many false discoveries because it reports all associations regardless of whether they are consistent, while the intersection heuristic may lead to either low power (first setting) or high false discovery rate (second setting).

% \begin{figure}[!htb]
%  \centering
%  \includegraphics[width=\linewidth]{figures_matteo/experiment_invariant.pdf}
%  \caption{Performance of the multi-environment knockoff filter (MEKF) on simulated data, with two alternative statistics. The performance is compared to that of two heuristic benchmarks. The nominal false discovery rate is 0.1 (dashed horizontal line). (a) Synthetic data in which most conditional associations are consistent. (b) Synthetic data in which most conditional associations are not consistent. The results are averaged over 100 independent experiments.}
%  \label{fig:inv_sim}
% \end{figure}

\subsection{Testing for partial consistency with synthetic data} \label{section:simulation_partial_invariant}

We consider partial consistency testing as in Section~\ref{sec:partial-consistency}.
The goal is to discover which variables are non-null in at least $r=2$ out of $E=5$ environments. The {\em intersection} benchmark reports all discoveries found in at least $r$ environments by separate analyses.
The data are generated from the same model as in the previous section, although with $p=200$ variables and utilizing different model parameters~$\beta$. The number of samples per environment is $n=600$.
Again, we consider two settings with alternative $\beta$ patterns; see Figure~\ref{fig:pinv_beta} in Appendix~\ref{sec:app-synthetic}.
In the first setting, 50 entries of $\beta^e$ are randomly chosen to be non-zero in all environments and these are the consistent associations we seek; then, 5 unique additional variables are set to be non-zero in each environment; see Figure~\ref{fig:pinv_beta}~(a).
The absolute values of all non-zero elements of $\beta^e$ are $a / \sqrt{n}$, where $a$ is a signal amplitude parameter which we will vary. The signs of non-null elements of $\beta^e$ are determined by independent coin flips.
In the second setting, 100 entries of $\beta^e$ are randomly chosen to be non-zero in the first four environments; then, the remaining 100 variables are set to be non-zero in the last environment; see Figure~\ref{fig:pinv_beta}~(b).
The values of the non-zero $\beta^e$ coefficients are determined as in the first setting.

The multi-environment knockoff filter is applied with the empirical cross-prior statistics. The p-values are filtered either by the accumulation test with HingeExp accumulation function, in which case any zero-sign ties are broken as in~\eqref{eq:pici-pvalues}, or with the selective SeqStep+ filter with cutoff parameter $c=0.5$. The nominal false discovery rate is 10\%.
Figure~\ref{fig:pc_sim} compares the performance of our method to that of the two heuristic benchmarks, as in the previous section.
The results show our method always controls the false discovery rate when applied with the selective SeqStep+ filter, and typically also does so in combination with the accumulation test.
Even though the accumulation test theoretically requires p-values with random tie breaking~\eqref{eq:pici-pvalues-coin}, Figure~\ref{fig:pc_sim_coin} shows the original ones in~\eqref{eq:pici-pvalues-coin} often lead to higher power, without noticeable loss of type-I error control.
Figure~\ref{fig:pc_sim} highlights that pooling is overly liberal, while the intersection heuristic may be either underpowered (setting 1) or too liberal (setting 2).

\begin{figure}[!htb]
 \centering
 \includegraphics[width=\linewidth]{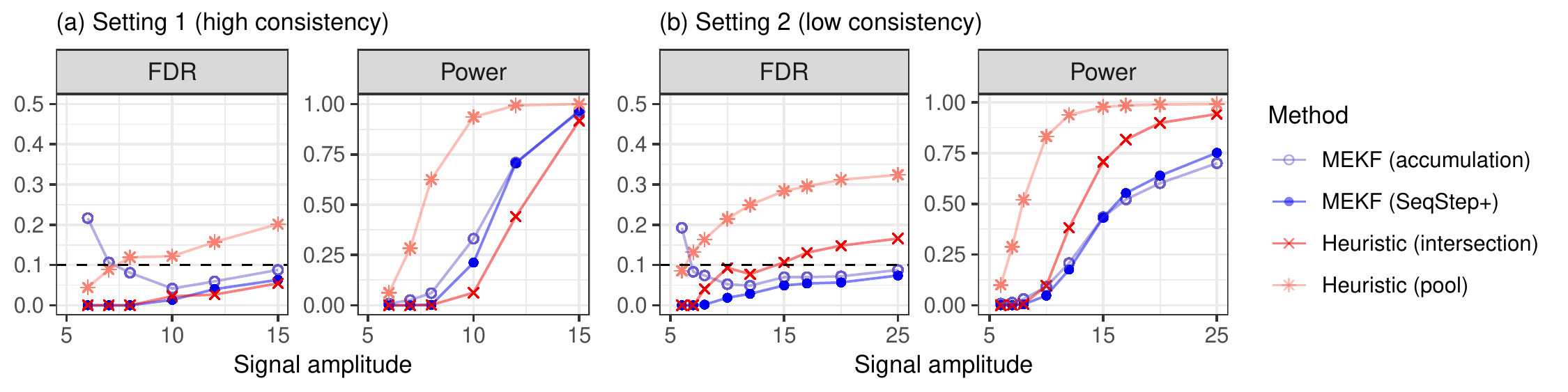}
 \caption{Performance of the multi-environment knockoff filter (MEKF) on simulated data, compared to two heuristics. The goal is to discover which variables are non-null in at least 2 out of 5 environments. Other details are as in Figure~\ref{fig:inv_sim}. Note that the accumulation version of our method is only guaranteed to control a modified version of the false discovery rate, which differs noticeably from our true objective only if the number of discoveries is very small.}
 \label{fig:pc_sim}
\end{figure}

\subsection{Causal inference with biased data}  \label{sec:sim-sampling}

We revisit the school admission example of Section~\ref{sec:robust-bias} by considering a population in which the $p= 200$ features $X \in \mathbb{R}^p$ of each applicant are generated from an autoregressive model of order one, with $\rho^e = 0.6 - 0.1e$ in the $e$-th school (environment), for $e \in \{1,2\}$.
A constant causal model describes the conditional GPA distribution: $Y = X \beta + \epsilon$, where $\epsilon$ is i.i.d.~Gaussian noise, and $\beta \in \mathbb{R}^p$ has 25 non-zero entries and is constant. The non-zero entries of $\beta$ are equal to $2/15$. The sample size $n$ is varied as a control parameter.
Our goal is to make inferences about the causal model by analyzing data collected through environment-specific biased sampling mechanisms.
In particular, a random sample from the population is observed in the $e$-th environment if and only if
$Y_i/\sigma + \sum_{i \in S_e} X_i > 0$,
for two subsets $S_1,S_2 \subset \{1,\ldots,p\}$ with $|S_1| = |S_2| = 25$, and $\smash{\sigma = 2\sqrt{2}}$.
This is a much more ambitious objective compared to Sections~\ref{section:simulation_invariant}--\ref{section:simulation_partial_invariant}. In fact, our method explicitly tests the consistency hypotheses $\mathcal{H}^{\ici}_{j}$~\eqref{eq:null-ici}, while the validity of its inferences about $\mathcal{H}^{\mathrm{causal}}_{j}$~\eqref{eq:null-causal} depends on the coherence across environments of the biased sampling mechanisms. Therefore, we anticipate our causal inferences will be valid only if $S_1$ and $S_2$ are sufficiently different from one another.

We compare our method applied with data-splitting statistics to the \emph{pool} and {\em intersection} benchmarks.
The knockoffs are based in each environment on the feature distribution estimated from a much larger sample with the same bias as the data we analyze; this separates the problem of constructing knockoffs from that of testing for consistent associations.
Figure~\ref{fig:biased_sam}(a) compares power and false discovery rate, defined in terms of the causal hypotheses $\mathcal{H}^{\mathrm{causal}}_{j}$~\eqref{eq:null-causal}, for disjoint $S_1,S_2$ as a function of the sample size $n$. Figure~\ref{fig:biased_sam}(b) plots analogous results with $n=1200$, as a function of the overlap between $S_1$ and $S_2$, which ranges from 0\% (disjoint) to 100\% (identical).
The results show our causal inferences are valid if the overlap between $S_1$ and $S_2$ is small; if it is large, the spurious associations due to sampling bias also become consistent. The intersection heuristic is not as powerful as our method here, while pooling leads to many more non-causal discoveries because it does not account for the sampling biases at all.

\begin{figure}[!htb]
\centering
  \includegraphics[width=\linewidth]{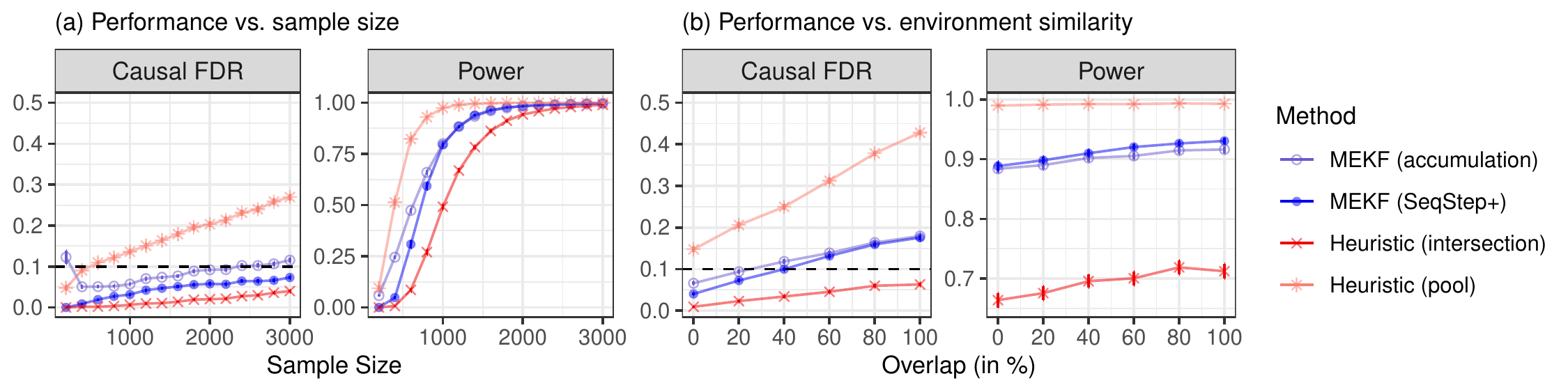}
  \caption{Performance of the multi-environment knockoff filter in a simulated study with environment-specific sample biases. The false discovery rate and power are defined in terms of causal hypotheses. Other details are as in Figure~\ref{fig:pc_sim}.}
  \label{fig:biased_sam}
\end{figure}

\subsection{Causal inference with dependent data} \label{sec:sim-homophily}

Imagine a population in which the $p=200$ features $X \in \mathbb{R}^p$ of each individual are generated from an autoregressive model of order one, with $\rho^e = 0.6 - 0.1e$ in the $e$-th environment, for $e \in \{1,2\}$.
The causal model is constant: $Y = X \beta + \epsilon$, where $\epsilon$ is standard i.i.d.~Gaussian noise and $\beta \in \mathbb{R}^p$ has 40 non-zero entries equal to 0.25.
Our goal is to identify the causal variables; however, we cannot collect independent samples from this population.
As in the example from Section~\ref{sec:robstness-homophily}, we say an individual $i$ belongs to a club if $Y^{(i)}>0$ (homophily).
For any $i$ belonging to a club in the $e$-th environment, we observe modified features $X_{j}^{(i)} \leftarrow X_{j}^{(i)} + I_j^e$  (contagion) for all $j \in S^e \subseteq \{1,\ldots,p\}$,
where each $I_j^e \in \{0,1\}$ is an independent coin flip shared by all individuals in environment $e$. We fix $|S_1| = |S_2| = 25$.

Figure~\ref{fig:corr_group}(a) compares the performance of our method, with data-splitting statistics, to that of the usual two benchmarks, as a function of the sample size $n$ in each environment. The nominal false discovery rate is 10\%. Here, $S_1,S_2$ are randomly chosen and disjoint.
The knockoffs are based on larger sets of independent samples, as in the previous section.
Again, we define the false discovery rate and power in terms of $\mathcal{H}^{\mathrm{causal}}_{j}$~\eqref{eq:null-causal}; therefore, the causal validity of our inferences is only guaranteed if $S_1 \cap S_2 = \emptyset$.
Indeed, here the multi-environment knockoff filter is more powerful than the intersection heuristic and controls the causal false discovery rate, unlike pooling.
Figure~\ref{fig:corr_group}(b) compares the performances of these methods as a function of $|S_1 \cap S_2|$, fixing $n=600$.
This shows the multi-environment knockoff filter yields valid causal inferences if two environments are sufficiently different.

\begin{figure}[!htb]
\centering
  \includegraphics[width=\linewidth]{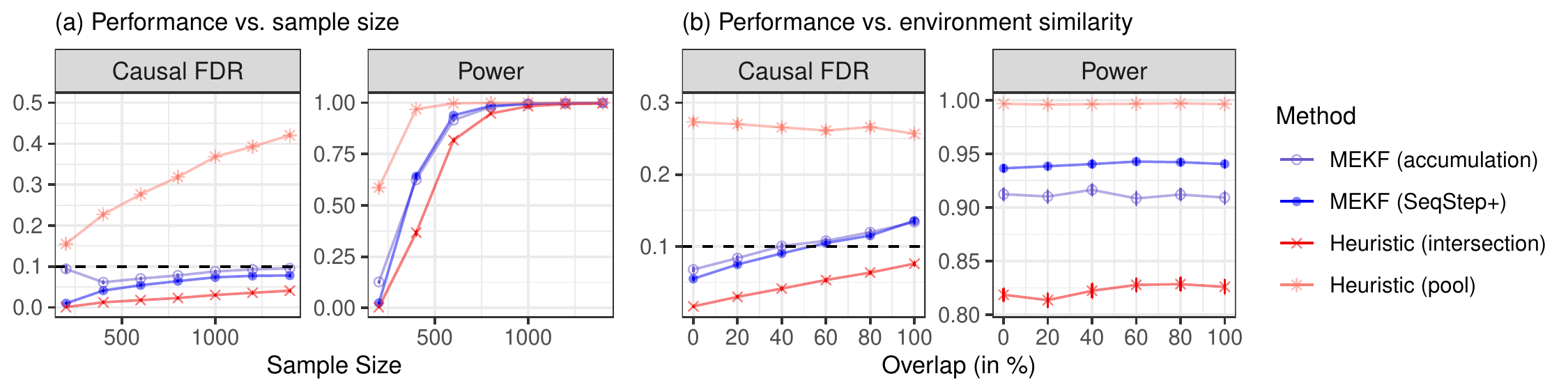}
  \caption{Performance of the proposed method in a simulated study with environment-specific homophily and contagion. Other details are as in Figure~\ref{fig:pc_sim}. }
  \label{fig:corr_group}
\end{figure}

\subsection{Causal inference in a simulated genome-wide association study} \label{sec:sim-gwas}

We analyze simulated yet realistic genetic association data involving different sub-populations, based on the haplotypes in the 1000 Genomes Project, Phase 3~\cite{10002015global}. This resource contains phased haplotypes for individuals belonging to one of five possible sub-populations: African (AFR), Admixed American (AMR), East Asian (EAS), European (EUR), and South Asian (SAS).
We utilize these real haplotypes to simulate genetic data from a hidden Markov model for 50,000 individuals belonging to one of the five sub-populations (10,000 individuals per sub-population); see Appendix~\ref{sec:app-gwas} for details.
This approach ensures the genotype distribution is known exactly, allowing us to focus fully on the problem of accounting for missing variants.
In particular, we can apply the algorithm from~\cite{sesia2020multi} to generate perfectly valid knockoffs separately within each sub-population, without having to estimate a hidden Markov model~\cite{sesia2018} or account for population structure~\cite{sesia2020controlling}.
We construct group-level knockoffs~\cite{sesia2020multi} for testing $\mathcal{H}^{\mathrm{ci}, e}_{G}$~\eqref{eq:null-ci-grouped} at two levels of resolutions, with genetic segments of median lengths 233~kb and 15~kb, respectively.
In the interest of time, we only analyze 359,811 biallelic single-nucleotide polymorphisms on chromosome 22 rather than the full genome.

Conditional on the genotypes, we simulate a continuous trait for all 50,000 individuals from a linear model with independent Gaussian errors and 50 causal variants.
This model is constant across all populations (environments).
The causal variables are randomly chosen among the 359,811 possible genotypes, ensuring each sub-population has at least 10 causal variants with minor allele frequency above 0.1. The signs of the causal effects are independent coin flips, and their sizes are inversely proportional to the standard deviation of the allele count, so that rarer variants have larger effects. The total heritability of the trait is varied as a control parameter.
All causal variants are unmeasured, so that their exact identification from the data is impossible; instead, the goal is to localize as precisely as possible which genetic segments are likely to contain causal variants~\cite{sesia2020multi}, controlling the false discovery rate.
The proportion of typed variants is varied between 1\% and 10\%. In each case, we construct knockoffs only for the measured variants.
%Because the number of causal variants in this example is not very large compared to what we would expect in a real genome-wide association study of a polygenic trait, a lower nominal false discovery rate may result in discoveries that are more affected by the randomness in the knockoffs, and hence less clearly interpretable~\cite{candes2018panning}.
This setup is likely to be even more challenging than a real genome-wide association study from the point of view of missing variants, because our genotyping is completely random. By contrast, genotyping arrays in real studies are carefully designed. For example, the UK Biobank~\cite{bycroft2018} data were collected using the UK Biobank Axiom\texttrademark Array, which specifically targets potentially causal coding variations, genomic regions of interest, and markers known to be associated with various phenotypes~\cite{bycroft2018}.
Thus, confounding may be a less severe problem in practice compared to what we shall experience here.

We perform (consistent) conditional testing at the 10\% nominal false discovery rate level but measure performance in stricter terms, based on the causal false discovery rate and power: a discovery is counted as true if and only if it reports a genetic segment containing a causal variant. The power is defined as the average proportion of segments containing causal variants that are discovered.
All results are averaged over 100 experiments with independent traits. In theory, the genotypes should also be resampled to ensure false discovery rate control because the knockoffs treat them as random; however, that would be computationally expensive with such large data.

Figure~\ref{fig:sim-missing-resolution}~(a) summarises the results of separate analyses in each sub-population as a function of the trait heritability, in the case in which only 1\% of all variants are observed. These analyses do not lead to valid causal inferences, although they correctly test conditional association, demonstrating the need for consistency especially at high resolution.
The multi-environment knockoff filter is applied to test whether any associations are significant in at least 3 environments, utilizing the data-splitting statistics  due to the large size of this data set. Statistical significance is determined either with the accumulation test or with the selective SeqStep+ approach.
Our method is compared to the {\em intersection} and {\em pool} benchmarks.
The results indicate the multi-environment knockoff filter controls the causal false discovery rate and, when applied with the accumulation test, is only slightly less powerful than pooling at low resolution. The selective SeqStep+ approach tends to yield lower power, plausibly because it cannot extract as much information from the p-values.
 At higher resolution, our method is not as powerful as pooling, while the causal false discovery rate inflation of the latter becomes more severe.
Unsurprisingly, all methods are less powerful at high resolution.
The causal false discovery rate violation of the intersection heuristic is smaller but noticeable.

\begin{figure}[!htb]
  \centering
  \includegraphics[width=0.9\linewidth]{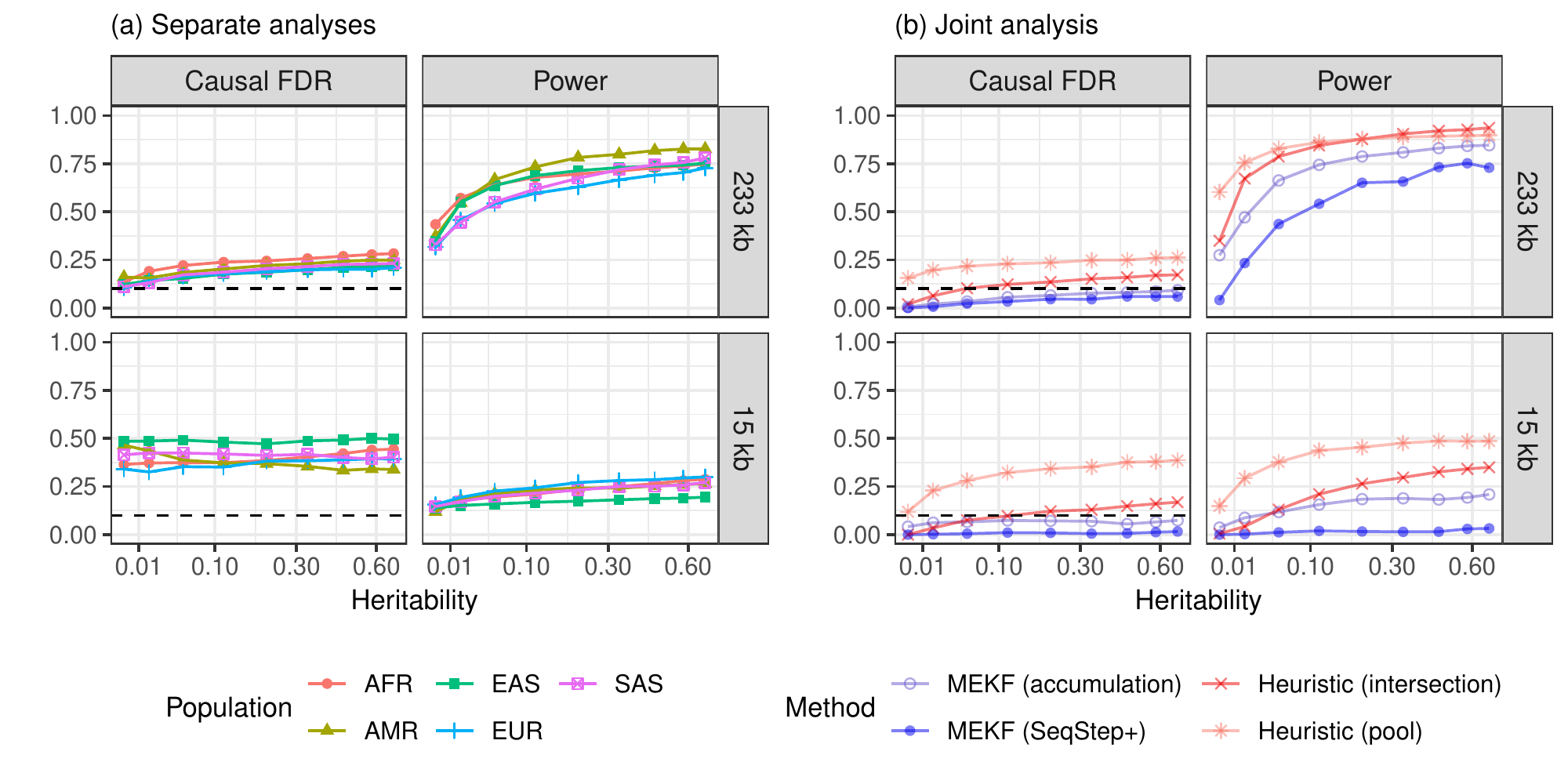}
  \caption{Analysis of a simulated multi-population genome-wide association study in which the causal variants are missing. The genotyping density is 1\%. Top: low-resolution analysis (233 kb); bottom: high-resolution analysis (15 kb). The empirical false discovery rate and power are defined in a strict causal sense. The multi-environment knockoff filter seeks associations supported by the data from at least 3 populations. The nominal false discovery rate is~10\%.}
  \label{fig:sim-missing-resolution}
\end{figure}

Figure~\ref{fig:sim-missing-density} in Appendix~\ref{sec:app-gwas} shows confounding due to unmeasured variants decreases as the genotying density increases, until all methods control the causal false discovery rate. This is intuitive because there is no confounding in the limit of high genotyping density. Meanwhile, as the genotying density increases, the multi-environment knockoff filter becomes unnecessarily conservative compared to pooling~\cite{sesia2020multi,sesia2020controlling}; this may seem unavoidable but the results in Figure~\ref{fig:inv_sim} suggest the relative performance of our method would improve if we applied the cross-prior statistics from Section~\ref{sec:ici-knockoffs} instead of the data-splitting ones adopted here for convenience.
Figure~\ref{fig:sim-missing-resolution-0.2} shows qualitatively similar results corresponding to analyses at the 20\% nominal false discovery rate level, which better highlight the causal type-I error inflation incurred by the heuristics.
Finally, Figure~\ref{fig:sim-missing-resolution-coin} shows our method performs similarly regardless of whether the accumulation test is applied the p-values computed with random tie breaking~\eqref{eq:pici-pvalues-coin} or without it~\eqref{eq:pici-pvalues}.

\section{Analysis of UK Biobank genome-wide association data} \label{sec:ukb}

\subsection{Data pre-processing}

We study four continuous traits (body mass index, height, platelet count, and systolic blood pressure) and four diseases (cardiovascular disease, diabetes, hypothyroidism, and respiratory disease) using the UK Biobank~\cite{bycroft2018} data; see Table~\ref{table:pheno-def} in Appendix~\ref{sec:app-ukb} for more details.
This analysis is based on the same quality control filtering and knockoffs for 486,975 genotyped and phased subjects in the UK Biobank
(application 27837) as in~\cite{sesia2020controlling}. The knockoffs preserve both the population structure and the kinship of the 136,818 individuals with close relatives; this accounts for most possible confounders except missing variants~\cite{sesia2020controlling}. Our goal is to address this remaining limitation with the multi-environment knockoff filter.
As in previous work~\cite{sesia2020controlling}, we only analyze 591,513 biallelic single nucleotide polymorphisms with minor allele frequency above 0.1\% and in Hardy-Weinberg equilibrium ($10^{-6}$) among the subset of 350,119 unrelated British individuals previously analyzed in~\cite{sesia2020multi}.
The genome is then partitioned into contiguous groups at 7 levels of resolution, ranging from that of single polymorphisms to that of 425 kb-wide groups, as in~\cite{sesia2020controlling}. The resolution of each genomic partition we consider is defined as the median width of its genetic segments.

The UK Biobank subjects who passed the above quality control are divided into five sub-populations based on their self-reported ancestry (African: 7,635; Asian: 3,284; British: 429,934; Non-British European: 28,994; and Indian: 7,628). We exclude subjects with missing ancestry information, as well as those falling outside these five broad categories; this leaves us with a total of 477,475 individuals; see Table~\ref{tab:ukb-ancestries-env}, Appendix~\ref{sec:app-ukb}, for additional details.

\subsection{Searching for consistent associations}

We apply the multi-environment knockoff filter to discover genetic segments containing distinct association with the phenotype in at least $r$ environments, with $r$ ranging from 2 to 5.
In all cases, the significance threshold is computed by applying the accumulation test to the p-values in~\eqref{eq:pici-pvalues}. In fact, the accumulation test without the random tie breaking~\eqref{eq:pici-pvalues-coin} tends to be more powerful than selective SeqStep+ (Section~\ref{section:experiments}), and tie breaking seems practically unnecessary; see Figure~\ref{fig:pc_sim_coin} in Appendix~\ref{sec:app-experiments}.
The analysis is performed at the 10\% false discovery rate level, separately for each level of resolution~\cite{sesia2020multi}.
 We adopt the data-splitting statistics because the data set is very large.
The intersection heuristic and the pool analysis on all UK Biobank samples from~\cite{sesia2020controlling} will serve as benchmarks.

We repeat all tests with 100 independent realizations of the $U_j$ variables in~\eqref{eq:pici-pvalues}; this allows some understanding and a possible reduction of the variability of any findings, as our method is randomized.
Alternatively, one may repeat the entire analysis starting from the generation of the knockoffs~\cite{ren2020derandomizing}; however, that would be impractical for a data set of this size. In comparison, the cost of resampling the $U_j$ variables many times is negligible.
Table~\ref{table:ukb_discoveries} reports the numbers of discoveries for {\em height} and {\em platelet count} thus obtained in at least 51 out of 100 randomizations. The results for other phenotypes are in Table~\ref{table:ukb_discoveries_large}, Appendix~\ref{sec:app-ukb}, for lack of space. Unfortunately, there are fewer consistent associations for the other phenotypes, consistently with previous observations that {\em height} and {\em platelet count} display the strongest signals~\cite{sesia2020multi,sesia2020controlling}.
Our ``stability selection''~\cite{meinshausen2010stability} reporting rule is not theoretically guaranteed to control the false discovery rate~\cite{candes2018panning,sesia2018}; however, we can empirically confirm it to be conservative; see Figure~\ref{fig:sim-missing-resolution-stability} in Appendix~\ref{sec:app-experiments}.
Figure~\ref{fig:analysis_numdisc} summarises the variability of the individual findings corresponding to different p-value randomizations.

\begin{table}[!htb]
\center
\small
\caption{Numbers of discoveries at different resolutions for two UK Biobank phenotypes. The second column indicates the numbers of populations (environments) across which the findings are consistent. The third column corresponds to the analysis of the pooled data from all populations~\cite{sesia2020controlling}. The nominal false discovery rate is 10\%.}
\label{table:ukb_discoveries}

\begin{tabular}[t]{cc>{\centering\arraybackslash}m{0.7cm}>{\centering\arraybackslash}m{0.7cm}>{\centering\arraybackslash}m{0.7cm}>{\centering\arraybackslash}m{0.7cm}>{\centering\arraybackslash}m{0.7cm}}
\toprule
\multicolumn{2}{c}{ } & \multicolumn{5}{c}{Number of environments} \\
\cmidrule(l{3pt}r{3pt}){3-7}
Phenotype & Resolution (kb) & 1 & 2 & 3 & 4 & 5\\
\midrule
 & single-SNP & 95 & 13 & 2 & 0 & 0\\

 & 3 & 570 & 9 & 6 & 0 & 0\\

 & 20 & 1503 & 33 & 0 & 0 & 0\\

 & 41 & 2384 & 42 & 7 & 7 & 2\\

 & 81 & 3006 & 48 & 24 & 0 & 0\\

 & 208 & 3339 & 103 & 23 & 7 & 3\\

\multirow[t]{-7}{*}{\centering\arraybackslash height} & 425 & 3073 & 68 & 26 & 3 & 0\\
\cmidrule{1-7}
 & single-SNP & 53 & 9 & 3 & 0 & 0\\

 & 3 & 246 & 10 & 4 & 4 & 0\\

 & 20 & 1002 & 27 & 16 & 2 & 0\\

 & 41 & 1261 & 52 & 12 & 9 & 0\\

 & 81 & 1570 & 104 & 15 & 8 & 0\\

 & 208 & 1743 & 98 & 16 & 14 & 2\\

\multirow[t]{-7}{*}{\centering\arraybackslash platelet} & 425 & 1653 & 119 & 9 & 11 & 0\\
\bottomrule
\end{tabular}

\end{table}

Several consistent associations are discovered, although the power seems lower compared to pooling~\cite{sesia2020controlling}.
This is unsurprising, especially if $r>2$, because the sample sizes are imbalanced: most individuals have either British or other European ancestry.
Figure~\ref{fig:chicago_biobank} visualizes some discoveries for platelet count through a Chicago plot~\cite{sesia2020multi}, highlighting in different colors the numbers of environments across which the findings are consistent.
It is not guaranteed that all discoveries corresponding to a fixed $r \in \{2,\ldots,E\}$ are also found with $r'<r$, although this occurs often; see Figure~\ref{fig:chicago_biobank_stacked} in Appendix~\ref{sec:app-ukb}.
Table~\ref{table:ukb_discoveries_methods} in Appendix~\ref{sec:app-ukb} summarises the findings obtained with selective SeqStep+ instead of the accumulation test, as well those obtained with the intersection heuristic.
Clearly, it cannot be determined from Table~\ref{table:ukb_discoveries_methods} which approach is most effective at causal inference because the ground truth is unknown. Therefore, we will seek more evidence in support of our findings using prior domain knowledge.

\begin{figure}[!htb]
\centering
\includegraphics[width = 0.8\textwidth]{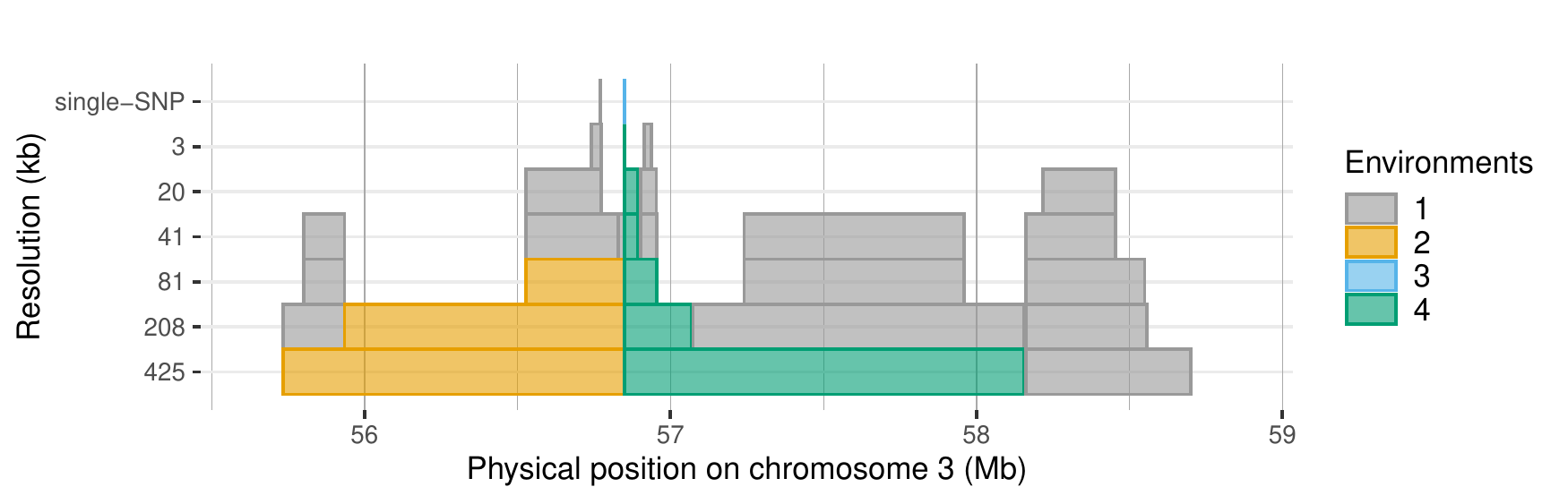}
\caption{Chicago plot of some discoveries on chromosome three for platelet count based on UK Biobank data from individuals in five sub-populations (environments). Each block represents a genetic segment containing distinct associations; the colors indicate the numbers of environments across which they are consistent. The vertical position denotes the resolution of the discovery measured in millions of base pairs (Mb). Other details are as in Table~\ref{table:ukb_discoveries}.}
\label{fig:chicago_biobank}
\end{figure}

\subsection{Validation of genetic findings}

Table~\ref{table:ukb_discoveries_confirm_proportion} demonstrates almost all of our consistent discoveries for height and platelet count are confirmed by the NHGRI-EBI GWAS Catalog~\cite{buniello2019nhgri} (accessed on April 15, 2021). We say that a discovered genetic segment is confirmed if it spans a genomic region containing reported associations for the same phenotype. Relatively fewer discoveries obtained by pooling~\cite{sesia2020controlling} are thus confirmed. Of course, this is not fully conclusive because the GWAS Catalog may include spurious associations and is likely to miss many causal ones, although it is a standard reference.
Table~\ref{table:ukb_discoveries_confirm_proportion} summarises the numbers of findings obtained with the intersection heuristic, as well as the proportions of those which are confirmed by the GWAS Catalog. This shows the intersection heuristic yields either fewer discoveries, or a (slightly) lower validation rate.
This is consistent with our simulations suggesting this heuristic is often either underpowered or excessively liberal. Analogous information for the other phenotypes is in Table~\ref{table:ukb_discoveries_confirm_proportion_others}, Appendix~\ref{sec:app-ukb}.
Table~\ref{table:ukb_discoveries_confirm} reports the names and associated genes of the genetic variants identified by our method at the single-nucleotide resolution. These results indicate all but two of our high-resolution consistent discoveries correspond to variants with known biological consequences, which are located on genes previously reported to be associated with the phenotypes of interest. The full list of discoveries is available online at \url{https://msesia.github.io/knockoffgwas/}.

\begin{table}[!htb]
\center
\small
\caption{Proportions of discoveries confirmed by the GWAS Catalog, for two UK Biobank phenotypes. The multi-environment knockoff filter (MEKF) reports discoveries that are consistent in at least 2 populations, as in Table~\ref{table:ukb_discoveries}. Pooling refers to the analysis of~\cite{sesia2020controlling}. The binomial p-value tests whether the proportion of our confirmed discoveries differs significantly from that corresponding to pooling. The intersection heuristic is the same as in Table~\ref{table:ukb_discoveries_methods}.}
\label{table:ukb_discoveries_confirm_proportion}

\begin{tabular}[t]{cccccc}
\toprule
Phenotype & Resolution (kb) & MEKF & Pooling~\cite{sesia2020controlling} & Binomial p-value & Intersection\\
\midrule
 & single-SNP & 11 / 13 (85\%) & 63 / 95 (66\%) & $3.11 \cdot 10^{-01}$ & 0 / 0\\

 & 3 & 9 / 9 (100\%) & 360 / 570 (63\%) & $5.34 \cdot 10^{-02}$ & 0 / 0\\

 & 20 & 33 / 33 (100\%) & 1043 / 1503 (69\%) & $3.12 \cdot 10^{-04}$ & 25 / 25 (100\%)\\

 & 41 & 42 / 42 (100\%) & 1567 / 2384 (66\%) & $6.99 \cdot 10^{-06}$ & 68 / 68 (100\%)\\

 & 81 & 48 / 48 (100\%) & 1875 / 3006 (62\%) & $1.94 \cdot 10^{-07}$ & 83 / 84 (99\%)\\

 & 208 & 103 / 103 (100\%) & 1968 / 3339 (59\%) & $1.21 \cdot 10^{-16}$ & 102 / 107 (95\%)\\

\multirow[t]{-7}{*}{\centering\arraybackslash height} & 425 & 68 / 68 (100\%) & 1794 / 3073 (58\%) & $1.16 \cdot 10^{-11}$ & 154 / 164 (94\%)\\
\cmidrule{1-6}
 & single-SNP & 8 / 9 (89\%) & 44 / 53 (83\%) & 1 & 0 / 0\\

 & 3 & 10 / 10 (100\%) & 200 / 246 (81\%) & $2.76 \cdot 10^{-01}$ & 0 / 0\\

 & 20 & 27 / 27 (100\%) & 684 / 1002 (68\%) & $9.31 \cdot 10^{-04}$ & 26 / 26 (100\%)\\

 & 41 & 52 / 52 (100\%) & 804 / 1261 (64\%) & $1.71 \cdot 10^{-07}$ & 50 / 50 (100\%)\\

 & 81 & 101 / 104 (97\%) & 921 / 1570 (59\%) & $1.54 \cdot 10^{-14}$ & 67 / 69 (97\%)\\

 & 208 & 97 / 98 (99\%) & 937 / 1743 (54\%) & $4.16 \cdot 10^{-18}$ & 57 / 58 (98\%)\\

\multirow[t]{-7}{*}{\centering\arraybackslash platelet} & 425 & 119 / 119 (100\%) & 887 / 1653 (54\%) & $1.67 \cdot 10^{-22}$ & 70 / 75 (93\%)\\
\bottomrule
\end{tabular}

\end{table}

\section{Discussion} \label{sec:discussion}

\textbf{Consistency, causal inference, and reliability.} This paper proposed a practical high-dimensional method to search for conditional associations that are consistent across environments, provably controlling the false discovery rate within the model-X framework~\cite{candes2018panning}.
  While consistency can lead to valid causal inferences under certain conditions, the relevance of our method extends beyond the relatively narrow scope of the assumptions necessary for such formal connection. In fact, conditional associations and consistency are meaningful statistical concepts even in situations where discussing causality would require more care, either because it is not obvious that the explanatory variables predate the outcome, as in the case of medical imaging data~\cite{castro2020causality}, or because there is no clear notion of possible interventions~\cite{hernan2008does}. However, non-causal conditional associations can still be informative, especially if they are reproducible outside the data set in which they were discovered. For example, consistent conditional associations are useful to make reliable predictions, and fitting predictive models that can be accurate across different environments ({\em transfer learning}) is a well-known challenge in many fields, including genetics~\cite{Duncan2019}, machine learning~\cite{pan2009survey}, and econometrics~\cite{heckman1979sample}, to name a few.
Although this paper does not address prediction explicitly, the problem is related and our proposed method could be repurposed to select good predictors for transfer learning.

\textbf{Genome-wide association studies.} Genome-wide association studies are a natural application for our method because the model-X setup is supported by scientific knowledge of genetic inheritance~\cite{sesia2018}. Further, these studies are primarily exploratory, aiming to prioritize variants for follow-up investigations, which makes the false discovery rate a meaningful measure~\cite{storey2003statistical,sabatti2003false,benjamini2005quantitative}.
A constant causal model is also quite realistic here. First, there is little ambiguity about the causal direction because the genotypes are fixed at conception while the phenotype manifests itself later. Second, the biological mechanisms translating genotypes to phenotypes are likely the same for all humans: the differences across sub-populations lie in the genotypes.
The paucity of non-British samples limits our power with the UK Biobank data, but the growing awareness that genetic studies should increase the representation of different ancestries~\cite{Popejoy2018,Duncan2019} suggests promising future opportunities, especially as some large diverse studies already exist~\cite{gaziano2016million}.

% \textbf{Extensions to transfer learning.}
%This work inspires a novel approach to a related transfer learning problem. Suppose we would like to leverage data from $E$ different environments to improve the power of conditional testing within a new environment $E+1$.
% Our method can be repurposed to address this task, as explained in Appendix~\ref{section:transfer}.
% Transfer learning can be seen as a special case of a more general problem in which one wishes to leverage ``side information'' to improve the power of knockoff statistics, which has been recently studied by~\cite{ren2020knockoffs}. Our solution proposed in Appendix~\ref{section:transfer} is original in that it leverages the full external data sets rather than summary statistics.

\textbf{Opportunities for future work.} Our method may be useful in many fields, including the social sciences; there, it is easy to envision collecting data from multiple environments and consistency may help ensure the samples are truly random and free of network effects. Further, it is increasingly common to find high-dimensional data with many associations, for which controlling the false discovery rate is desirable. Regarding methodology, it may be possible to develop more powerful test statistics. Finally, we have proved selective SeqStep+~\cite{barber2015controlling} and the accumulation test~\cite{li2017accumulation} are valid under a mild form of dependency which suggests broader applicability than previously known.

\section*{Acknowledgements}

We thank Stefan Wager for insightful comments about an earlier manuscript draft, as well as the Center for Advanced Research Computing at the University of Southern California and the Research Computing Center at Stanford University for computing resources.
We are grateful to the participants and investigators of the UK Biobank.

\printbibliography

@article {Kosinski5802,
	author = {Kosinski, Michal and Stillwell, David and Graepel, Thore},
	title = {Private traits and attributes are predictable from digital records of human behavior},
	volume = {110},
	number = {15},
	pages = {5802--5805},
	year = {2013},
	doi = {10.1073/pnas.1218772110},
	publisher = {National Academy of Sciences},
	issn = {0027-8424},
	URL = {https://www.pnas.org/content/110/15/5802},
	eprint = {https://www.pnas.org/content/110/15/5802.full.pdf},
        journal={Proc. Natl. Acad. Sci. U.S.A.},
}

@article{Waldron2014,
    author = {Waldron, Levi and Haibe-Kains, Benjamin and Culhane, Aedín C. and Riester, Markus and Ding, Jie and Wang, Xin Victoria and Ahmadifar, Mahnaz and Tyekucheva, Svitlana and Bernau, Christoph and Risch, Thomas and Ganzfried, Benjamin Frederick and Huttenhower, Curtis and Birrer, Michael and Parmigiani, Giovanni},
    title = "{Comparative Meta-analysis of Prognostic Gene Signatures for Late-Stage Ovarian Cancer}",
    journal = {JNCI: Journal of the National Cancer Institute},
    volume = {106},
    number = {5},
    year = {2014},
    month = {04},
    issn = {0027-8874},
    doi = {10.1093/jnci/dju049},
    url = {https://doi.org/10.1093/jnci/dju049},
    note = {dju049},
    eprint = {https://academic.oup.com/jnci/article-pdf/106/5/dju049/17313389/dju049.pdf},
}

@article{Efron2020,
author = {Bradley Efron},
title = {Prediction, Estimation, and Attribution},
journal = {J. Am. Stat. Assoc},
volume = {115},
number = {530},
pages = {636-655},
year  = {2020},
publisher = {Taylor & Francis},
doi = {10.1080/01621459.2020.1762613},
eprint = { 
        https://doi.org/10.1080/01621459.2020.1762613}
}

@article{heckman1979sample,
  title={Sample selection bias as a specification error},
  author={Heckman, James J},
  journal={Econometrica},
  pages={153--161},
  year={1979},
  publisher={JSTOR}
}

@article{pan2009survey,
  title={A survey on transfer learning},
  author={Pan, Sinno Jialin and Yang, Qiang},
  journal={IEEE Transactions on knowledge and data engineering},
  volume={22},
  number={10},
  pages={1345--1359},
  year={2009},
  publisher={IEEE}
}

@article{hernan2008does,
  title={Does obesity shorten life? The importance of well-defined interventions to answer causal questions},
  author={Hern{\'a}n, Miguel A and Taubman, Sarah L},
  journal={Int. J. Obes.},
  volume={32},
  number={3},
  pages={S8--S14},
  year={2008},
  publisher={Nature Publishing Group}
}

@article{castro2020causality,
  title={Causality matters in medical imaging},
  author={Castro, Daniel C and Walker, Ian and Glocker, Ben},
  journal={Nat. Commun.},
  volume={11},
  number={1},
  pages={1--10},
  year={2020},
  publisher={Nature Publishing Group}
}

@article{schaid2018genome,
  title={From genome-wide associations to candidate causal variants by statistical fine-mapping},
  author={Schaid, Daniel J and Chen, Wenan and Larson, Nicholas B},
  journal={Nat. Rev. Genet.},
  volume={19},
  number={8},
  pages={491--504},
  year={2018},
  publisher={Nature Publishing Group}
}

@article{friedman2010regularization,
  title={Regularization paths for generalized linear models via coordinate descent},
  author={Friedman, Jerome and Hastie, Trevor and Tibshirani, Rob},
  journal={J. Stat. Softw.},
  volume={33},
  number={1},
  pages={1},
  year={2010},
  publisher={NIH Public Access}
}

@article{prive2019efficient,
  title={Efficient implementation of penalized regression for genetic risk prediction},
  author={Priv{\'e}, Florian and Aschard, Hugues and Blum, Michael GB},
  journal={Genetics},
  volume={212},
  number={1},
  pages={65--74},
  year={2019},
  publisher={Oxford University Press}
}

@article{berisa2016approximately,
  title={Approximately independent linkage disequilibrium blocks in human populations},
  author={Berisa, Tomaz and Pickrell, Joseph K},
  journal={Bioinformatics},
  volume={32},
  number={2},
  pages={283},
  year={2016},
  publisher={Oxford University Press}
}

@article{pritchard2001linkage,
  title={Linkage disequilibrium in humans: models and data},
  author={Pritchard, Jonathan K and Przeworski, Molly},
  journal={Am. J. Hum. Genet},
  volume={69},
  number={1},
  pages={1--14},
  year={2001},
  publisher={Elsevier}
}

@article{devlin1999genomic,
  title={Genomic control for association studies},
  author={Devlin, Bernie and Roeder, Kathryn},
  journal={Biometrics},
  volume={55},
  number={4},
  pages={997--1004},
  year={1999},
  publisher={Wiley Online Library}
}

@article{shimodaira2000improving,
  title={Improving predictive inference under covariate shift by weighting the log-likelihood function},
  author={Shimodaira, Hidetoshi},
  journal={J. Statist. Plann. Inference},
  volume={90},
  number={2},
  pages={227--244},
  year={2000},
  publisher={Elsevier}
}

@article{sugiyama2007covariate,
  title={Covariate shift adaptation by importance weighted cross validation.},
  author={Sugiyama, Masashi and Krauledat, Matthias and M{\"u}ller, Klaus-Robert},
  journal={J. Mach. Learn. Res.},
  volume={8},
  number={5},
  year={2007}
}

@article{daly2001high,
  title={High-resolution haplotype structure in the human genome},
  author={Daly, Mark J and Rioux, John D and Schaffner, Stephen F and Hudson, Thomas J and Lander, Eric S},
  journal={Nat. Genet.},
  volume={29},
  number={2},
  pages={229--232},
  year={2001},
  publisher={Nature Publishing Group}
}

@article{benjamini2005quantitative,
  title={Quantitative trait Loci analysis using the false discovery rate},
  author={Benjamini, Yoav and Yekutieli, Daniel},
  journal={Genetics},
  volume={171},
  number={2},
  pages={783--790},
  year={2005},
  publisher={Oxford University Press}
}

@article{sabatti2003false,
  title={False discovery rate in linkage and association genome screens for complex disorders},
  author={Sabatti, Chiara and Service, Susan and Freimer, Nelson},
  journal={Genetics},
  volume={164},
  number={2},
  pages={829--833},
  year={2003},
  publisher={Oxford University Press}
}

@article{storey2003statistical,
  title={Statistical significance for genomewide studies},
  author={Storey, John D and Tibshirani, Robert},
  journal={Proc. Natl. Acad. Sci. U.S.A.},
  volume={100},
  number={16},
  pages={9440--9445},
  year={2003},
  publisher={National Acad Sciences}
}

@book{imbens2015causal,
  title={Causal inference in statistics, social, and biomedical sciences},
  author={Imbens, Guido W and Rubin, Donald B},
  year={2015},
  publisher={Cambridge University Press}
}

@article{keele2015statistics,
  title={The statistics of causal inference: A view from political methodology},
  author={Keele, Luke},
  journal={Political Analysis},
  pages={313--335},
  year={2015},
  publisher={JSTOR}
}

@article{neyman1935statistical,
  title={Statistical problems in agricultural experimentation},
  author={Neyman, Jerzy and Iwaszkiewicz, Karolina},
  journal={Supplement to to J. R. Stat. Soc.},
  volume={2},
  number={2},
  pages={107--180},
  year={1935},
  publisher={JSTOR}
}

@book{fisher1935,
  address = {Edinburgh},
  author = {Fisher, R. A.},
  publisher = {Oliver and Boyd},
  title = {{The Design of Experiments}},
  year = 1935
}

@article{slatkin2008linkage,
  title={Linkage disequilibrium—understanding the evolutionary past and mapping the medical future},
  author={Slatkin, Montgomery},
  journal={Nat. Rev. Genet.},
  volume={9},
  number={6},
  pages={477--485},
  year={2008},
  publisher={Nature Publishing Group}
}

@article{wall2003haplotype,
  title={Haplotype blocks and linkage disequilibrium in the human genome},
  author={Wall, Jeffrey D and Pritchard, Jonathan K},
  journal={Nat. Rev. Genet.},
  volume={4},
  number={8},
  pages={587--597},
  year={2003},
  publisher={Nature Publishing Group}
}

@article{laan1997demographic,
  title={Demographic history and linkage disequilibrium in human populations},
  author={Laan, Maris and P{\"a}{\"a}bo, Svante},
  journal={Nat. Genet.},
  volume={17},
  number={4},
  pages={435--438},
  year={1997},
  publisher={Nature Publishing Group}
}

@article{slatkin1994linkage,
  title={Linkage disequilibrium in growing and stable populations.},
  author={Slatkin, Montgomery},
  journal={Genetics},
  volume={137},
  number={1},
  pages={331--336},
  year={1994},
  publisher={Genetics Soc America}
}

@article{tibshirani2011regression,
  title={Regression shrinkage and selection via the lasso: a retrospective},
  author={Tibshirani, Robert},
  journal={J. R. Stat. Soc. B},
  volume={73},
  number={3},
  pages={273--282},
  year={2011},
  publisher={Wiley Online Library}
}

@article{heller2014replicability,
  title={Replicability analysis for genome-wide association studies},
  author={Heller, Ruth and Yekutieli, Daniel},
  journal={Ann. Appl. Stat.},
  volume={8},
  number={1},
  pages={481--498},
  year={2014},
  publisher={Institute of Mathematical Statistics}
}

@article{friston2005conjunction,
  title={Conjunction revisited},
  author={Friston, Karl J and Penny, William D and Glaser, Daniel E},
  journal={Neuroimage},
  volume={25},
  number={3},
  pages={661--667},
  year={2005},
  publisher={Elsevier}
}

@article{liu2020fast,
  title={Fast and powerful conditional randomization testing via distillation},
  author={Liu, Molei and Janson, Lucas},
  journal={preprint at arXiv:2006.03980},
  year={2020}
}

@article{benjamini2001control,
  title={The control of the false discovery rate in multiple testing under dependency},
  author={Benjamini, Yoav and Yekutieli, Daniel},
  journal={Ann. Stat.},
  pages={1165--1188},
  year={2001},
  publisher={JSTOR}
}

@book{hume1739,
	publisher = {London: John Noon},
	author = {David Hume},
	title = {A Treatise of Human Nature: A Critical Edition},
	year = {1739}
}

@article{buniello2019nhgri,
  title={The {NHGRI-EBI GWAS Catalog} of published genome-wide association studies, targeted arrays and summary statistics 2019},
  author={Buniello, Annalisa and MacArthur, Jacqueline A L and Cerezo, Maria and Harris, Laura W and Hayhurst, James and Malangone, Cinzia and McMahon, Aoife and Morales, Joannella and Mountjoy, Edward and Sollis, Elliot and others},
  journal={Nucleic Acids Res.},
  volume={47},
  number={D1},
  pages={D1005--D1012},
  year={2019},
  publisher={Oxford University Press}
}

@article{katsevich2018controlling,
  title={Filtering the rejection set while preserving false discovery rate control},
  author={Katsevich, Eugene and Sabatti, Chiara and Bogomolov, Marina},
  journal={J. Am. Stat. Assoc},
  number={just-accepted},
  pages={1--27},
  year={2021},
  publisher={Taylor \& Francis}
}

@article{ren2020derandomizing,
  title={Derandomizing Knockoffs},
  author={Ren, Zhimei and Wei, Yuting and Cand{\`e}s, Emmanuel},
  journal={preprint at arXiv:2012.02717},
  year={2020}
}

@article{meinshausen2010stability,
  title={Stability selection},
  author={Meinshausen, Nicolai and B{\"u}hlmann, Peter},
  journal={J. R. Stat. Soc. B},
  volume={72},
  number={4},
  pages={417--473},
  year={2010},
  publisher={Wiley Online Library}
}

@article{lee2020network,
  title={Network dependence can lead to spurious associations and invalid inference},
  author={Lee, Youjin and Ogburn, Elizabeth L},
  journal={J. Am. Stat. Assoc},
  pages={1--15},
  year={2020},
  publisher={Taylor \& Francis}
}

@article{shalizi2011homophily,
  title={Homophily and contagion are generically confounded in observational social network studies},
  author={Shalizi, Cosma Rohilla and Thomas, Andrew C},
  journal={Sociol. Methods Res.},
  volume={40},
  number={2},
  pages={211--239},
  year={2011},
  publisher={Sage Publications Sage CA: Los Angeles, CA}
}

@article{shen2019false,
  title={False discovery rate control in cancer biomarker selection using knockoffs},
  author={Shen, Arlina and Fu, Han and He, Kevin and Jiang, Hui},
  journal={Cancers},
  volume={11},
  number={6},
  pages={744},
  year={2019},
  publisher={Multidisciplinary Digital Publishing Institute}
}

@article{srinivasan2020compositional,
  title={Compositional knockoff filter for high-dimensional regression analysis of microbiome data},
  author={Srinivasan, Arun and Xue, Lingzhou and Zhan, Xiang},
  journal={Biometrics},
  year={2020},
  publisher={Wiley Online Library}
}

@article{fan2020ipad,
  title={{IPAD}: stable interpretable forecasting with knockoffs inference},
  author={Fan, Yingying and Lv, Jinchi and Sharifvaghefi, Mahrad and Uematsu, Yoshimasa},
  journal={J. Am. Stat. Assoc},
  volume={115},
  number={532},
  pages={1822--1834},
  year={2020},
  publisher={Taylor \& Francis}
}

@article{chia2020interpretable,
  title={Interpretable Classification of Bacterial Raman Spectra with Knockoff Wavelets},
  author={Chia, Charmaine and Sesia, Matteo and Ho, Chi-Sing and Jeffrey, Stefanie and Dionne, Jennifer and Cand{\`e}s, Emmanuel and Howe, Roger},
  journal={preprint at arXiv:2006.04937},
  year={2021}
}

@article{liu2019power,
  title={Power analysis of knockoff filters for correlated designs},
  author={Liu, Jingbo and Rigollet, Philippe},
  journal={NeurIPS},
  year={2019}
}

@article{wang2020power,
  title={A Power Analysis of the Conditional Randomization Test and Knockoffs},
  author={Wang, Wenshuo and Janson, Lucas},
  journal={preprint at arXiv:2010.02304},
  year={2020}
}

@article{katsevich2020theoretical,
  title={A theoretical treatment of conditional independence testing under model-{X}},
  author={Katsevich, Eugene and Ramdas, Aaditya},
  journal={preprint at arXiv:2005.05506},
  year={2020}
}

@article{barber2020robust,
  title={Robust inference with knockoffs},
  author={Barber, Rina Foygel and Cand{\`e}, Emmanuel and Samworth, Richard J},
  journal={Ann. Stat.},
  volume={48},
  number={3},
  pages={1409--1431},
  year={2020},
  publisher={Institute of Mathematical Statistics}
}

@article{spector2020powerful,
  title={Powerful Knockoffs via Minimizing Reconstructability},
  author={Spector, Asher and Janson, Lucas},
  journal={preprint at arXiv:2011.14625},
  year={2020}
}

@inproceedings{gimenez2019knockoffs,
  title={Knockoffs for the mass: new feature importance statistics with false discovery guarantees},
  author={Gimenez, Jaime Roquero and Ghorbani, Amirata and Zou, James},
  booktitle={22nd International Conference on Artificial Intelligence and Statistics},
  pages={2125--2133},
  year={2019},
  organization={PMLR}
}

@article{romano2019,
author = {Yaniv Romano and Matteo Sesia and Emmanuel Cand{\`e}s},
title = {Deep Knockoffs},
journal = {J. Am. Stat. Assoc.},
volume = {0},
number = {ja},
pages = {1-27},
year  = {2019},
publisher = {Taylor & Francis},
doi = {10.1080/01621459.2019.1660174},
eprint = {https://doi.org/10.1080/01621459.2019.1660174}
}

@article{bates2020metropolized,
  title={Metropolized knockoff sampling},
  author={Bates, Stephen and Cand{\`e}s, Emmanuel and Janson, Lucas and Wang, Wenshuo},
journal = {J. Am. Stat. Assoc.},
  pages={1--15},
  year={2020},
  publisher={Taylor \& Francis}
}

@article{yu2020causality,
  title={Causality-based feature selection: Methods and evaluations},
  author={Yu, Kui and Guo, Xianjie and Liu, Lin and Li, Jiuyong and Wang, Hao and Ling, Zhaolong and Wu, Xindong},
  journal={ACM Computing Surveys (CSUR)},
  volume={53},
  number={5},
  pages={1--36},
  year={2020},
  publisher={ACM New York, NY, USA}
}

@article{wang2020towards,
  title={Towards efficient and effective discovery of Markov blankets for feature selection},
  author={Wang, Hao and Ling, Zhaolong and Yu, Kui and Wu, Xindong},
  journal={Inf. Sci.},
  volume={509},
  pages={227--242},
  year={2020},
  publisher={Elsevier}
}

@article{ling2019bamb,
  title={{BAMB}: A balanced Markov blanket discovery approach to feature selection},
  author={Ling, Zhaolong and Yu, Kui and Wang, Hao and Liu, Lin and Ding, Wei and Wu, Xindong},
  journal={ACM Transactions on Intelligent Systems and Technology (TIST)},
  volume={10},
  number={5},
  pages={1--25},
  year={2019},
  publisher={ACM New York, NY, USA}
}

@article{rojas2018invariant,
  title={Invariant models for causal transfer learning},
  author={Rojas-Carulla, Mateo and Sch{\"o}lkopf, Bernhard and Turner, Richard and Peters, Jonas},
  journal={J. Mach. Learn. Res.},
  volume={19},
  number={1},
  pages={1309--1342},
  year={2018},
  publisher={JMLR. org}
}

@article{yu2019multi,
  title={Multi-source causal feature selection},
  author={Yu, Kui and Liu, Lin and Li, Jiuyong and Ding, Wei and Le, Thuc Duy},
  journal={IEEE transactions on pattern analysis and machine intelligence},
  volume={42},
  number={9},
  pages={2240--2256},
  year={2019},
  publisher={IEEE}
}

@inproceedings{zhang2015multi,
  title={Multi-source domain adaptation: A causal view},
  author={Zhang, Kun and Gong, Mingming and Sch{\"o}lkopf, Bernhard},
  booktitle={Proceedings of the AAAI Conference on Artificial Intelligence},
  volume={29},
  number={1},
  year={2015}
}

@article{glymour2019review,
  title={Review of causal discovery methods based on graphical models},
  author={Glymour, Clark and Zhang, Kun and Spirtes, Peter},
  journal={Frontiers in genetics},
  volume={10},
  pages={524},
  year={2019},
  publisher={Frontiers}
}

@article{zhang2008completeness,
  title={On the completeness of orientation rules for causal discovery in the presence of latent confounders and selection bias},
  author={Zhang, Jiji},
  journal={Artificial Intelligence},
  volume={172},
  number={16-17},
  pages={1873--1896},
  year={2008},
  publisher={Elsevier}
}

@article{koivisto2004exact,
  title={Exact {Bayesian} structure discovery in {Bayesian} networks},
  author={Koivisto, Mikko and Sood, Kismat},
  journal={J. Mach. Learn. Res.},
  volume={5},
  pages={549--573},
  year={2004},
  publisher={JMLR. org}
}

@article{chickering2002optimal,
  title={Optimal structure identification with greedy search},
  author={Chickering, David Maxwell},
  journal={J. Mach. Learn. Res.},
  volume={3},
  number={Nov},
  pages={507--554},
  year={2002}
}

@misc{spirtes1999algorithm,
  title={An algorithm for causal inference in the presence of latent variables and selection bias in computation, causation and discovery, 1999},
  author={Spirtes, Peter and Meek, C and Richardson, T},
  year={1999},
  publisher={MIT Press}
}

@article{rothenhausler2021,
author = {Rothenh{\"a}usler, Dominik and Meinshausen, Nicolai and Bühlmann, Peter and Peters, Jonas},
title = {Anchor regression: Heterogeneous data meet causality},
journal = {J. R. Stat. Soc. B},
volume = {},
number = {},
pages = {},
year = {2021},
doi = {https://doi.org/10.1111/rssb.12398},
url = {https://rss.onlinelibrary.wiley.com/doi/abs/10.1111/rssb.12398},
}

@article{mooij2020joint,
  title={Joint Causal Inference from Multiple Contexts},
  author={Mooij, Joris M and Magliacane, Sara and Claassen, Tom},
  journal={J. Mach. Learn. Res.},
  volume={21},
  number={99},
  pages={1--108},
  year={2020}
}

@article{arjovsky2020invariant,
  title={Invariant Risk Minimization},
  author={Arjovsky, Martin and Bottou, L{\'e}on and Gulrajani, Ishaan and Lopez-Paz, David},
  journal={Stat},
  volume={1050},
  pages={27},
  year={2020}
}

@article{heinze2018invariant,
  title={Invariant causal prediction for nonlinear models},
  author={Heinze-Deml, Christina and Peters, Jonas and Meinshausen, Nicolai},
  journal={Journal of Causal Inference},
  volume={6},
  number={2},
  year={2018},
  publisher={De Gruyter}
}

@book{bollen1989,
  title={Structural Equations with Latent Variables},
  author={Boolen, Kenneth},
  year={1989},
  publisher={Wiley, New York}
}

@article{rubin2005causal,
  title={Causal inference using potential outcomes: Design, modeling, decisions},
  author={Rubin, Donald B},
  journal={J. Am. Stat. Assoc},
  volume={100},
  number={469},
  pages={322--331},
  year={2005},
  publisher={Taylor \& Francis}
}

@Article{benjamini1995,
    author = {Benjamini, Y. and Hochberg, Y.},
    title = {Controlling the false discovery rate: a practical and powerful approach to multiple testing},
    journal = {J. R. Stat. Soc. B.},
    volume = {57},
    pages = {289--300},
    year = 1995,
}

@article{gaziano2016million,
  title={{Million Veteran Program}: A mega-biobank to study genetic influences on health and disease},
  author={Gaziano, John Michael and Concato, John and Brophy, Mary and Fiore, Louis and Pyarajan, Saiju and Breeling, James and Whitbourne, Stacey and Deen, Jennifer and Shannon, Colleen and Humphries, Donald and others},
  journal={J. Clin. Epidemiol.},
  volume={70},
  pages={214--223},
  year={2016},
  publisher={Elsevier}
}

@Article{Duncan2019,
   Author="Duncan, L.  and Shen, H.  and Gelaye, B.  and Meijsen, J.  and Ressler, K.  and Feldman, M.  and Peterson, R.  and Domingue, B. ",
   Title="{{A}nalysis of polygenic risk score usage and performance in diverse human populations}",
   Journal="Nat. Commun.",
   Year="2019",
   Volume="10",
   Number="1",
   Pages="3328",
   Month="07"
}

@article{Popejoy2018,
author = {Popejoy, Alice B. and Ritter, Deborah I. and Crooks, Kristy and Currey, Erin and Fullerton, Stephanie M. and Hindorff, Lucia A. and Koenig, Barbara and Ramos, Erin M. and Sorokin, Elena P. and Wand, Hannah and Wright, Mathew W. and Zou, James and Gignoux, Christopher R. and Bonham, Vence L. and Plon, Sharon E. and Bustamante, Carlos D.},
doi = {10.1002/humu.23644},
journal = {Hum. Mutat.},
title = {{The clinical imperative for inclusivity: Race, ethnicity, and ancestry (REA) in genomics}},
year = {2018}
}

@article {sesia2020controlling,
	author = {Sesia, Matteo and Bates, Stephen and Cand{\`e}s, Emmanuel and Marchini, Jonathan and Sabatti, Chiara},
	title = {{FDR} control in {GWAS} with population structure},
	elocation-id = {2020.08.04.236703},
	year = {2021},
	publisher = {Cold Spring Harbor Laboratory},
	journal = {preprint at bioRxiv}
}

@Article{bycroft2018,
    author={Bycroft, Clare and Freeman, Colin and Petkova, Desislava and Band, Gavin and Elliott, Lloyd T. and Sharp, Kevin and Motyer, Allan and Vukcevic, Damjan and Delaneau, Olivier and O'Connell, Jared and Cortes, Adrian and Welsh, Samantha and Young, Alan and Effingham, Mark and McVean, Gil and Leslie, Stephen and Allen, Naomi and Donnelly, Peter and Marchini, Jonathan},
    title = {The {UK} Biobank resource with deep phenotyping and genomic data},
    journal = {Nature},
    volume = {562},
    pages = {203--209},
    year = 2018,
}

@article{o2016haplotype,
  title={Haplotype estimation for biobank-scale data sets},
  author={O'Connell, Jared and Sharp, Kevin and Shrine, Nick and Wain, Louise and Hall, Ian and Tobin, Martin and Zagury, Jean-Francois and Delaneau, Olivier and Marchini, Jonathan},
  journal={Nat. Genet.},
  volume={48},
  number={7},
  pages={817},
  year={2016},
  publisher={Nature Publishing Group}
}

@Article{li2003,
    author = {Li, N. and Stephens, M.},
    title = {Modeling linkage disequilibrium and identifying recombination hotspots using single-nucleotide polymorphism data},
    journal = {Genetics},
    volume = {165},
    pages = {2213--2233},
    year = 2003,
}

@Article{marchini2010,
    author = {Marchini, J. and Howie, B.},
    title = {Genotype imputation for genome-wide association studies},
    journal = {Nat. Rev. Genet.},
    volume = {11},
    pages = {499--511},
    year = 2010,
}

@article{neyman1933ix,
  title={On the problem of the most efficient tests of statistical hypotheses},
  author={Neyman, Jerzy and Pearson, Egon Sharpe},
  journal={Philos. Trans. R. Soc.},
  volume=231,
  number={694-706},
  pages={289--337},
  year=1933,
  publisher={The Royal Society London}
}

@article{hargittai2015bigger,
  title={Is bigger always better? Potential biases of big data derived from social network sites},
  author={Hargittai, Eszter},
  journal={Ann. Am. Acad. Pol. Soc. Sci.},
  volume={659},
  number={1},
  pages={63--76},
  year={2015},
  publisher={Sage Publications Sage CA: Los Angeles, CA}
}

@article{harford2014,
author = {Harford, Tim},
title = {Big data: A big mistake?},
journal = {Significance},
volume = {11},
number = {5},
pages = {14-19},
doi = {https://doi.org/10.1111/j.1740-9713.2014.00778.x},
url = {https://rss.onlinelibrary.wiley.com/doi/abs/10.1111/j.1740-9713.2014.00778.x},
eprint = {https://rss.onlinelibrary.wiley.com/doi/pdf/10.1111/j.1740-9713.2014.00778.x},
year = {2014}
}

@article{berkson1946limitations,
  title={Limitations of the application of fourfold table analysis to hospital data},
  author={Berkson, Joseph},
  journal={Biometrics Bulletin},
  volume={2},
  number={3},
  pages={47--53},
  year={1946},
  publisher={JSTOR}
}

@article{peters2016causal,
  title={Causal inference by using invariant prediction: identification and confidence intervals},
  author={Peters, Jonas and B{\"u}hlmann, Peter and Meinshausen, Nicolai},
  journal={J. R. Stat. Soc. B},
  pages={947--1012},
  year={2016},
  publisher={JSTOR}
}

@article{Bates2020,
	author = {Bates, Stephen and Sesia, Matteo and Sabatti, Chiara and Cand{\`e}s, Emmanuel},
	title = {Causal inference in genetic trio studies},
	volume = {117},
	number = {39},
	pages = {24117--24126},
	year = {2020},
	publisher = {National Academy of Sciences},
	eprint = {https://www.pnas.org/content/117/39/24117.full.pdf},
  journal={Proc. Natl. Acad. Sci. U.S.A},
}

@article{sesia2020multi,
  title={Multi-resolution localization of causal variants across the genome},
  author={Sesia, Matteo and Katsevich, Eugene and Bates, Stephen and Cand{\`e}s, Emmanuel and Sabatti, Chiara},
  journal={Nat. Commun.},
  volume={11},
  number={1},
  pages={1--10},
  year={2020},
  publisher={Nature Publishing Group}
}

@article{sesia2018,
    author = {Sesia, M and Sabatti, C and Candès, E J},
    title = "{Gene hunting with hidden Markov model knockoffs}",
    journal = {Biometrika},
    volume = {106},
    number = {1},
    pages = {1-18},
    year = {2018},
    month = {08},
    issn = {0006-3444},
    doi = {10.1093/biomet/asy033},
    url = {https://doi.org/10.1093/biomet/asy033},
    eprint = {https://academic.oup.com/biomet/article-pdf/106/1/1/27774559/asy033.pdf},
}

@article{10002015global,
  title={A global reference for human genetic variation},
  author={1000 Genomes Project Consortium and others},
  journal={Nature},
  volume={526},
  number={7571},
  pages={68},
  year={2015},
  publisher={Nature Publishing Group}
}

@article{candes2018panning,
  title={Panning for gold: ``model-X'' knockoffs for high dimensional controlled variable selection},
  author={Cand{\`e}s, Emmanuel and Fan, Yingying and Janson, Lucas and Lv, Jinchi},
  journal={J. R. Stat. Soc. B},
  volume={80},
  number={3},
  pages={551--577},
  year={2018},
  publisher={Wiley Online Library}
}

@article{barber2015controlling,
  title={Controlling the false discovery rate via knockoffs},
  author={Barber, Rina Foygel and Cand{\`e}s, Emmanuel},
  journal={Ann. Stat.},
  volume={43},
  number={5},
  pages={2055--2085},
  year={2015},
  publisher={Institute of Mathematical Statistics}
}

@article{wang2016detecting,
  title={Detecting Replicating Signals using Adaptive Filtering Procedures with the Application in High-throughput Experiments},
  author={Wang, Jingshu and Su, Weijie and Sabatti, Chiara and Owen, Art B},
  journal={preprint at arXiv:1610.03330},
  year={2016}
}

@article{barber2018robust,
  title={Robust inference with knockoffs},
  author={Barber, Rina Foygel and Cand{\`e}s, Emmanuel and Samworth, Richard J},
  journal={Ann. Stat.},
  volume={48},
  number={3},
  pages={1409--1431},
  year={2020},
  publisher={Institute of Mathematical Statistics}
}

@article{li2017accumulation,
  title={Accumulation tests for FDR control in ordered hypothesis testing},
  author={Li, Ang and Barber, Rina Foygel},
  journal={J. Am. Stat. Assoc},
  volume={112},
  number={518},
  pages={837--849},
  year={2017},
  publisher={Taylor \& Francis}
}

@book{pearl2009causality,
  title={Causality},
  author={Pearl, Judea},
  year={2009},
  publisher={Cambridge university press}
}

@article{herbert2020spectre,
  title={The spectre of Berkson's paradox: Collider bias in Covid-19 research},
  author={Herbert, Annie and Griffith, Gareth and Hemani, Gibran and Zuccolo, Luisa},
  journal={Significance},
  volume={17},
  number={4},
  pages={6--7},
  year={2020},
  publisher={Wiley Online Library}
}

@article{mcpherson2001birds,
  title={Birds of a feather: Homophily in social networks},
  author={McPherson, Miller and Smith-Lovin, Lynn and Cook, James M},
  journal={Annu. Rev. Sociol.},
  volume={27},
  number={1},
  pages={415--444},
  year={2001},
  publisher={Annual Reviews 4139 El Camino Way, PO Box 10139, Palo Alto, CA 94303-0139, USA}
}

@article{fowler2008dynamic,
  title={Dynamic spread of happiness in a large social network: longitudinal analysis over 20 years in the Framingham Heart Study},
  author={Fowler, James H and Christakis, Nicholas A},
  journal={Bmj},
  volume={337},
  year={2008},
  publisher={British Medical Journal Publishing Group}
}

@article{aral2009distinguishing,
  title={Distinguishing influence-based contagion from homophily-driven diffusion in dynamic networks},
  author={Aral, Sinan and Muchnik, Lev and Sundararajan, Arun},
  journal={Proc. Natl. Acad. Sci. U.S.A.},
  volume={106},
  number={51},
  pages={21544--21549},
  year={2009},
  publisher={National Acad Sciences}
}

@article{dawes1975graduate,
  title={Graduate admission variables and future success},
  author={Dawes, Robyn M},
  journal={Science},
  volume={187},
  number={4178},
  pages={721--723},
  year={1975},
  publisher={JSTOR}
}

@article{simes1986improved,
  title={An improved Bonferroni procedure for multiple tests of significance},
  author={Simes, R John},
  journal={Biometrika},
  volume={73},
  number={3},
  pages={751--754},
  year={1986},
  publisher={Oxford University Press}
}

@article{benjamini2008screening,
  title={Screening for partial conjunction hypotheses},
  author={Benjamini, Yoav and Heller, Ruth},
  journal={Biometrics},
  volume={64},
  number={4},
  pages={1215--1222},
  year={2008},
  publisher={Wiley Online Library}
}

\appendix

\renewcommand{\thefigure}{A\arabic{figure}}
\setcounter{figure}{0}
\renewcommand{\thetable}{A\arabic{table}}
\setcounter{table}{0}
\renewcommand{\thetheorem}{A\arabic{theorem}}
\setcounter{theorem}{0}
\renewcommand{\thelemma}{A\arabic{lemma}}
\setcounter{lemma}{0}
\renewcommand{\theprop}{A\arabic{prop}}
\setcounter{prop}{0}

\section{Testing for consistency with conditional randomizations} \label{sec:app-crt}

The conditional randomization test was proposed by~\cite{candes2018panning} as an alternative to knockoffs, and it may be seen as an instance of Fisher's randomization test~\cite{fisher1935} within the model-X framework.
We discuss here how to utilize it to test the consistent hypothesis $\mathcal{H}^{\ici}_{j}$~\eqref{eq:null-ici}, or the partially consistent $\mathcal{H}^{\pici,r}_{j}$~\eqref{eq:null-pici}. We do not aim to control the false discovery rate over all variables; instead, we focus on a single $j \in \{1,\ldots,p\}$. This problem is of separate interest because the conditional randomization test~\cite{candes2018panning} gives (approximately) continuous p-values, while the single bit of information obtainable with knockoffs can only be significant at an aggregate level, within a multiple-testing procedure~\cite{barber2015controlling}.

The conditional randomization test simulates independent realizations $X'_{j}$ of the variable of interest, $X_j$, conditional on all other predictors, $X_{-j}$, independently of the outcome. Then, it compares importance statistics $T_j$ based on the original $X_j$ to the empirical distribution of the analogous quantities $T'_j$ evaluated on the perturbed data set obtained by replacing $X_j$ with the random $X'_{j}$. The output p-value is defined roughly as one minus the empirical percentile of $T_j$ in the aforementioned distribution; i.e., larger values of $T_j$ result in smaller p-values. Environment by environment, this procedure produces a conservative p-value $p_j^e$ for $\mathcal{H}^{\mathrm{ci}, e}_{j}$~\eqref{eq:null-ci} because the distribution of $T'_j$ over multiple realizations of the random $X'_j$ is equivalent to the true null distribution of $T_j$ under $\mathcal{H}^{\mathrm{ci}, e}_{j}$~\eqref{eq:null-ci}~\cite{candes2018panning}.
A strength of this test is that it is flexible and potentially powerful: it can accommodate any statistics, similarly to knockoffs~\cite{candes2018panning}.
Two limitations compared to the latter are: (i) conditional randomization tends to be computationally more expensive, depending on the statistics~\cite{liu2020fast}, because $T'_j$ must be evaluated many times to obtain small p-values; (ii) the p-values for different $j$ are not independent, complicating the control of the false discovery rate~\cite{benjamini2001control}.
%In any case, we are not concerned with such issues here; our goal is simply to combine a given collection of p-values $p_j^e$ corresponding to different environments $e$, for a particular $j$.

For any fixed variable $j$, let $p_j^e$ be the conditional randomization p-value for testing $\mathcal{H}^{\mathrm{ci}, e}_{j}$~\eqref{eq:null-ci} in environment $e$.
 Then, a conservative p-value for testing $\mathcal{H}^{\ici}_{j}$~\eqref{eq:null-ici} is simply given by
\begin{align}  \label{eq:crt-ici}
p^{\inv}_j \defeq \max\cb{p_j^1, \dots, p_j^E}.
\end{align}
Indeed, $\mathcal{H}^{\ici}_{j}$~\eqref{eq:null-ici} implies there exists at least one environment $e$ such that $\mathcal{H}^{\mathrm{ci}, e}_{j}$~\eqref{eq:null-ci} is true. Then, $\forall \alpha \in (0,1)$,
\begin{align*}
  \P{p^{\inv}_j \leq \alpha}
   & = \P{\max\cb{p_j^1, \dots, p_j^E} \leq \alpha} \leq \P{p^{e}_j \leq \alpha} \leq \alpha,
\end{align*}
because $p^{e}_j$ is a conservative p-value for $\mathcal{H}^{\mathrm{ci}, e}_{j}$~\eqref{eq:null-ci}.

To test the partial consistency hypotheses $\mathcal{H}^{\pici,r}_{j}$~\eqref{eq:null-pici}, for any fixed $r \leq E$, one can combine the p-values as follows. First, sort the p-values for different environments in ascending order: $p_j^{(1)} \leq  \dots \leq p_j^{(E)}$. Then, define
\begin{align} \label{eq:simes_p_value}
p^{\inv, r}_j \defeq \min_{r \leq e \leq E} \cb{\frac{E-r+1}{e-r+1} p_j^{(e)} }.
\end{align}
This is known as Simes' partial conjunction p-value, at it is valid for $\mathcal{H}^{\pici,r}_{j}$~\eqref{eq:null-pici} if the p-values $p_j^e$ from different environments are mutually independent~\citep{simes1986improved, benjamini2008screening}. Note that $p^{\inv}_j$~\eqref{eq:crt-ici} is a special case of $p^{\inv, r}_j$~\eqref{eq:simes_p_value} with $r = E$.

The following experiments demonstrate the performance of the above conditional randomization p-values for consistency testing. We simulate $E=3$ environments, $p = 100$ variables, and $n = 200$ observations per environment. The explanatory variables are generated from an autoregressive model of order one with correlation parameter $\rho = 0.5$. For each $e$, the distribution of  $Y^e \mid X^e$ is given by a linear model with Gaussian errors: $Y^e = X^e \beta^e + \epsilon^e$.
The vector $\beta^e \in \mathbb{R}^p$ is an environment-specific parameter, and $\epsilon^e$ represents i.i.d.~standard Gaussian noise.
In each environment, the non-zero entries of $\beta^e$ are equal to $3/\sqrt{n}$. The non-zero entries of $\beta$ in each environment, $S^1, S^2 \in \{1,\ldots,p\}$, are chosen at random such that $S^3 = S^3 \cap S^2 = S^3 \cap S^2 \cap S^1$, $S^2 = S^2 \cap S^1$, while $|S^3| = 20$, $|S^2| = 40$ and $|S^1| = 60$.

\begin{figure}[!htb]
\includegraphics[width = \textwidth]{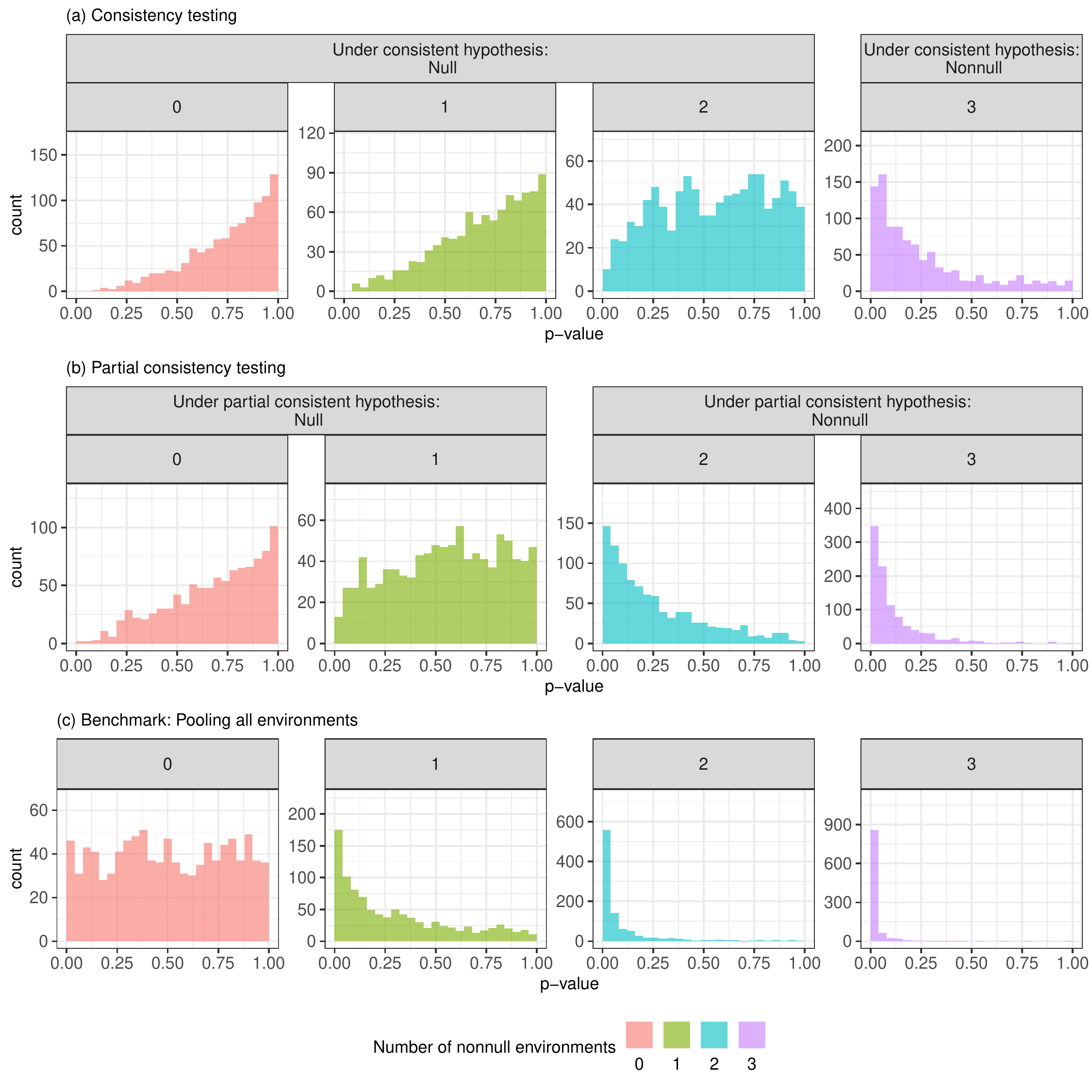}
\caption{Distributions of p-values for testing consistent (a) and partially consistent (b) associations via conditional randomization, over 1000 experiments. Different columns correspond to different tested variables. The null $\mathcal{H}^{\ici}_{j}$~\eqref{eq:null-ici} is true for the first three variables in (a), while the null $\mathcal{H}^{\pici,r}_{j}$~\eqref{eq:null-pici} is true for the first two variables in (b). In (c), we show the distribution of p-values obtained by applying the conditional randomization test to the pooled data.
}
\label{fig:CRT}
\end{figure}

We focus on four variables indexed by $j_0, j_1, j_2, j_3$, which are non-null in the sense of $\mathcal{H}^{\mathrm{ci}, e}_{j}$~\eqref{eq:null-ci} within exactly $0,1,2,3$ environments, respectively.
The test is carried out separately environment by environment. We take the importance statistics to be the absolute values of the lasso coefficients tuned by 10-fold cross validation. The number of randomizations is 100.
The resulting p-values are combined by applying~\eqref{eq:crt-ici}, or~\eqref{eq:simes_p_value} with $r=2$. This leads to a p-value $p^{\inv}_j$, or $p^{\pinv, 2}_j$, for each $j \in \cb{j_0, j_1, j_2, j_3}$.
Figure~\ref{fig:CRT} visualizes the distributions of $p^{\inv}_j$ and $p^{\pinv, 2}_j$ over experiments based on independent data. When the consistent conditional association is tested, $p^{\inv}_{j_0}$, $p^{\inv}_{j_1}$, $p^{\inv}_{j_2}$ are stochastically larger than Uniform$[0,1]$, as $\mathcal{H}^{\ici}_{j}$~\eqref{eq:null-ici} is true for the first three variables. At the same time, $p^{\inv}_{j_3}$ is stochastically smaller, suggesting the method can achieve non-trivial power.
Similarly, $p^{\pinv,2}_{j_0}$ and $p^{\pinv,2}_{j_1}$ are stochastically larger than Uniform$[0,1]$, as $\mathcal{H}^{\pici,2}_{j}$~\eqref{eq:null-pici} is true for the first two variables, while $p^{\pinv,2}_{j_2}$ and $p^{\pinv, 2}_{j_3}$ are clearly stochastically smaller than Uniform$[0,1]$.
As a heuristic benchmark, Figure~\ref{fig:CRT} shows the distribution of p-values obtained by applying the conditional randomization test to the pooled data from all environments. Clearly, this is not a valid test of any consistency hypotheses. In fact, only $p^{\operatorname{pool}}_{j_0}$ is stochastically larger than Uniform$[0,1]$, while the other p-values are stochastically smaller. This demonstrates that $p^{\operatorname{pool}}_{j_1}$ and $p^{\operatorname{pool}}_{j_2}$ are not valid p-values for $\mathcal{H}^{\ici}_{j}$~\eqref{eq:null-ici}, and $p^{\operatorname{pool}}_{j_1}$ is not a valid p-value for $\mathcal{H}^{\pici,2}_{j}$~\eqref{eq:null-pici}.

\section{An example of invalid multi-environment statistics} \label{sec:app-counter}

Following the discussion in Section~\ref{sec:ici-knockoffs}, we include here an example of invalid multi-environment knockoff statistics that do not lead to false discovery rate control. These naive statistics are invalid because their magnitudes are computed by looking at the unperturbed data from all environments, although their signs only depend on the observations from the environment of interest.
We simulate $E = 2$ environments, $p = 100$ variables, and $n = 200$ observations per environment. The explanatory variables are generated from an autoregressive model of order one with correlation parameter $\rho = 0.6$. The distribution of  $Y^e \mid X^e$, for each $e \in \{1,2\}$, is given by a linear model with Gaussian errors: $Y^e = X^e \beta^e + \epsilon^e$.
The vector $\beta^e \in \mathbb{R}^p$ is an environment-specific parameter, and $\epsilon^e$ represents i.i.d.~standard Gaussian noise.
In each environment, 70 entries of $\beta^e$ are non-zero while the others are equal to $a/\sqrt{n}$, where $a$ is a control parameter.
The non-zero entries of $\beta$ in each environment, $S^1, S^2 \in \{1,\ldots,p\}$, are chosen at random such that $|S^1 \cap S^2| = 40$.
The goal is to discover the set of consistent non-nulls, controlling the false discovery rate below $10\%$.

\begin{figure}[!htb]
    \centering
    \includegraphics[width = 0.65 \linewidth]{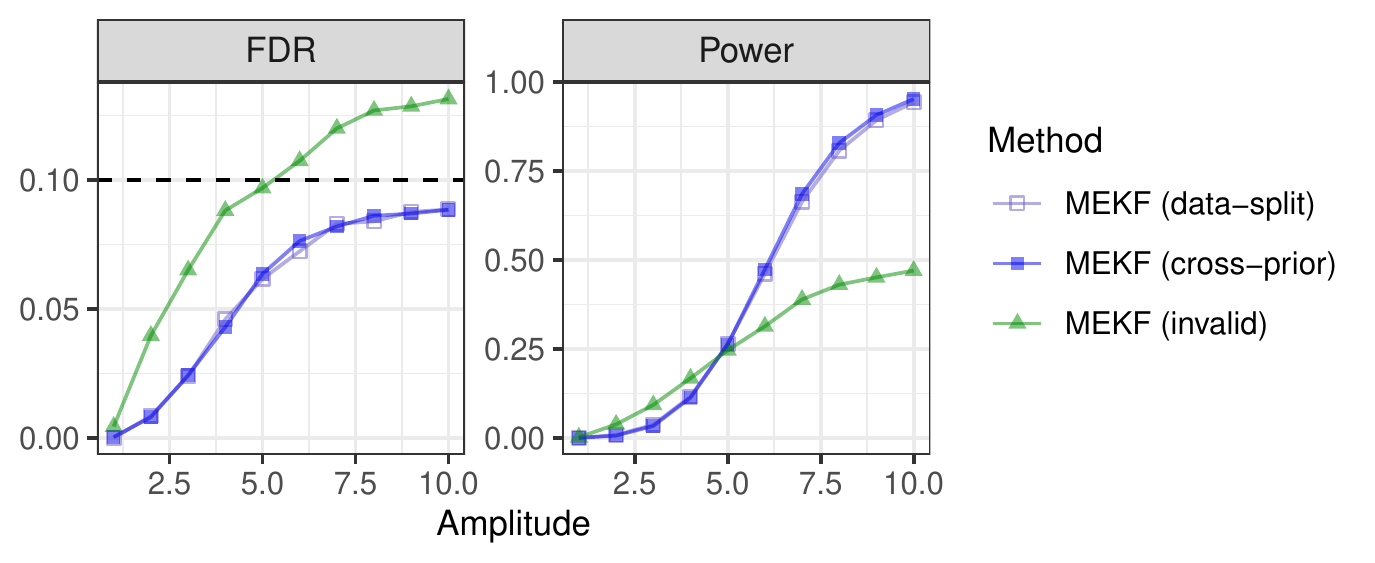}
    \caption{An example in which a naive implementation of the multi-environment knockoff filter with invalid statistics would not control the false discovery rate, while our method does. The false discovery rate and power evaluated over 500 experiments is shown as a function of the signal amplitude. The nominal false discovery rate is 10\%.}
    \label{fig:counter_example}
\end{figure}

Figure~\ref{fig:counter_example} compares the performances of three methods as a function of the signal amplitude~$a$. The first two are our multi-environment knockoff filter with data-split and cross-prior statistics, respectively. The third one naively computes the magnitude of each $W_j^e$ with the standard lasso coefficient difference statistics applied to the pooled data, and then it determines the sign of $W_j^e$ by applying the same procedure on the data from environment $e$. The results show the naive approach does not control the false discovery rate, confirming the importance of our careful construction of multi-environment knockoff statistics (Definition~\ref{def:kenv-stats}).

\section{Additional proofs}

\subsection{Selective SeqStep+ test with dependent p-values} \label{app:proof-selseqstep}

\theoremseqstep*
\begin{proof}[Proof of Theorem~\ref{th:seqstep-fdr}]
Our proof is based on a leave-one-out argument following closely that of Theorem~2 in~\cite{barber2018robust}.
The rejection threshold of selective SeqStep+ at level $\alpha$ is:
\begin{align} \label{eq:sel-seq-step-threshold}
  \hat{\omega} = \min \left\{ \omega : \frac{1 + |\{ j : |W_{j}^{\inv}| \geq \omega, p_j^{\inv} > c\}| }{|\{ j : |W_{j}^{\inv}| \geq \omega, p_j^{\inv} \leq c\}| \lor 1} \leq \frac{1-c}{c} \alpha \right\}.
\end{align}
Define $\mathcal{H}^{\ici} \subseteq \{1,\ldots,p\}$ as the subset of true null $\mathcal{H}^{\ici}_j$~\eqref{eq:null-ici}, and $\hat{\omega}_j$ as the rejection threshold resulting from $p_j^{\inv} \to 0$ in~\eqref{eq:sel-seq-step-threshold}.
Then, the false discovery rate is
\begin{align*}
  \text{FDR}
  & = \E{\frac{\sum_{j \in \mathcal{H}^{\ici}}\I{|W_{j}^{\inv}| \geq \hat{\omega}, p_j^{\inv} \leq c} }{|\{ j : |W_{j}^{\inv}| \geq \hat{\omega}, p_j^{\inv} \leq c\}| \lor 1}} \\
  & = \E{\frac{\sum_{j \in \mathcal{H}^{\ici}}\I{|W_{j}^{\inv}| \geq \hat{\omega}, p_j^{\inv} \leq c} }{1 + |\{ j : |W_{j}^{\inv}| \geq \hat{\omega},  p_j^{\inv} > c\}|} \cdot \frac{1 + |\{ j : |W_{j}^{\inv}| \geq \hat{\omega}, p_j^{\inv} > c\}|}{|\{ j : |W_{j}^{\inv}| \geq \hat{\omega}, p_j^{\inv} \leq c\}| \lor 1}} \\
  & \leq \E{\frac{\sum_{j \in \mathcal{H}^{\ici}}\I{|W_{j}^{\inv}| \geq \hat{\omega}, p_j^{\inv} \leq c} }{1 + |\{ j : |W_{j}^{\inv}| \geq \hat{\omega}, p_j^{\inv} > c\}|}} \cdot \frac{1-c}{c} \alpha \\
  & \leq \E{\frac{\sum_{j \in \mathcal{H}^{\ici}}\I{|W_{j}^{\inv}| \geq \hat{\omega}, p_j^{\inv} \leq c} }{1 + \sum_{j \in \mathcal{H}^{\ici}}\I{|W_{j}^{\inv}| \geq \hat{\omega}, p_j^{\inv} > c} }} \cdot \frac{1-c}{c} \alpha \\
  & = \sum_{j \in \mathcal{H}^{\ici}} \E{\frac{\I{|W_{j}^{\inv}| \geq \hat{\omega}, p_j^{\inv} \leq c} }{1 + \sum_{l \in \mathcal{H}^{\ici}, l \neq j}\I{l : |W_{l}^{\inv}| \geq \hat{\omega} , p_l^{\inv} > c} }} \cdot \frac{1-c}{c} \alpha \\
  & = \sum_{j \in \mathcal{H}^{\ici}} \E{\frac{\I{|W_{j}^{\inv}| \geq \hat{\omega}_j, p_j^{\inv} \leq c} }{1 + \sum_{l \in \mathcal{H}^{\ici}, l \neq j}\I{l : |W_{l}^{\inv}| \geq \hat{\omega}_j , p_l^{\inv} > c} }} \cdot \frac{1-c}{c} \alpha,
\end{align*}
Above, the last equality follows with the same argument as in the proof of Theorem~2 of~\cite{barber2018robust}: if $p_j^{\inv} \leq c$, then $\hat{\omega} = \hat{\omega}_j$.
Then, as $\hat{\omega}_j$ is only a function of $|W^{\inv}|$ and $p_{-j}^{\inv}$, we can write
\begin{align*}
  \text{FDR}
  & \leq \sum_{j \in \mathcal{H}^{\ici}} \E{\frac{\I{|W_{j}^{\inv}| \geq \hat{\omega}_j} \P{p_j^{\inv} \leq c \mid p_{-j}^{\inv}, |W^{\inv}|} }{1 + \sum_{l \in \mathcal{H}^{\ici}, l \neq j}\I{l : |W_{l}^{\inv}| \geq \hat{\omega}_j , p_l^{\inv} > c} }} \cdot \frac{1-c}{c} \alpha \\
  & \leq \sum_{j \in \mathcal{H}^{\ici}} \E{\frac{\I{|W_{j}^{\inv}| \geq \hat{\omega}_j} \P{p_j^{\inv} > c \mid p_{-j}^{\inv}, |W^{\inv}|} }{1 + \sum_{l \in \mathcal{H}^{\ici}, l \neq j}\I{l : |W_{l}^{\inv}| \geq \hat{\omega}_j , p_l^{\inv} > c} }} \cdot \alpha \\
  & = \sum_{j \in \mathcal{H}^{\ici}} \E{\frac{\I{|W_{j}^{\inv}| \geq \hat{\omega}_j, p_j^{\inv} > c} }{1 + \sum_{l \in \mathcal{H}^{\ici}, l \neq j}\I{l : |W_{l}^{\inv}| \geq \hat{\omega}_j , p_l^{\inv} > c} }} \alpha,
\end{align*}
where the second inequality follows directly from the ``almost-independence'' property of Proposition~\eqref{prop:pvals-invariant} because
\begin{align*}
  \frac{\P{p_j^{\inv} \leq c \mid p_{-j}^{\inv}, |W^{\inv}|}}{\P{p_j^{\inv} > c \mid p_{-j}^{\inv}, |W^{\inv}|}} \leq \frac{c}{1-c}.
\end{align*}
Now, it follows from Lemma 6 of~\cite{barber2018robust} that
\begin{align*}
  \frac{\sum_{j \in \mathcal{H}^{\ici}} \I{|W_{j}^{\inv}| \geq \hat{\omega}_j, p_j^{\inv} > c} }{1 + \sum_{l \in \mathcal{H}^{\ici}, l \neq j}\I{l : |W_{l}^{\inv}| \geq \hat{\omega}_j , p_l^{\inv} > c} }
  & = \frac{\sum_{j \in \mathcal{H}^{\ici}} \I{|W_{j}^{\inv}| \geq \hat{\omega}_j, p_j^{\inv} > c} }{1 + \sum_{l \in \mathcal{H}^{\ici}, l \neq j}\I{l : |W_{l}^{\inv}| \geq \hat{\omega}_l , p_l^{\inv} > c} } \\
  & = \frac{\sum_{j \in \mathcal{H}^{\ici}} \I{|W_{j}^{\inv}| \geq \hat{\omega}_j, p_j^{\inv} > c} }{\sum_{l \in \mathcal{H}^{\ici}}\I{l : |W_{l}^{\inv}| \geq \hat{\omega}_l , p_l^{\inv} > c} }  = 1.
\end{align*}
Hence we proved $\text{FDR} \leq \alpha$.
\end{proof}

\subsection{Accumulation test with dependent p-values} \label{app:proof-acc}

Our proof of Theorem~\ref{th:acc-fdr} follows the strategy of~\cite{li2017accumulation} (Theorem 2 therein), with some modifications to relax as needed their independence assumption.
The new idea is to couple our p-values to {\em imaginary and mutually independent} p-values obtained by replacing $n^-_j$ in~\eqref{eq:pici-pvalues-coin} with the number of negative signs in a suitable subset of true null environments for column $j$.
Our result leverages two lemmas: the first one is borrowed from~\cite{li2017accumulation}, and the second one is a generalization of a similar result from~\cite{li2017accumulation} which we prove at the end of this section.

\begin{lemma}[Lemma B.3 from~\cite{li2017accumulation}] \label{lemma:li2017}
Let $B_1, \ldots, B_m \in \{0,1\}$ be independent, with $B_j \overset{\mathrm{i.i.d.}}{\sim } \text{Bernoulli}(\rho)$ for all $j \in \mathcal{H}_0$, for some subset $\mathcal{H}_0 \subseteq \{1,\ldots,m\}$ and $\rho \in (0,1)$. Let $\{ \mathcal{F}_k\}_{k=1,\ldots,m}$ be any filtration in reverse time (i.e., $\mathcal{F}_{k+1} \subseteq \mathcal{F}_{k}$) such that:
\begin{align}
& B_j \in \mathcal{F}_k \text{ for all } j \notin \mathcal{H}_0, \text{ and for all } j > k \text{ with } j \in \mathcal{H}_0, \label{eq:filtration-properties-1} \\
& \sum_{j \leq k, j \in \mathcal{H}_0} B_j \in \mathcal{F}_k, \text{and} \label{eq:filtration-properties-2} \\
& \{B_j : j \leq k, j \in \mathcal{H}_0\} \text{ are exchangeable with respect to } \mathcal{F}_k, \label{eq:filtration-properties-3}
\end{align}
for all $k \in \{1,\ldots,m\}$. Then, the following is a supermartingale with respect to $\{\mathcal{F}_k\}$ and $\E{M_n} \leq 1/\rho$:
\begin{align}
  M_k \defeq \frac{1 + \#\{j \leq k : j \in \mathcal{H}_0\}}{1 + \sum_{j \leq k, j \in \mathcal{H}_0} B_j}.
\end{align}
\end{lemma}

\begin{lemma} \label{lemma:technical-accumulation}
Let $p_1,\ldots,p_m$ be the p-values defined in~\eqref{eq:pici-pvalues-coin}, sorted in decreasing order of $|W_j^{\pinv,r}|$, for $j \in \{1,\ldots,m\}$. Let $\hat{k}$ denote the number of rejections obtained by applying the accumulation test of~\cite{li2017accumulation} at level $\alpha$ to these p-values, with a monotone increasing accumulation function $h : [0,1] \mapsto [0,\infty)$ such that $\int_{0}^{1} h(t) dt = 1$. That is,
\begin{align} \label{eq:stopping-time}
\hat{k} = \max \left\{ k \in \{1,\ldots,m\} : \frac{1}{k} \sum_{j=1}^{k} h(p_j) \leq \alpha \right\}.
\end{align}
Then, for some fixed $C>0$,
\begin{align} \label{eq:lemma-supermartingale-result}
  \E{\frac{\#\{j \leq \hat{k} : j \in \mathcal{H}^{\pinv,r} \}}{C + \sum_{j=1}^{\hat{k}} h(p_j)}}
  \leq \frac{1}{\int_{0}^{1} [h(t) \land C] dt}.
\end{align}
\end{lemma}

\theoremaccumulation*
\begin{proof}[Proof of Theorem~\ref{th:acc-fdr}]

Assume the p-values are sorted in decreasing order of $|W_j^{\pinv,r}|$ and let $\hat{k}$ be the number of rejections made by the accumulation test at a fixed nominal level~$\alpha$, as defined in~\eqref{eq:stopping-time}.  To simplify the notation, we will write $p_j$ instead of $p_{j}^{\pinv,r}$.
Proceeding as in~\cite{li2017accumulation}, we see that
\begin{align*}
    \E{\text{mFDP}_{C/\alpha}(\hat{k})}
    & = \E{\frac{\#\{j \leq \hat{k} : j \in \mathcal{H}^{\pinv,r} \}}{C/\alpha + \hat{k}}} \\
    & = \E{\frac{\#\{j \leq \hat{k} : j \in \mathcal{H}^{\pinv,r} \}}{C + \sum_{j=1}^{\hat{k}} h(p_j)} \cdot \frac{C + \sum_{j=1}^{\hat{k}} h(p_j)}{C/\alpha + \hat{k}} } \\
    & \leq \alpha \cdot \E{\frac{\#\{j \leq \hat{k} : j \in \mathcal{H}^{\pinv,r} \}}{C + \sum_{j=1}^{\hat{k}} h(p_j)}},
\end{align*}
where the inequality follows from the definition of $\hat{k}$~\eqref{eq:stopping-time}.
Lemma~\ref{lemma:technical-accumulation} completes the proof.
The same argument can be repurposed to prove the stricter threshold $\hat{k}^+$ defined in~\cite{li2017accumulation} controls the (unmodified) false discovery rate.
\end{proof}

\begin{proof}[Proof of Lemma~\ref{lemma:technical-accumulation}]
  This proof follows a similar strategy as that of Lemma~B.2 in~\cite{li2017accumulation}, although we additionally need to leverage the special structure of our partial consistency p-values~\eqref{eq:pici-pvalues-coin} to get around their lack of independence.

Take any $j$ such that $\mathcal{H}^{\pinv,r}_j$ is true. Because there must be at least $K-r+1$ null environments for this variable, we can define $\tilde{n}^{-}_j$ as the number of negative signs among the first (in any arbitrary order) $K-r+1$ null entries of the $j$-th column of $\mathbf{W}$. We assume hereafter that any zero entries in $\mathbf{W}$ have been randomly assigned a positive or negative sign by flipping independent fair coins. Concretely, let us define the list of indices for these environments as $\mathcal{E}^0_j$, as we will need to refer to them later. Define also $\vardbtilde{n}^{-}_j = n^{-}_j - \tilde{n}^{-}_j \geq 0$, the number of negative signs from environments other than those in $\mathcal{E}^0_j$.
Define then $\tilde{p}_j$ as the p-value obtained by replacing $n^-_j$ with $\tilde{n}^{-}_j$ in~\eqref{eq:pici-pvalues-coin}, that is
\begin{align*}
  \tilde{p}_j
  & \defeq \Psi\left( E-r+1, \frac{1}{2}, \tilde{n}^-_j -1 \right) + U_j \cdot \psi\left( E-r+1, \frac{1}{2}, \tilde{n}^-_j \right).
\end{align*}
The imaginary p-values $\tilde{p}_j$ for $j \in \mathcal{H}^{\pinv,r}$ are independent of each other because they are only affected by the signs of the true null entries in $\mathbf{W}$, which satisfies Definition~\ref{def:kenv-stats}. They are also exactly uniformly distributed,
\begin{align} \label{eq:imaginary-pvals-iid}
  \tilde{p}_j \overset{\text{i.i.d.}}{\sim} \text{Uniform}[0,1],
\end{align}
for $j \in \mathcal{H}^{\pinv,r}$, because they are the randomized binomial p-values corresponding to
\begin{align} \label{eq:imaginary-counts-iid}
  \tilde{n}^{-}_j \overset{\text{i.i.d.}}{\sim} \text{Binomial}\left(E-r+1, 1/2\right).
\end{align}
Further, note that $\tilde{p}_j \leq p_j$ almost-surely because $\tilde{n}^-_j \leq n^-_j$; the independent uniform random variables $U_j$ are taken to be the same for both real and imaginary p-values.
Therefore, as we assumed the accumulation function to be monotone increasing, we also have that $h(\tilde{p}_j) \leq h(p_j)$ and, for all $k \in \{1,\ldots,p\}$,
\begin{align*}
    M_k \defeq \frac{\#\{j \leq k : j \in \mathcal{H}^{\pinv,r} \}}{C + \sum_{j=1}^{k} h(p_j)}
        \leq \frac{\#\{j \leq k : j \in \mathcal{H}^{\pinv,r} \}}{C + \sum_{j=1}^{k} h(\tilde{p}_j)} \defeqi \tilde{M}_k.
\end{align*}
Now we can deal with $\tilde{M}_k$ with the same approach of~\cite{li2017accumulation}. Define an i.i.d.~sequence $V_j \sim \text{Uniform}[0,1]$, for $j \in \{1,\ldots,m\}$,  independent of everything else, and random variables $B_j = \I{V_i \leq h(\tilde{p}_i) / C}$.
Conditional on $\tilde{p}_1,\ldots,\tilde{p}_m$, the variables $B_j$ are independent and $B_j \sim \text{Bernoulli}([h(\tilde{p}_j)/C] \land 1)$. Further, marginally,
\begin{align*}
  \E{B_j} = \E{\frac{h(\tilde{p}_j) \land C}{C}} = \frac{1}{C} \int_{0}^{1} [h(t) \land C] dt \defeqi \rho.
\end{align*}
As in~\cite{li2017accumulation}, we would like to apply Lemma~\ref{lemma:li2017} to bound
\begin{align*}
  \E{\frac{1 + \#\{j \leq \hat{k} : j \in \mathcal{H}_0\}}{1 + \sum_{j \leq \hat{k}, j \in \mathcal{H}_0} B_j}},
\end{align*}
which requires a reverse-time filtration satisfying~\eqref{eq:filtration-properties-1}--\eqref{eq:filtration-properties-3}, and then show $\hat{k}$ is a stopping time with respect to it.

First, let $\mathcal{G}$ be the $\sigma$-algebra generated by the following variables:
\begin{itemize}[noitemsep,topsep=0pt]
  \item $|\mathbf{W}|$ (the absolute values of all entries in $\mathbf{W}$);
  \item $\text{sign}(W^e_j)$ for $j \in \{1,\ldots,m\}$ and $e \notin \mathcal{H}_j$ (the signs of $W_{j}^e$ for non-null environments);
%  \item $\text{sign}(W^e_j)$ for $j \in \{1,\ldots,m\}$ and $e \in \mathcal{H}_j \setminus \mathcal{E}^0_j$ (the signs of $W_{j}^e$ for null environments ignored by $\tilde{n}^{-}_j$);
  \item $\tilde{p}_j$ for $j \notin \mathcal{H}^{\pinv,r}$ (the imaginary p-values corresponding to non-null consistent hypotheses).
\end{itemize}
Second, for any $k \in \{1,\ldots,m\}$, let $\mathcal{F}_k$ be the union of $\mathcal{G}$ with the $\sigma$-algebra generated by:
\begin{itemize}[noitemsep,topsep=0pt]
\item $(\tilde{p}_j, \tilde{n}^{-}_j, U_j, \vardbtilde{n}^{-}_j)$ for all $j > k$;
\item $\{(\tilde{p}_j, \tilde{n}^{-}_j, U_j, \vardbtilde{n}^{-}_j)\}_{j=1}^{k}$ (as an unordered set).
\end{itemize}
This is a reverse-time filtration and it satisfies~\eqref{eq:filtration-properties-1}--\eqref{eq:filtration-properties-3}. The key observations here are that the property of $\mathbf{W}$ in Definition~\ref{def:kenv-stats} implies~\eqref{eq:imaginary-pvals-iid}--\eqref{eq:imaginary-counts-iid} hold also conditional on $\mathcal{G}$, and that $\vardbtilde{n}^{-}_j$ and $\tilde{n}^{-}_j$ are independent.
Regarding $\hat{k}$~\eqref{eq:stopping-time}, note that it is defined in terms of $p_j$, not $\tilde{p}_j$ (as in practice one needs to evaluate $\hat{k}$ without knowing a priori which environments correspond to true nulls). However, the real p-values $p_j$ can be reconstructed exactly from knowledge of $\tilde{n}^{-}_j, U_j, \vardbtilde{n}^{-}_j$ and the information in $\mathcal{G}$.
 Therefore, $\hat{k}$ is a stopping time with respect to $\{\mathcal{F}_k\}$. Thus, we can apply Lemma~\ref{lemma:li2017} in combination with the optional stopping theorem to conclude that
\begin{align} \label{eq:ost}
  \E{\frac{1 + \#\{j \leq \hat{k} : j \in \mathcal{H}_0\}}{1 + \sum_{j \leq \hat{k}, j \in \mathcal{H}_0} B_j}} \leq \frac{1}{\rho} = \frac{1}{\frac{1}{C} \int_{0}^{1} [h(t) \land C]}.
\end{align}
The rest of the proof follows closely that of Lemma~B.2 in~\cite{li2017accumulation} but it is nonetheless recalled here for completeness.
\begin{align*}
  \E{\frac{1 + \#\{j \leq \hat{k} : j \in \mathcal{H}_0\}}{1 + \sum_{j \leq \hat{k}, j \in \mathcal{H}_0} B_j}}
  & = \E{ \E{\frac{1 + \#\{j \leq \hat{k} : j \in \mathcal{H}_0\}}{1 + \sum_{j \leq \hat{k}, j \in \mathcal{H}_0} B_j} \mid \tilde{p}_1, \ldots, \tilde{p}_m} }  \\
  & = \E{ \left( 1 + \#\{j \leq \hat{k} : j \in \mathcal{H}_0\} \right) \E{\frac{1}{1 + \sum_{j \leq \hat{k}, j \in \mathcal{H}_0} B_j} \mid \tilde{p}_1, \ldots, \tilde{p}_m} }  \\
  & \geq \E{ \left( 1 + \#\{j \leq \hat{k} : j \in \mathcal{H}_0\} \right) \frac{1}{\E{1 + \sum_{j \leq \hat{k}, j \in \mathcal{H}_0} B_j \mid \tilde{p}_1, \ldots, \tilde{p}_m}}} \\
  & = \E{ \left( 1 + \#\{j \leq \hat{k} : j \in \mathcal{H}_0\} \right) \frac{1}{1+ \sum_{j \leq \hat{k}, j \in \mathcal{H}_0} \frac{h(\tilde{p}_j) \land C}{C}}} \\
  & \geq C \cdot \E{ \frac{1 + \#\{j \leq \hat{k} : j \in \mathcal{H}_0\}}{C+ \sum_{j \leq \hat{k}, j \in \mathcal{H}_0} h(\tilde{p}_j)}}.
\end{align*}
Above, the first inequality is Jensen's inequality. The proof of this lemma is then completed by~\eqref{eq:ost}.
\end{proof}

\section{Additional details about numerical experiments}  \label{sec:app-experiments}

\subsection{Synthetic data} \label{sec:app-synthetic}

\begin{figure}[!htb]
  \centering
  \begin{subfigure}{0.45\textwidth}
    \centering
    \includegraphics[width=0.7\linewidth]{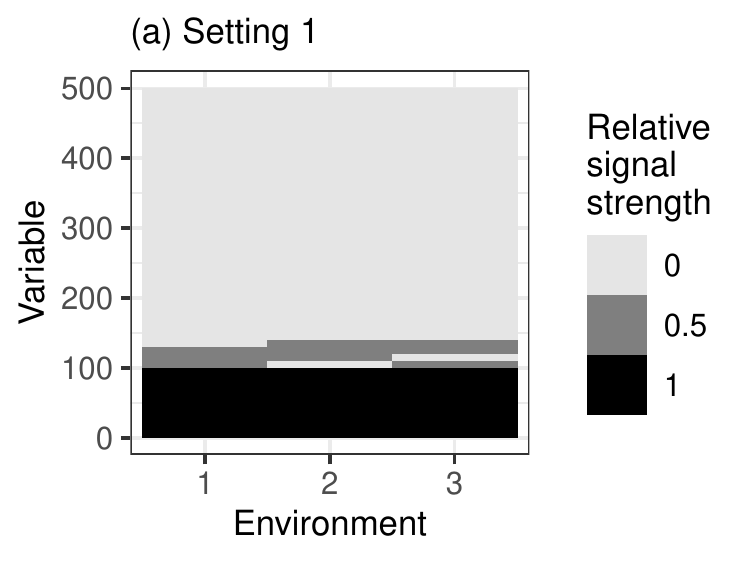}
  \end{subfigure}
  \begin{subfigure}{0.45\textwidth}
    \includegraphics[width=0.7\linewidth]{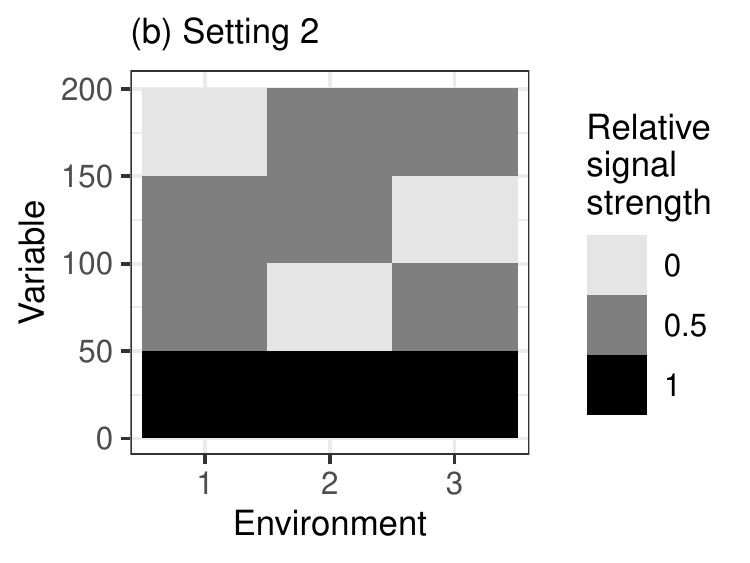}
  \end{subfigure}
  \caption{True hypothesis structure for the numerical experiments of Figure~\ref{fig:inv_sim}. The shaded rectangles indicate which variables are non-null in each environment. Darker shades indicate stronger associations. The variables are sorted here for ease of visualization; in the experiments, their order is randomized.}
  \label{fig:inv_beta}
\end{figure}

\begin{figure}[!htb]
  \centering
  \begin{subfigure}{0.45\textwidth}
    \centering
    \includegraphics[width=0.7\linewidth]{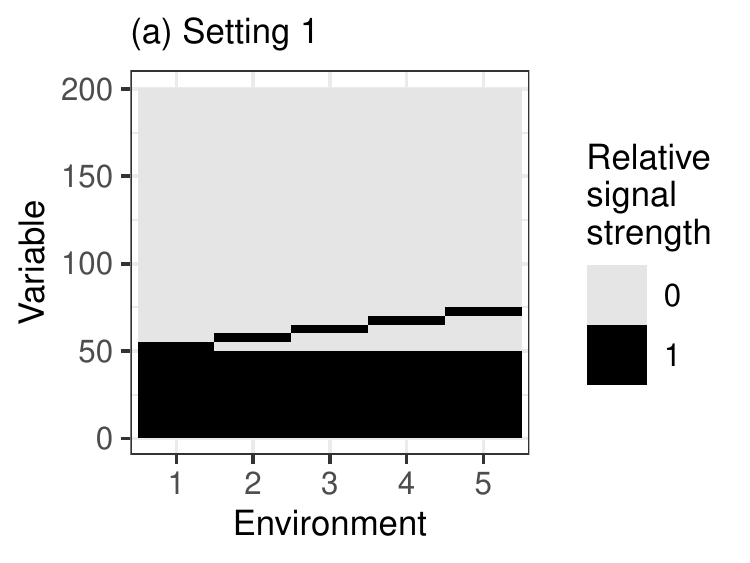}
  \end{subfigure}
  \begin{subfigure}{0.45\textwidth}
    \includegraphics[width=0.7\linewidth]{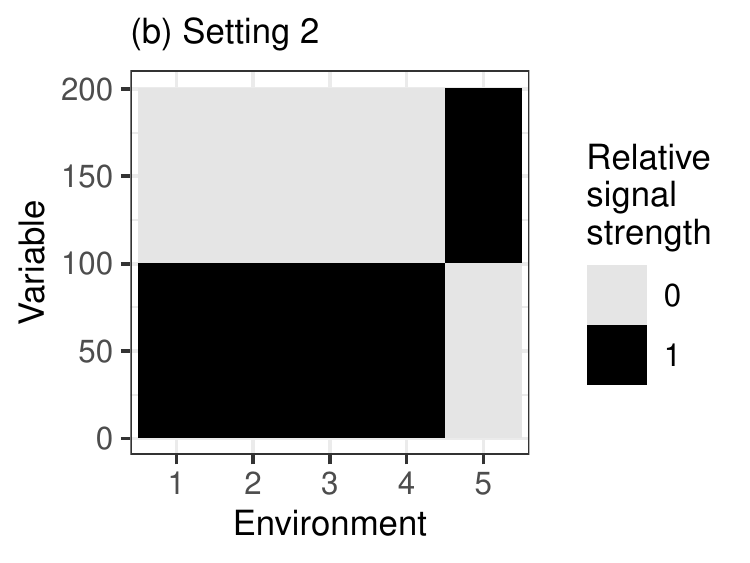}
  \end{subfigure}
  \caption{True hypothesis structure for the numerical experiments of Figure~\ref{fig:pc_sim}. Other details are as in Figure~\ref{fig:inv_beta}.}
  \label{fig:pinv_beta}
\end{figure}

\begin{figure}[!htb]
 \centering
 \includegraphics[width=\linewidth]{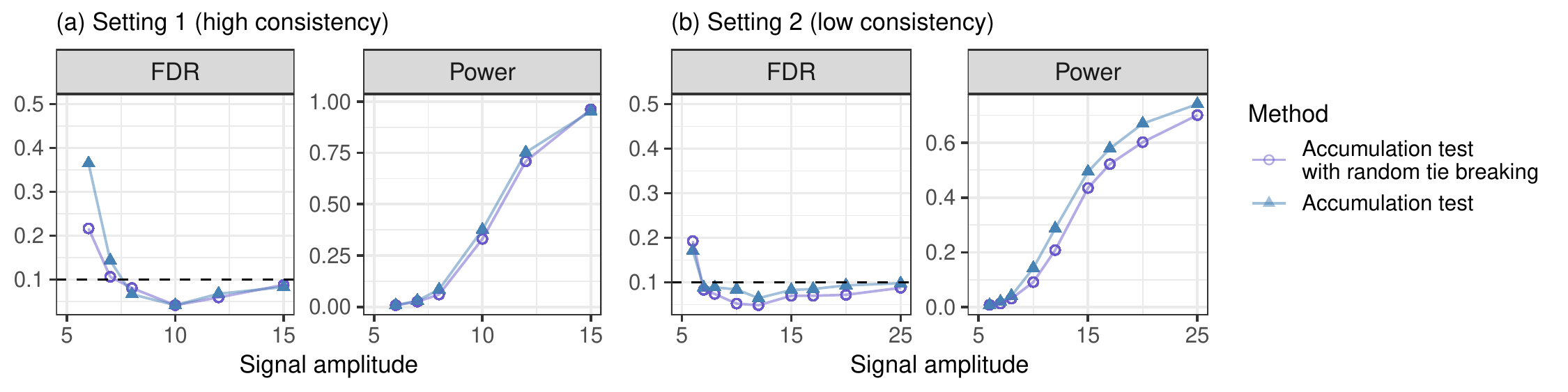}
 \caption{Performance of the accumulation test applied with two alternative implementations of our knockoff-based p-values: those in~\eqref{eq:pici-pvalues-coin} (random tie breaking for zero statistics) and~\eqref{eq:pici-pvalues} (no random tie breaking).}
 \label{fig:pc_sim_coin}
\end{figure}

\subsection{Simulated genetic study} \label{sec:app-gwas}

Synthetic genotypes from different sub-populations are generated based on the phased haplotypes in the 1000 Genomes Project as follows. First, 20 haplotype sequences are randomly picked from each of the five populations represented therein (AFR, AMR, EAS, EUR, SAS); these haplotypes will serve as ancestral motifs in a hidden Markov model similar to that of SHAPEIT~\cite{o2016haplotype}. Our hidden Markov model is constant across populations, and in each of them it describes the distribution of new haplotypes sequences as an imperfect mosaic of the corresponding 1000 Genomes motifs. More precisely, letting $H = (H_1, \ldots, H_p) \in \{0,1\}^p$ denote a sequence of haplotypes at $p$ sites, $H$ satisfies
\begin{align} \label{eq:HMM_def}
\begin{split}
\begin{cases}
  Z  \sim \text{MC}\left(Q\right), & \text{(latent Markov chain)}, \\
  H_j \mid Z \overset{\text{ind.}}{\sim} f_j(H_j \mid Z_j), & \text{(emission distribution)},
\end{cases}
\end{split}
\end{align}
where $Z = (Z_1, \ldots, Z_p)$ are latent random variables, each taking values in $\{1,\ldots,L\}$, for $L=20$.
Above, $\text{MC}\left(Q\right)$ is a Markov chain with initial probabilities $Q_1$ and transition matrices $(Q_2,\ldots,Q_p)$:
\begin{align} \label{eq:hmm}
  Q_{1}(l) & = \frac{1}{L},
  & Q_{j}(l' \mid l)
   = \begin{cases}
    \left(1-e^{-\rho d_j}\right) \frac{1}{L} + e^{-\rho d_j}, & \text{if } l' = l, \\
    \left(1-e^{-\rho d_j}\right) \frac{1}{L}, & \text{if } l' \neq l,
  \end{cases}
\end{align}
for all $j \in \{1,\ldots,p\}$ and $l \in \{1,\ldots,L\}$.
Above, $d_j$ indicates the genetic distance between loci $j$ and
$j-1$, measured in cM and provided by the 1000 Genomes Project.
In our simulations, we simply set $\rho=1$ and let the emission distribution of $H_j \mid Z_j$ be such that $H_j$ is equal
to the reference (motif) haplotype indexed by $Z_j$ with probability 0.999 (0.1\% per-site mutation rate).
Then, to obtain unphased genotypes for $n$ individuals, we sample $2n$ independent haplotype sequences from the above model and combine them pairwise by taking element-by-element sums.
Finally, we can leverage our exact knowledge of this hidden Markov model to generate knockoff copies of the typed variants with the same approach as in~\cite{sesia2020multi}; in fact, it is immediate to show any subset of haplotypes from the above model still jointly follows a hidden Markov model from the same SHAPEIT family, with suitably modified genetic distances.

\begin{figure}[!htb]
  \centering
  \includegraphics[width=0.9\linewidth]{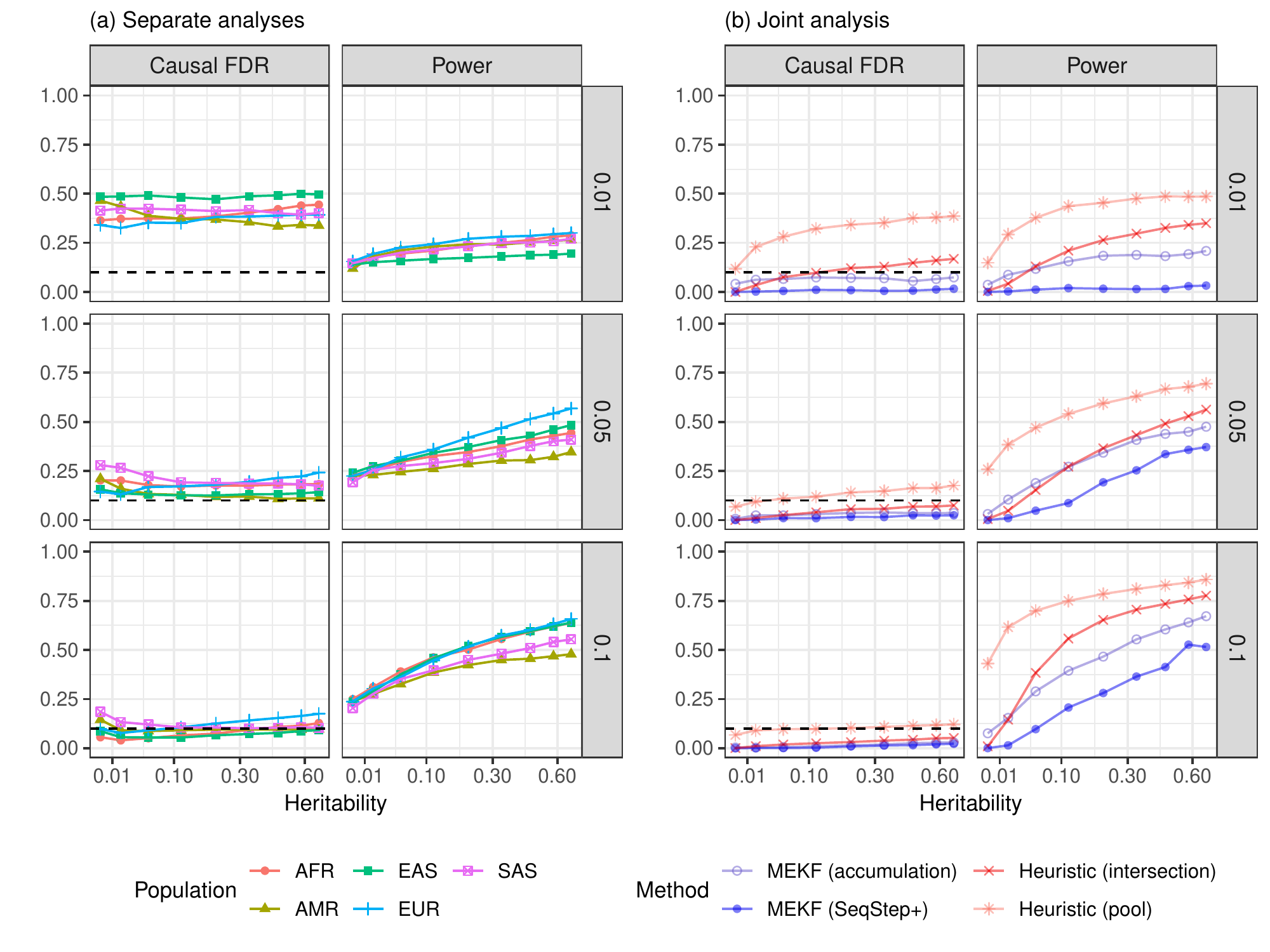}
  \caption{Analysis of a simulated multi-population genome-wide association study in which the true causal variants are missing, for different genotyping densities. Resolution: 15 kb. Other details are as in Figure~\ref{fig:sim-missing-resolution}.}
  \label{fig:sim-missing-density}
\end{figure}

\begin{figure}[!htb]
  \centering
  \includegraphics[width=0.9\linewidth]{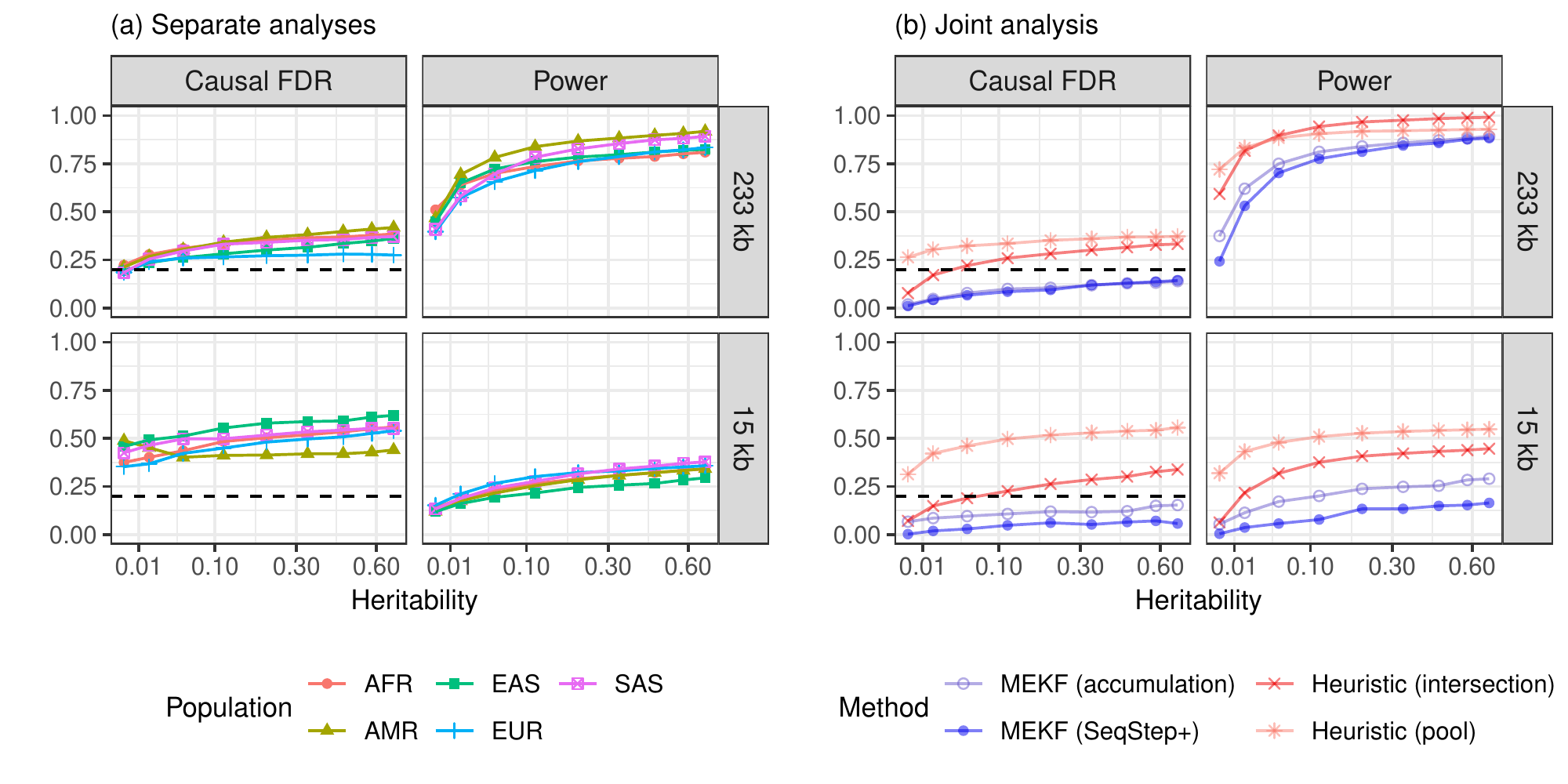}
  \caption{Analysis of a simulated multi-population genome-wide association study in which the true causal variants are missing. The nominal false discovery rate is~20\%. Other details are as in Figure~\ref{fig:sim-missing-resolution}.}
  \label{fig:sim-missing-resolution-0.2}
\end{figure}

\begin{figure}[!htb]
  \centering
  \includegraphics[width=0.65\linewidth]{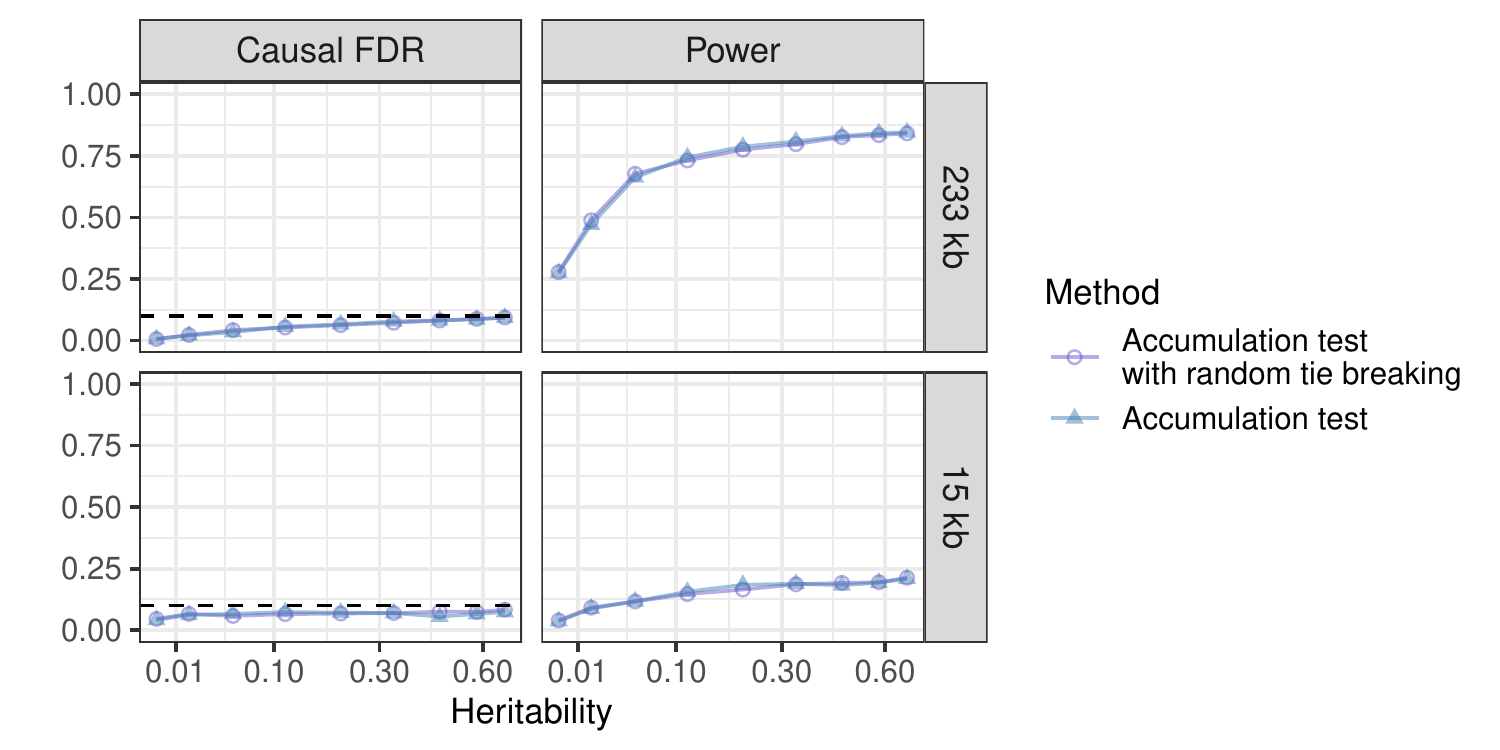}
 \caption{Performance in a simulated genome-wide association study of the accumulation test applied with two alternative implementations of our p-values: those in~\eqref{eq:pici-pvalues-coin} (random tie breaking for zero statistics) and~\eqref{eq:pici-pvalues} (no random tie breaking). The two methods are essentially equivalent here. Other details are as in Figure~\ref{fig:sim-missing-resolution} (a).}
  \label{fig:sim-missing-resolution-coin}
\end{figure}

\begin{figure}[!htb]
  \centering
  \includegraphics[width=0.65\linewidth]{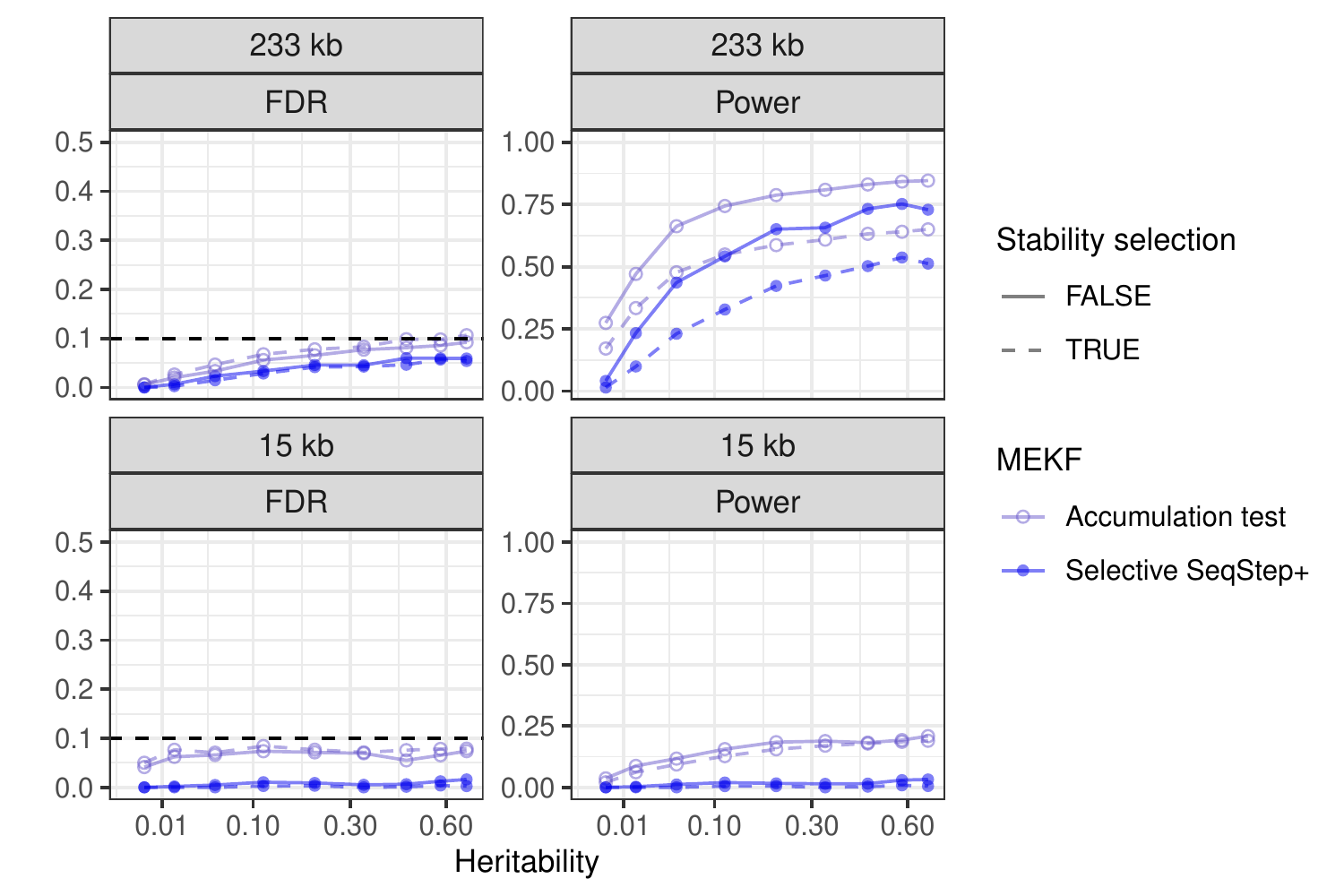}
 \caption{Performance in a simulated genome-wide association study of our methods with and without stability selection, as implemented in our analysis of the UK Biobank data of Section~\ref{sec:ukb}. Other details are as in Figure~\ref{fig:sim-missing-resolution} (a).}
  \label{fig:sim-missing-resolution-stability}
\end{figure}

\FloatBarrier

\section{Additional details about the analysis of the UK Biobank data} \label{sec:app-ukb}

\begin{table}[!htp]
\centering
\footnotesize
\caption{Definitions of the UK Biobank phenotypes used in our analysis, which match those of~\cite{sesia2020controlling}. For binary disease-status phenotypes, the number of cases refers to the subset of individuals that passed our quality control. } \label{table:pheno-def}

\begin{tabular}{c|c|>{\centering\arraybackslash}p{2.5cm}|c|>{\centering\arraybackslash}p{4cm}}
\toprule
Name & Description & Number of cases & UK Biobank Fields & UK Biobank Codes \\
\midrule
bmi & body mass index & continuous & 21001-0.0 &\\
cvd & cardiovascular disease & 148715 & 20002-0.0--20002-0.32 & 1065, 1066, 1067, 1068, 1081, 1082, 1083, 1425, 1473, 1493\\
diabetes & diabetes & 19897 & 20002-0.0--20002-0.32 & 1220 \\
height & standing height & continuous & 50-0.0 &\\
hypothyroidism & hypothyroidism & 22493 & 20002-0.0--20002-0.32 & 1226 \\
platelet & platelet count & continuous & 30080-0.0 &\\
respiratory & respiratory disease & 64945 & 20002-0.0--20002-0.32 & 1111, 1112, 1113, 1114, 1115, 1117, 1413, 1414, 1415, 1594\\
sbp & systolic blood pressure & continuous & 4080-0.0, 4080-0.1 &\\
\bottomrule
\end{tabular}
\end{table}

\begin{table}[!htb]
\footnotesize
\centering
\caption{Definitions of environments for UK Biobank individuals in terms of self-reported ancestries.} \label{tab:ukb-ancestries-env}
\begin{tabular}{ccc}
\toprule
Environment & Sample size & Self-reported ancestries \\
\midrule
African & 7,623 & ``African'', ``Caribbean'', ``Any other black background'', ``Black or Black British''\\
Asian & 3,284 & ``Asian or Asian British'', ``Chinese'', ``Any other Asian background''\\
British & 429,934 & ``British''\\
European & 28,994 & ``Any other white background'', ``Irish'', ``White''\\
Indian & 7,628 & ``Indian'', ``Pakistani'', ``Bangladeshi''\\
\bottomrule
\end{tabular}
\end{table}

\begin{table}[!htb]
\center
\small
\caption{Numbers of discoveries made with the multi-environment knockoff filter at different resolutions for several UK Biobank phenotypes. Other details are as in Table~\ref{table:ukb_discoveries}.}
\label{table:ukb_discoveries_large}

\begin{tabular}[t]{cc>{\centering\arraybackslash}m{0.7cm}>{\centering\arraybackslash}m{0.7cm}>{\centering\arraybackslash}m{0.7cm}>{\centering\arraybackslash}m{0.7cm}>{\centering\arraybackslash}m{0.7cm}}
\toprule
\multicolumn{2}{c}{ } & \multicolumn{5}{c}{Number of environments} \\
\cmidrule(l{3pt}r{3pt}){3-7}
Phenotype & Resolution (kb) & 1 & 2 & 3 & 4 & 5\\
\midrule
 & single-SNP & 0 & 0 & 0 & 0 & 0\\

 & 3 & 10 & 0 & 0 & 0 & 0\\

 & 20 & 343 & 8 & 3 & 2 & 0\\

 & 41 & 918 & 6 & 3 & 3 & 0\\

 & 81 & 1480 & 3 & 4 & 3 & 0\\

 & 208 & 2395 & 5 & 7 & 0 & 0\\

\multirow[t]{-7}{*}{\centering\arraybackslash bmi} & 425 & 2460 & 13 & 2 & 0 & 0\\
\cmidrule{1-7}
 & single-SNP & 0 & 0 & 0 & 0 & 0\\

 & 3 & 22 & 0 & 0 & 0 & 0\\

 & 20 & 239 & 8 & 0 & 0 & 0\\

 & 41 & 339 & 0 & 0 & 0 & 0\\

 & 81 & 566 & 0 & 0 & 0 & 0\\

 & 208 & 940 & 2 & 0 & 0 & 0\\

\multirow[t]{-7}{*}{\centering\arraybackslash cvd} & 425 & 1089 & 12 & 0 & 0 & 0\\
\cmidrule{1-7}
 & single-SNP & 0 & 2 & 2 & 0 & 0\\

 & 3 & 21 & 5 & 0 & 0 & 0\\

 & 20 & 61 & 6 & 0 & 2 & 0\\

 & 41 & 109 & 4 & 3 & 0 & 0\\

 & 81 & 109 & 2 & 3 & 0 & 0\\

 & 208 & 113 & 5 & 0 & 0 & 0\\

\multirow[t]{-7}{*}{\centering\arraybackslash diabetes} & 425 & 194 & 2 & 0 & 0 & 0\\
\cmidrule{1-7}
 & single-SNP & 19 & 0 & 0 & 0 & 0\\

 & 3 & 40 & 2 & 0 & 0 & 0\\

 & 20 & 105 & 5 & 0 & 0 & 0\\

 & 41 & 222 & 5 & 0 & 0 & 0\\

 & 81 & 277 & 7 & 0 & 0 & 0\\

 & 208 & 295 & 11 & 0 & 0 & 0\\

\multirow[t]{-7}{*}{\centering\arraybackslash hypothyroidism} & 425 & 335 & 10 & 0 & 0 & 0\\
\cmidrule{1-7}
 & single-SNP & 0 & 0 & 0 & 0 & 0\\

 & 3 & 0 & 0 & 0 & 0 & 0\\

 & 20 & 83 & 4 & 0 & 0 & 0\\

 & 41 & 123 & 2 & 0 & 0 & 0\\

 & 81 & 193 & 15 & 0 & 0 & 0\\

 & 208 & 262 & 0 & 0 & 0 & 0\\

\multirow[t]{-7}{*}{\centering\arraybackslash respiratory} & 425 & 383 & 5 & 0 & 0 & 0\\
\cmidrule{1-7}
 & single-SNP & 0 & 0 & 0 & 0 & 0\\

 & 3 & 83 & 0 & 0 & 0 & 0\\

 & 20 & 191 & 2 & 0 & 0 & 0\\

 & 41 & 511 & 3 & 0 & 0 & 0\\

 & 81 & 830 & 6 & 0 & 0 & 0\\

 & 208 & 1183 & 4 & 0 & 0 & 0\\

\multirow[t]{-7}{*}{\centering\arraybackslash sbp} & 425 & 1543 & 14 & 0 & 0 & 0\\
\bottomrule
\end{tabular}

\end{table}

\begin{figure}[!htb]
 \centering
 \includegraphics[width=\linewidth]{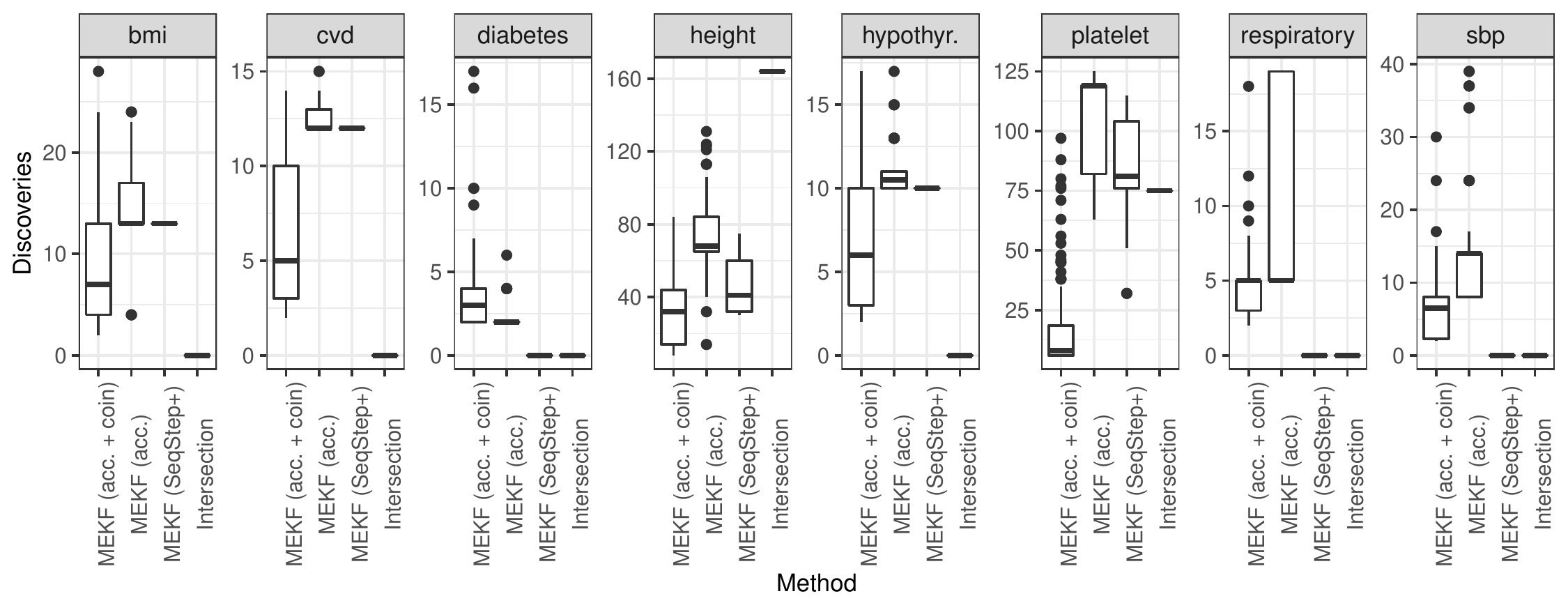}
 \caption{Numbers of low-resolution discoveries consistent across at least two populations, as obtained with different methods. For methods involving randomized p-values (multi-environment knockoff filter with the accumulation test or the selective SeqStep+), this figure shows the empirical distribution of the numbers of discoveries over 100 independent realizations of the p-values conditional on the knockoff test statistics. Other details are as in Table~\ref{table:ukb_discoveries}.}
 \label{fig:analysis_numdisc}
\end{figure}

\begin{figure}[!htb]
\centering
\includegraphics[width = 0.686\textwidth]{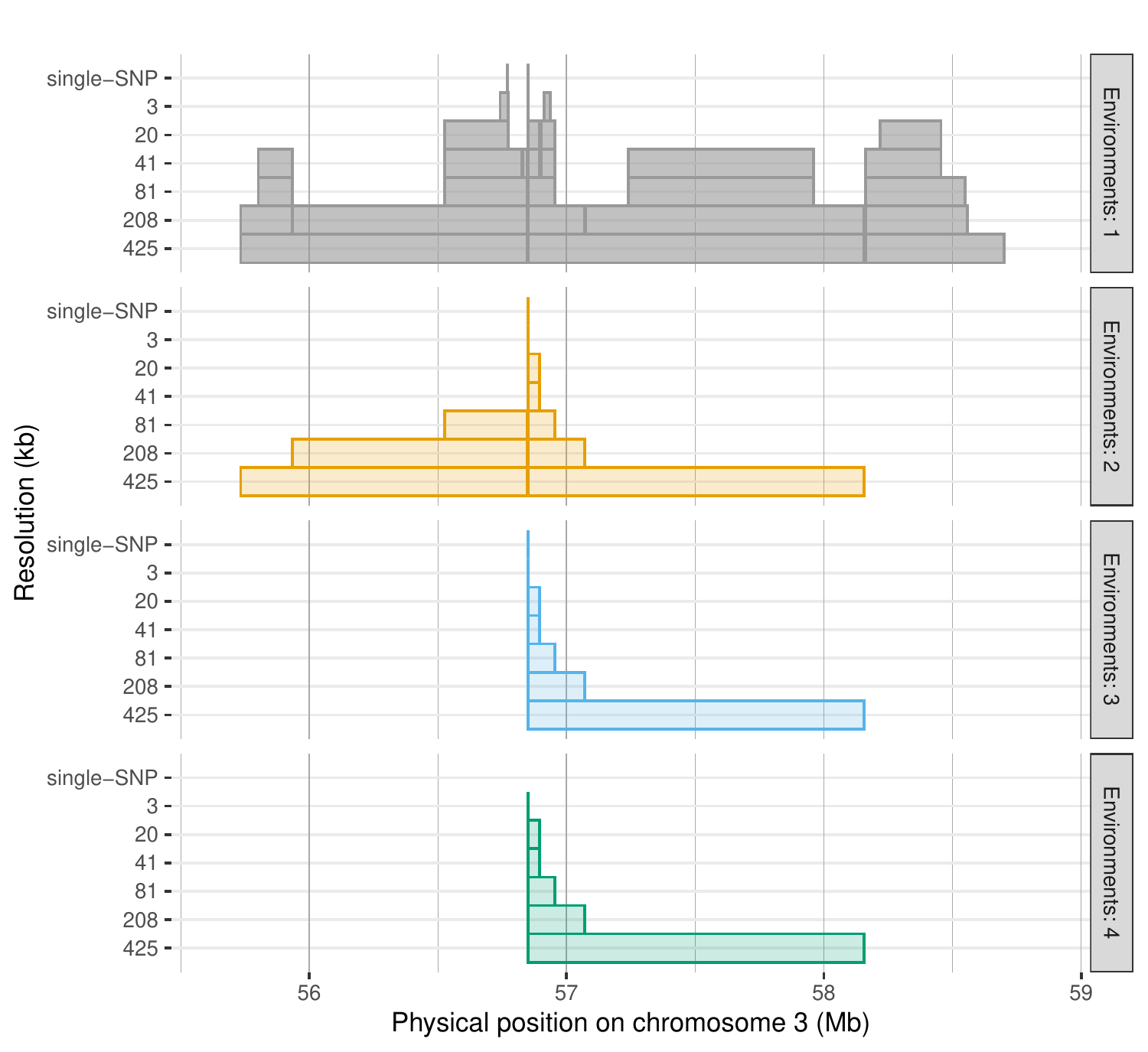}
\caption{Chicago plot of some discoveries on chromosome three for platelet count based on UK Biobank data from individuals from five environments, as in Figure~\ref{fig:chicago_biobank}. Here, the findings are shown separately based on the numbers of environments across which they are consistent.}
\label{fig:chicago_biobank_stacked}
\end{figure}

\begin{table}[!htb]
\center
\footnotesize
\caption{Numbers of consistent discoveries for several UK Biobank phenotypes obtained using three alternative methods: the multi-environment knockoff filter implemented with the accumulation test (Acc.) or the selective SeqStep+ (SStep), and the intersection heuristic (Int.). Other details are as in Table~\ref{table:ukb_discoveries}.}
\label{table:ukb_discoveries_methods}

\begin{tabular}[t]{cc>{\centering\arraybackslash}m{0.5cm}>{\centering\arraybackslash}m{0.5cm}>{\centering\arraybackslash}m{0.5cm}>{\centering\arraybackslash}m{0.5cm}>{\centering\arraybackslash}m{0.5cm}ccccccc}
\toprule
\multicolumn{2}{c}{ } & \multicolumn{12}{c}{Number of environments / Consistent testing method} \\
\cmidrule(l{3pt}r{3pt}){3-14}
\multicolumn{2}{c}{ } & \multicolumn{3}{c}{2} & \multicolumn{3}{c}{3} & \multicolumn{3}{c}{4} & \multicolumn{3}{c}{5} \\
\cmidrule(l{3pt}r{3pt}){3-5} \cmidrule(l{3pt}r{3pt}){6-8} \cmidrule(l{3pt}r{3pt}){9-11} \cmidrule(l{3pt}r{3pt}){12-14}
Phenotype & Resolution (kb) & Acc. & SStep & Int. & Acc. & SStep & Int. & Acc. & SStep & Int. & Acc. & SStep & Int.\\
\midrule
 & single-SNP & 0 & 0 & 0 & 0 & 0 & 0 & 0 & 0 & 0 & 0 & 0 & 0\\

 & 3 & 0 & 0 & 0 & 0 & 0 & 0 & 0 & 0 & 0 & 0 & 0 & 0\\

 & 20 & 8 & 0 & 0 & 3 & 0 & 0 & 2 & 0 & 0 & 0 & 0 & 0\\

 & 41 & 6 & 0 & 0 & 3 & 0 & 0 & 3 & 0 & 0 & 0 & 0 & 0\\

 & 81 & 3 & 0 & 0 & 4 & 0 & 0 & 3 & 0 & 0 & 0 & 0 & 0\\

 & 208 & 5 & 0 & 0 & 7 & 0 & 0 & 0 & 0 & 0 & 0 & 0 & 0\\

\multirow[t]{-7}{*}{\centering\arraybackslash bmi} & 425 & 13 & 0 & 0 & 2 & 0 & 0 & 0 & 0 & 0 & 0 & 0 & 0\\
\cmidrule{1-14}
 & single-SNP & 0 & 0 & 0 & 0 & 0 & 0 & 0 & 0 & 0 & 0 & 0 & 0\\

 & 3 & 0 & 0 & 0 & 0 & 0 & 0 & 0 & 0 & 0 & 0 & 0 & 0\\

 & 20 & 8 & 0 & 0 & 0 & 0 & 0 & 0 & 0 & 0 & 0 & 0 & 0\\

 & 41 & 0 & 0 & 0 & 0 & 0 & 0 & 0 & 0 & 0 & 0 & 0 & 0\\

 & 81 & 0 & 0 & 0 & 0 & 0 & 0 & 0 & 0 & 0 & 0 & 0 & 0\\

 & 208 & 2 & 0 & 0 & 0 & 0 & 0 & 0 & 0 & 0 & 0 & 0 & 0\\

\multirow[t]{-7}{*}{\centering\arraybackslash cvd} & 425 & 12 & 12 & 0 & 0 & 0 & 0 & 0 & 0 & 0 & 0 & 0 & 0\\
\cmidrule{1-14}
 & single-SNP & 2 & 0 & 0 & 2 & 0 & 0 & 0 & 0 & 0 & 0 & 0 & 0\\

 & 3 & 5 & 0 & 0 & 0 & 0 & 0 & 0 & 0 & 0 & 0 & 0 & 0\\

 & 20 & 6 & 0 & 0 & 0 & 0 & 0 & 2 & 0 & 0 & 0 & 0 & 0\\

 & 41 & 4 & 0 & 0 & 3 & 0 & 0 & 0 & 0 & 0 & 0 & 0 & 0\\

 & 81 & 2 & 0 & 0 & 3 & 0 & 0 & 0 & 0 & 0 & 0 & 0 & 0\\

 & 208 & 5 & 0 & 0 & 0 & 0 & 0 & 0 & 0 & 0 & 0 & 0 & 0\\

\multirow[t]{-7}{*}{\centering\arraybackslash diabetes} & 425 & 2 & 0 & 0 & 0 & 0 & 0 & 0 & 0 & 0 & 0 & 0 & 0\\
\cmidrule{1-14}
 & single-SNP & 13 & 0 & 0 & 2 & 0 & 0 & 0 & 0 & 0 & 0 & 0 & 0\\

 & 3 & 9 & 0 & 0 & 6 & 0 & 0 & 0 & 0 & 0 & 0 & 0 & 0\\

 & 20 & 33 & 29 & 25 & 0 & 0 & 0 & 0 & 0 & 0 & 0 & 0 & 0\\

 & 41 & 42 & 20 & 68 & 7 & 0 & 0 & 7 & 0 & 0 & 2 & 0 & 0\\

 & 81 & 48 & 49 & 84 & 24 & 23 & 0 & 0 & 0 & 0 & 0 & 0 & 0\\

 & 208 & 103 & 33 & 107 & 23 & 16 & 0 & 7 & 0 & 0 & 3 & 0 & 0\\

\multirow[t]{-7}{*}{\centering\arraybackslash height} & 425 & 68 & 31 & 164 & 26 & 15 & 0 & 3 & 0 & 0 & 0 & 0 & 0\\
\cmidrule{1-14}
 & single-SNP & 0 & 0 & 0 & 0 & 0 & 0 & 0 & 0 & 0 & 0 & 0 & 0\\

 & 3 & 2 & 0 & 0 & 0 & 0 & 0 & 0 & 0 & 0 & 0 & 0 & 0\\

 & 20 & 5 & 0 & 0 & 0 & 0 & 0 & 0 & 0 & 0 & 0 & 0 & 0\\

 & 41 & 5 & 0 & 0 & 0 & 0 & 0 & 0 & 0 & 0 & 0 & 0 & 0\\

 & 81 & 7 & 0 & 0 & 0 & 0 & 0 & 0 & 0 & 0 & 0 & 0 & 0\\

 & 208 & 11 & 11 & 0 & 0 & 0 & 0 & 0 & 0 & 0 & 0 & 0 & 0\\

\multirow[t]{-7}{*}{\centering\arraybackslash hypothyroidism} & 425 & 10 & 10 & 0 & 0 & 0 & 0 & 0 & 0 & 0 & 0 & 0 & 0\\
\cmidrule{1-14}
 & single-SNP & 9 & 0 & 0 & 3 & 0 & 0 & 0 & 0 & 0 & 0 & 0 & 0\\

 & 3 & 10 & 10 & 0 & 4 & 0 & 0 & 4 & 0 & 0 & 0 & 0 & 0\\

 & 20 & 27 & 20 & 26 & 16 & 0 & 0 & 2 & 0 & 0 & 0 & 0 & 0\\

 & 41 & 52 & 0 & 50 & 12 & 12 & 0 & 9 & 0 & 0 & 0 & 0 & 0\\

 & 81 & 104 & 69 & 69 & 15 & 15 & 3 & 8 & 0 & 0 & 0 & 0 & 0\\

 & 208 & 98 & 58 & 58 & 16 & 14 & 8 & 14 & 0 & 0 & 2 & 0 & 0\\

\multirow[t]{-7}{*}{\centering\arraybackslash platelet} & 425 & 119 & 70 & 75 & 9 & 0 & 6 & 11 & 0 & 0 & 0 & 0 & 0\\
\cmidrule{1-14}
 & single-SNP & 0 & 0 & 0 & 0 & 0 & 0 & 0 & 0 & 0 & 0 & 0 & 0\\

 & 3 & 0 & 0 & 0 & 0 & 0 & 0 & 0 & 0 & 0 & 0 & 0 & 0\\

 & 20 & 4 & 0 & 0 & 0 & 0 & 0 & 0 & 0 & 0 & 0 & 0 & 0\\

 & 41 & 2 & 0 & 0 & 0 & 0 & 0 & 0 & 0 & 0 & 0 & 0 & 0\\

 & 81 & 15 & 0 & 0 & 0 & 0 & 0 & 0 & 0 & 0 & 0 & 0 & 0\\

 & 208 & 0 & 0 & 0 & 0 & 0 & 0 & 0 & 0 & 0 & 0 & 0 & 0\\

\multirow[t]{-7}{*}{\centering\arraybackslash respiratory} & 425 & 5 & 0 & 0 & 0 & 0 & 0 & 0 & 0 & 0 & 0 & 0 & 0\\
\cmidrule{1-14}
 & single-SNP & 0 & 0 & 0 & 0 & 0 & 0 & 0 & 0 & 0 & 0 & 0 & 0\\

 & 3 & 0 & 0 & 0 & 0 & 0 & 0 & 0 & 0 & 0 & 0 & 0 & 0\\

 & 20 & 2 & 0 & 0 & 0 & 0 & 0 & 0 & 0 & 0 & 0 & 0 & 0\\

 & 41 & 3 & 0 & 0 & 0 & 0 & 0 & 0 & 0 & 0 & 0 & 0 & 0\\

 & 81 & 6 & 0 & 0 & 0 & 0 & 0 & 0 & 0 & 0 & 0 & 0 & 0\\

 & 208 & 4 & 0 & 0 & 0 & 0 & 0 & 0 & 0 & 0 & 0 & 0 & 0\\

\multirow[t]{-7}{*}{\centering\arraybackslash sbp} & 425 & 14 & 0 & 0 & 0 & 0 & 0 & 0 & 0 & 0 & 0 & 0 & 0\\
\bottomrule
\end{tabular}

\end{table}

\begin{table}[!htb]
\center
\small
\caption{Proportions of discoveries confirmed by previously known associations in the GWAS Catalog, for several UK Biobank phenotypes. Other details are as in Table~\ref{table:ukb_discoveries_confirm_proportion}.}
\label{table:ukb_discoveries_confirm_proportion_others}

\begin{tabular}[t]{cccccc}
\toprule
Phenotype & Resolution (kb) & MEKF & Pooling & Binomial p-value & Intersection\\
\midrule
 & single-SNP & 0 / 0 & 0 / 0 &  & 0 / 0\\

 & 3 & 0 / 0 & 10 / 10 (100\%) &  & 0 / 0\\

 & 20 & 8 / 8 (100\%) & 308 / 343 (90\%) & $7.22 \cdot 10^{-01}$ & 0 / 0\\

 & 41 & 6 / 6 (100\%) & 658 / 918 (72\%) & $2.79 \cdot 10^{-01}$ & 0 / 0\\

 & 81 & 3 / 3 (100\%) & 875 / 1480 (59\%) & $3.95 \cdot 10^{-01}$ & 0 / 0\\

 & 208 & 5 / 5 (100\%) & 1113 / 2395 (46\%) & $5.14 \cdot 10^{-02}$ & 0 / 0\\

\multirow[t]{-7}{*}{\centering\arraybackslash bmi} & 425 & 13 / 13 (100\%) & 1146 / 2460 (47\%) & $3.56 \cdot 10^{-04}$ & 0 / 0\\
\cmidrule{1-6}
 & single-SNP & 0 / 0 & 0 / 0 &  & 0 / 0\\

 & 3 & 0 / 0 & 69 / 83 (83\%) &  & 0 / 0\\

 & 20 & 2 / 2 (100\%) & 166 / 191 (87\%) & 1 & 0 / 0\\

 & 41 & 3 / 3 (100\%) & 358 / 511 (70\%) & $6.19 \cdot 10^{-01}$ & 0 / 0\\

 & 81 & 6 / 6 (100\%) & 518 / 830 (62\%) & $1.41 \cdot 10^{-01}$ & 0 / 0\\

 & 208 & 4 / 4 (100\%) & 645 / 1183 (55\%) & $1.87 \cdot 10^{-01}$ & 0 / 0\\

\multirow[t]{-7}{*}{\centering\arraybackslash sbp} & 425 & 14 / 14 (100\%) & 714 / 1543 (46\%) & $1.83 \cdot 10^{-04}$ & 0 / 0\\
\cmidrule{1-6}
 & single-SNP & 0 / 0 & 0 / 0 &  & 0 / 0\\

 & 3 & 0 / 0 & 21 / 22 (95\%) &  & 0 / 0\\

 & 20 & 8 / 8 (100\%) & 178 / 239 (74\%) & $2.19 \cdot 10^{-01}$ & 0 / 0\\

 & 41 & 0 / 0 & 248 / 339 (73\%) &  & 0 / 0\\

 & 81 & 0 / 0 & 365 / 566 (64\%) &  & 0 / 0\\

 & 208 & 2 / 2 (100\%) & 559 / 940 (59\%) & $6.56 \cdot 10^{-01}$ & 0 / 0\\

\multirow[t]{-7}{*}{\centering\arraybackslash cvd} & 425 & 11 / 12 (92\%) & 717 / 1089 (66\%) & $1.16 \cdot 10^{-01}$ & 0 / 0\\
\cmidrule{1-6}
 & single-SNP & 0 / 0 & 0 / 0 &  & 0 / 0\\

 & 3 & 0 / 0 & 0 / 0 &  & 0 / 0\\

 & 20 & 4 / 4 (100\%) & 74 / 83 (89\%) & 1 & 0 / 0\\

 & 41 & 2 / 2 (100\%) & 110 / 123 (89\%) & 1 & 0 / 0\\

 & 81 & 14 / 15 (93\%) & 156 / 193 (81\%) & $3.90 \cdot 10^{-01}$ & 0 / 0\\

 & 208 & 0 / 0 & 196 / 262 (75\%) &  & 0 / 0\\

\multirow[t]{-7}{*}{\centering\arraybackslash respiratory} & 425 & 5 / 5 (100\%) & 267 / 383 (70\%) & $3.28 \cdot 10^{-01}$ & 0 / 0\\
\cmidrule{1-6}
 & single-SNP & 0 / 0 & 7 / 19 (37\%) &  & 0 / 0\\

 & 3 & 2 / 2 (100\%) & 23 / 40 (57\%) & $6.48 \cdot 10^{-01}$ & 0 / 0\\

 & 20 & 5 / 5 (100\%) & 71 / 105 (68\%) & $3.00 \cdot 10^{-01}$ & 0 / 0\\

 & 41 & 5 / 5 (100\%) & 101 / 222 (45\%) & $4.97 \cdot 10^{-02}$ & 0 / 0\\

 & 81 & 7 / 7 (100\%) & 126 / 277 (45\%) & $1.35 \cdot 10^{-02}$ & 0 / 0\\

 & 208 & 8 / 11 (73\%) & 141 / 295 (48\%) & $1.88 \cdot 10^{-01}$ & 0 / 0\\

\multirow[t]{-7}{*}{\centering\arraybackslash hypothyroidism} & 425 & 8 / 10 (80\%) & 140 / 335 (42\%) & $3.74 \cdot 10^{-02}$ & 0 / 0\\
\cmidrule{1-6}
 & single-SNP & 2 / 2 (100\%) & 0 / 0 &  & 0 / 0\\

 & 3 & 3 / 5 (60\%) & 20 / 21 (95\%) & $1.51 \cdot 10^{-01}$ & 0 / 0\\

 & 20 & 5 / 6 (83\%) & 54 / 61 (89\%) & 1 & 0 / 0\\

 & 41 & 4 / 4 (100\%) & 94 / 109 (86\%) & $9.63 \cdot 10^{-01}$ & 0 / 0\\

 & 81 & 2 / 2 (100\%) & 95 / 109 (87\%) & 1 & 0 / 0\\

 & 208 & 5 / 5 (100\%) & 101 / 113 (89\%) & $9.90 \cdot 10^{-01}$ & 0 / 0\\

\multirow[t]{-7}{*}{\centering\arraybackslash diabetes} & 425 & 2 / 2 (100\%) & 157 / 194 (81\%) & 1 & 0 / 0\\
\bottomrule
\end{tabular}

\end{table}

\begin{table}[!htb]
\center
\footnotesize
\caption{Consistent discoveries at the single-nuclotide resolution for UK Biobank phenotypes, compared to associations previously reported in the GWAS Catalog. Asterisks indicate associations previously reported in the GWAS Catalog for the gene corresponding to our variant, rather than for the variant itself.}
\label{table:ukb_discoveries_confirm}

\begin{tabular}[t]{cccccccc}
\toprule
Phenotype & Chr. & Position (Mb) & SNP & Association & Env. & Consequence & Gene\\
\midrule
 & 10 & 114.758 & rs7903146 & diabetes & 3 & intron & TCF7L2\\

\multirow[t]{-2}{*}{\centering\arraybackslash diabetes} & 11 & 92.709 & rs10830963 & diabetes & 2 & intron & MTNR1B\\
\cmidrule{1-8}
 & 2 & 56.097 & rs3791679 & height & 2 & intron & EFEMP1\\

 & 3 & 141.126 & rs1344672 & $~$height$^*$ & 2 & intron & ZBTB38\\

 & 4 & 18.025 & rs2011603 & $~$height$^*$ & 3 & 2KB upstream & LCORL\\

 & 6 & 7.720 & rs12198986 & height & 2 & regulatory region & \\

 & 6 & 19.839 & rs41271299 & height & 2 & intron & ID4\\

 & 8 & 130.726 & rs6470764 & height & 2 & intron & GSDMC\\

 & 11 & 75.276 & rs606452 & height & 2 & intron & SERPINH1\\

 & 12 & 66.360 & rs8756 & height & 2 & 3 prime UTR & HMGA2\\

 & 12 & 93.979 & rs11107116 & height & 2 & intron & SOCS2\\

 & 15 & 99.195 & rs2871865 & height & 3 & intron & IGF1R\\

 & 15 & 100.693 & rs72755233 & height & 2 & missense & ADAMTS17\\

 & 19 & 55.880 & rs4252548 & height & 2 & missense & IL11\\

\multirow[t]{-13}{*}{\centering\arraybackslash height} & 20 & 34.026 & rs143384 & height & 2 & 3 prime UTR & GDF5\\
\cmidrule{1-8}
 & 3 & 56.850 & rs1354034 & platelet & 3 & intron & ARHGEF3\\

 & 5 & 75.997 & Affx-26978473 &  & 2 &  & \\

 & 6 & 135.419 & rs7775698 & platelet & 2 & intron & HBS1L\\

 & 9 & 4.763 & rs385893 & platelet & 3 & regulatory region & AL353151.2, ECM1P1\\

 & 10 & 65.028 & rs10761731 & platelet & 2 & intron & JMJD1C\\

 & 12 & 111.885 & rs3184504 & platelet & 2 & missense & SH2B3\\

 & 12 & 111.885 & rs72650673 & hematocrit & 2 & missense & SH2B3\\

 & 18 & 20.721 & rs11082304 & platelet & 2 & intron & CABLES1\\

\multirow[t]{-9}{*}{\centering\arraybackslash platelet} & 19 & 16.186 & rs8109288 & platelet & 3 & n.c.~transcript exon & AC008894.3, TPM4\\
\bottomrule
\end{tabular}

%%% Local Variables:
%%% mode: latex
%%% TeX-master: "../invariant"
%%% End:

\end{table}

\FloatBarrier

\end{document}